\documentclass[11pt]{article}
\usepackage[utf8]{inputenc}
\usepackage[margin=1in]{geometry}
\usepackage[T1]{fontenc} %
\usepackage{graphicx}
\usepackage{mathpazo}
\usepackage{hyperref}
\hypersetup{
  colorlinks = true,
  urlcolor = {blueGrotto},
  linkcolor = {royalBlue}, 
  citecolor = {navyBlue} 
}

\usepackage{setspace}
\usepackage{xcolor}
\definecolor{niceRed}{RGB}{190,38,38}
\definecolor{niceYellow}{HTML}{f5b400}
\definecolor{blueGrotto}{HTML}{059DC0}
\definecolor{royalBlue}{HTML}{057DCD}
\definecolor{navyBlue}{HTML}{0B579C}
\definecolor{limeGreen}{HTML}{81B622}
\definecolor{nicePurple}{HTML}{9c27b0}
\definecolor{lightRoyalBlue}{HTML}{def2ff}  
\definecolor{ivory}{HTML}{FFFFF0}

\definecolor{lightRoyalBlue}{HTML}{def2ff} %

\usepackage[backref=true,style=alphabetic,natbib=true,maxbibnames=999,maxalphanames=4,sortcites=true]{biblatex}
\AtBeginRefsection{\GenRefcontextData{sorting=ynt}}
\AtEveryCite{\localrefcontext[sorting=ynt]}

\usepackage{microtype}
\usepackage{subfigure}
\usepackage{booktabs} %
\PassOptionsToPackage{hyphens}{url}
 
\usepackage{hyperref}
\usepackage{algorithm2e}
\usepackage{amsmath}
\usepackage{amssymb}
\usepackage{mathtools}
\usepackage{amsthm}
\usepackage{bbm}
\usepackage[capitalise,noabbrev,nameinlink]{cleveref} %
\usepackage{nicefrac}
\usepackage{thm-restate}
\usepackage{tikz}
\usepackage[framemethod=TikZ]{mdframed}
\mdfsetup{%
backgroundcolor=white, 
linewidth=1,
roundcorner=4pt}

\usepackage{mathrsfs}
\usepackage{mathtools} 
\usepackage{dsfont}
\usepackage{color}
\usepackage{color,fancybox,graphicx,multicol} 
\usepackage{enumitem} 
\usepackage{color-edits}
\addauthor[Manolis]{mz}{teal}
\addauthor[Anay]{am}{blue}
\addauthor[Alkis]{ak}{limeGreen}

\usepackage{tcolorbox}
\tcbuselibrary{skins}
\usepackage{multirow}

\usetikzlibrary{shapes.geometric} %

\theoremstyle{plain}
\newtheorem{theorem}{Theorem}[section]
\newtheorem{proposition}[theorem]{Proposition}
\newtheorem{lemma}[theorem]{Lemma}
\newtheorem{fact}[theorem]{Fact}
\newtheorem{corollary}[theorem]{Corollary}

\theoremstyle{definition}
\newtheorem{definition}{Definition}
\newtheorem*{definition*}{Definition}
\newtheorem{assumption}{Assumption}
\theoremstyle{remark}
\newtheorem{remark}[theorem]{Remark}

\theoremstyle{plain}
\newtheorem{inftheorem}[theorem]{Informal Theorem}

\AfterEndEnvironment{definition}{\noindent\ignorespaces}
\AfterEndEnvironment{assumption}{\noindent\ignorespaces} %
\AfterEndEnvironment{lemma}{\noindent\ignorespaces} %
\AfterEndEnvironment{theorem}{\noindent\ignorespaces} %
\AfterEndEnvironment{inftheorem}{\noindent\ignorespaces} %
\AfterEndEnvironment{proposition}{\noindent\ignorespaces} %
\AfterEndEnvironment{fact}{\noindent\ignorespaces}
\AfterEndEnvironment{question}{\noindent\ignorespaces}
\AfterEndEnvironment{corollary}{\noindent\ignorespaces}
\AfterEndEnvironment{model}{\noindent\ignorespaces}
\AfterEndEnvironment{remark}{\noindent\ignorespaces}
\AfterEndEnvironment{proof}{\noindent\ignorespaces}
\AfterEndEnvironment{fact}{\noindent\ignorespaces}
\AfterEndEnvironment{restatable}{\noindent\ignorespaces} %

\usepackage{etoolbox}

\crefname{section}{Section}{Sections}
\crefname{theorem}{Theorem}{Theorems}
\crefname{assumption}{Assumption}{Assumptions}
\crefname{lemma}{Lemma}{Lemmas}
\crefname{definition}{Definition}{Definitions}
\crefname{conjecture}{Conjecture}{Conjectures}
\crefname{corollary}{Corollary}{Corollaries}
\crefname{construction}{Construction}{Constructions}
\crefname{conjecture}{Conjecture}{Conjectures}
\crefname{claim}{Claim}{Claims}
\crefname{observation}{Observation}{Observations}
\crefname{proposition}{Proposition}{Propositions}
\crefname{fact}{Fact}{Facts}
\crefname{question}{Question}{Questions}
\crefname{problem}{Problem}{Problems}
\crefname{remark}{Remark}{Remarks}
\crefname{example}{Example}{Examples}
\crefname{equation}{Equation}{Equations}
\crefname{appendix}{Appendix}{Appendices}
\crefname{algorithm}{Algorithm}{Algorithms}
\crefname{table}{Table}{Tables}
\crefname{model}{Model}{Models}
\crefname{figure}{Figure}{Figures}

\crefname{inftheorem}{Informal Theorem}{Informal Theorems}
\crefname{minftheorem}{Main Informal Theorem}{Main Informal Theorems}
\crefname{maintheorem}{Main Theorem}{Main Theorems}

\newcommand{\eat}[1]{}
\newcommand{\yesnum}{\addtocounter{equation}{1}\tag{\theequation}}

\makeatletter
\newcommand{\tagnum}[2]{%
    \refstepcounter{equation}%
    \tag{#1) \ (\theequation}%
    \protected@write \@auxout {}{%
        \string \newlabel {#2}{{\theequation}{\thepage}{}{equation.\theequation}{}}%
    }%
}
\makeatother

\newcommand{\white}[1]{\textcolor{white}{#1}}

\def\abs#1{\left| #1 \right|}
\def\sabs#1{| #1 |}

\newcommand{\sinparen}[1]{(#1)}
\newcommand{\sinbrace}[1]{\{#1\}}
\newcommand{\sinsquare}[1]{[#1]}
\newcommand{\inbrace}[1]{\left\{#1\right\}}

\newcommand{\inparen}[1]{\left(#1\right)}
\newcommand{\insquare}[1]{\left[#1\right]}
\newcommand{\inangle}[1]{\left\langle#1\right\rangle}

\newcommand{\norm}[1]{\ensuremath{\left\lVert #1 \right\rVert}}

\newcommand{\N}{\mathbb{N}}
\newcommand{\R}{\mathbb{R}}
\newcommand{\Q}{\mathbb{Q}}

\newcommand{\zo}{\ensuremath{\inbrace{0, 1}}}
\DeclareMathOperator{\Ex}{\mathbb{E}}
\DeclareMathOperator*{\OldPr}{Pr} %
\let\Pr\relax %
\DeclareMathOperator{\Pr}{\OldPr\nolimits}

\DeclareMathOperator{\cov}{Cov}
\DeclareMathOperator{\var}{Var}
\newcommand{\argmin}{\operatornamewithlimits{arg\,min}}

\newcommand{\nfrac}[2]{\nicefrac{#1}{#2}}
\newcommand{\sfrac}[2]{{#1/#2}}  %
\newcommand{\myfrac}[1]{\sfrac{1}{#1}}
 
\newcommand{\st}{\mathrm{s.t.}}
\newcommand{\poly}{\mathrm{poly}}
\newcommand{\eps}{\varepsilon}
\renewcommand{\epsilon}{\varepsilon}

\newcommand*\circled[1]{\tikz[baseline=(char.base)]{
            \node[shape=circle,draw,inner sep=1pt] (char) {#1};}}

\newcommand{\wh}[1]{\widehat{#1}}

\newcommand{\wt}[1]{\widetilde{#1}}

\newcommand{\cB}{\mathcal{B}}

\newcommand{\cD}{\mathcal{D}}

\newcommand{\cH}{\mathcal{H}}
\newcommand{\cI}{\mathcal{I}} 

\newcommand{\cN}{\mathcal{N}}
\newcommand{\cP}{\mathcal{P}}

\newcommand{\cX}{\mathcal{X}}

\DeclareMathAlphabet{\mathdutchcal}{U}{dutchcal}{m}{n}
\SetMathAlphabet{\mathdutchcal}{bold}{U}{dutchcal}{b}{n}
\DeclareMathAlphabet{\mathdutchbcal}{U}{dutchcal}{b}{n}

\newcommand{\hypo}[1]{\mathdutchcal{#1}}
\newcommand{\hyH}{\hypo{H}}
\newcommand{\hyG}{\hypo{G}}

\newcommand{\cS}{\hypo{S}}

\newcommand{\toa}[1]{\textit{#1}}

\newcommand{\quadtext}[1]{\quad\text{#1}\quad}
\newcommand{\qquadtext}[1]{\qquad\text{#1}\qquad}

\newcommand{\quadand}{\quadtext{and}}
\newcommand{\qquadand}{\qquadtext{and}}

\newcommand{\neyman}{\ensuremath{\tau_{\rm Neyman}}}

\newcommand{\ipw}{\ensuremath{\tau_\mathrm{IPW}}}
\newcommand{\doublyrobust}{\ensuremath{\tau_\mathrm{DR}}}

\newcommand{\CIPW}{\text{CIPW}}
\DeclareMathOperator{\MSE}{MSE}
\DeclareMathOperator{\RMSE}{RMSE}
\DeclareMathOperator{\bias}{Bias}
\DeclareMathOperator{\variance}{Var}
\DeclareMathOperator{\diam}{diam}

\newcommand{\nll}{\texttt{null}}

\newcommand{\outlier}{\ensuremath{\mathscr{O}}}
\newcommand{\goodlocal}[1]{\ensuremath{(#1)}-good-local partition}

\newcommand{\vc}{{\ensuremath{{\rm VC}}}}

\newcommand{\problemSS}{\textbf{\textsf{Subset-Sum}}}

\newcommand{\problemRMSE}{\textbf{\textsf{Min-RMSE}}}
\newcommand{\instanceSS}{\ensuremath{\cI_{\rm SS}}}
\newcommand{\instanceMSE}{\ensuremath{\cI_{\rm MSE}}}

\newcommand{\dataset}{\mathscr{C}}
\newcommand{\opt}{{\rm OPT}}

\def\colt{} %
\let\colt\undefined %

\ifdefined\colt
  \newcommand{\proofof}[1]{\textbf{of #1.}}
  \newcommand{\proofsketch}[1]{\textbf{sketch of #1.}}
\else
  \makeatletter
  \renewenvironment{proof}[1][\proofname]{\par
  \pushQED{\qed}%
  \normalfont \topsep6\p@\@plus6\p@\relax
  \trivlist
      \item[\hskip\labelsep
            \itshape
        #1\@addpunct{}\hskip\labelsep]%
    }{%
      \popQED\endtrivlist\@endpefalse
    }
  \makeatother

  \renewcommand{\proofname}{Proof}
  \newcommand{\proofof}[1]{\textit{\hspace{-4mm} of \expandafter#1.}}%
  \newcommand{\proofsketch}[1]{\textit{\hspace{-4mm} sketch of \expandafter#1}.}
\fi

\renewcommand{\paragraph}[1]{\medskip \noindent\textbf{#1}~}
    
\newcommand{\smallmath}[1]{\ensuremath{\text{\small $#1$}}}

\setlist[enumerate]{after=\vspace{\dimexpr-\baselineskip+12pt\relax},itemsep=2pt}
\setlist[itemize]{after=\vspace{\dimexpr-\baselineskip+12pt\relax},itemsep=2pt}

\newcommand\blfootnote[1]{%
  \begingroup
  \renewcommand\thefootnote{}\footnote{#1}%
  \addtocounter{footnote}{-1}%
  \endgroup
}

\usepackage{array}
\newcolumntype{L}[1]{>{\raggedright\let\newline\\\arraybackslash\hspace{0pt}}m{#1}}
\newcolumntype{C}[1]{>{\centering\let\newline\\\arraybackslash\hspace{0pt}}m{#1}}
\newcolumntype{R}[1]{>{\raggedleft\let\newline\\\arraybackslash\hspace{0pt}}m{#1}}

\title{Smaller Confidence Intervals From IPW Estimators\\ via Data-Dependent Coarsening}

\author{
        {\begin{tabular}{C{4.75cm}C{4.75cm}C{4.75cm}}
        {\bf Alkis Kalavasis} & {\bf Anay Mehrotra} & {\bf Manolis Zampetakis}\\[2mm]
        {Yale University} & {Yale University} & {Yale University} \\[-4mm]
        {\small\phantom{....................}} \mbox{\small\url{alkis.kalavasis@yale.edu}} & {\small \phantom{............}} \mbox{\small\url{anaymehrotra1@gmail.com}} & \mbox{\small\url{manolis.zampetakis@yale.edu}}
        \\
        \end{tabular}}
}

\date{}

\begin{document}

\maketitle

\begin{abstract}
    Inverse propensity-score weighted (IPW) estimators are prevalent in causal inference for estimating average treatment effects in observational studies. Under unconfoundedness, given accurate propensity scores and $n$ samples, the size of confidence intervals of IPW estimators scales down with $n$, and, several of their variants improve the rate of scaling. However, neither IPW estimators nor their variants are robust to inaccuracies: even if a single covariate has an $\eps>0$ additive error in the propensity score, the size of confidence intervals of these estimators can increase arbitrarily. Moreover, even without errors, the rate with which the confidence intervals of these estimators go to zero with $n$ can be arbitrarily slow in the presence of extreme propensity scores (those close to 0 or 1).

    We introduce a family of Coarse IPW (CIPW) estimators that captures existing IPW estimators and their variants. Each CIPW estimator is an IPW estimator on a \toa{coarsened} covariate space, where certain covariates are merged. Under mild assumptions, e.g., Lipschitzness in expected outcomes and sparsity of extreme propensity scores, we give an efficient algorithm to find a robust estimator: given $\eps$-inaccurate propensity scores and $n$ samples, its confidence interval size scales with $\eps+\sfrac{1}{\sqrt{n}}$. In contrast, under the same assumptions, existing estimators' confidence interval sizes are $\Omega(1)$ irrespective of $\eps$ and $n$. Crucially, our estimator is data-dependent and we show that no data-independent CIPW estimator can be robust to inaccuracies. 

    \blfootnote{Accepted for presentation at the 37th Conference on Learning Theory (COLT) 2024}
    
\end{abstract}

\pagenumbering{arabic}

\section{Introduction}\label{sec:intro}

  An observational study consists of several individuals who are described by a vector of \textit{covariates} $x\in \R^d$. In the study, each individual $x$ receives \textit{treatment} with some fixed but unknown probability $e(x)$. We observe tuples of the form $(x, y, t)$ of whether the treatment was assigned ($t = 1$) or not ($t = 0$) and the corresponding treatment-dependent \textit{outcome} $y$. 
  {Let {$X$ be the covariate random variable and}  $Y(t)$ be the random variable denoting the outcome when $T=t$ (for $t\in\zo{}$), so that $y = T Y(1) + (1-T) Y(0)$} {(note that $T, Y(1), Y(0)$ can depend on $X$).}
  The key quantity of interest is the \textit{average treatment effect} $\tau$, which is the average effect of the treatment on the population of individuals, i.e., {$\tau \coloneqq\Ex[Y(1)-Y(0)]$}. As an example the individuals can be patients with covariates that include the details of their medical history, the treatment corresponds to whether they received some medication or not, and the outcome is the the extent of the patient's symptoms. In another example the individuals are customers, the covariates include price sensitivity and interests, the treatment corresponds to whether they receive some discount, and the outcome is their satisfaction.

  The difficulty in estimating the average treatment effect $\tau$ is that, for each individual, we only observe their outcome either with or without treatment but not both, i.e., the data is \textit{censored}. To estimate $\tau$ from censored data, we need to account for the probabilities $e(x)$ with which each individual is assigned treatment. These probabilities are known as \textit{propensity scores}.

  \textit{Inverse propensity-score weighted (IPW)} estimators are a family of estimators that use the propensity scores $e(x)$ for each $x$ and a censored dataset $\dataset$ as input, and output a value $\ipw{}(\dataset;e)$. This value is an unbiased estimate of $\tau$ if $\dataset$ is drawn from an underlying distribution $\cD$ that satisfies a standard assumption--namely \emph{unconfoundedness}--which we formally define in \cref{sec:preliminaries}.
  IPW estimators are widely used in a variety of fields from Economics \citep{dehejia1998causal,dehejia2002propensity,galiani2005water,abadie2006large,abadie2016matching}, to Statistics \citep{rosenbaum2002overt,rubin2006matched}, to Medicine \citep{robin1997estimating,christakis2003health,austin2008critical}, to Political Science \citep{brunell2004turnout,sekhon2004quality,ho2007matching}, to Sociology \citep{morgan2006matching,lee2009estimation,oakes2017methods}.
    There are several reasons for their popularity including that $\ipw{}$ is: (1) easy to describe, (2) efficiently computable, (3) unbiased given \emph{accurate} propensity scores $e(x)$, and (4) asymptotically normal. Nevertheless, these vanilla IPW estimators suffer from some stability issues that we explore below.
    
    \paragraph{Issue I: Inaccuracies.} 
    IPW estimators enjoy the aforementioned good statistical properties when provided with accurate propensity scores. 
    Unfortunately, the IPW estimator $\ipw{}(\dataset;\wh{e})$ can be arbitrarily far from $\tau$, when the propensity score estimates $\wh{e}$ slightly differ from the true propensity scores $e$ 
    \emph{even on one covariate}
    by an additive amount $\eps>0$. This leads us to the following question.
    \begin{mdframed}
        \textbf{Q1.}
        {Given inaccurate propensity scores $\wh{e}$, s.t., $\norm{\wh{e}-e}_\infty\leq\eps$, can we estimate $\tau$ to $O(\eps)$ error?}
    \end{mdframed} 
    \noindent %
    In the special case of $\eps=0$, practitioners and researchers have devoted significant effort to reducing the size of the confidence intervals arising from IPW estimators, e.g.,  \citep{hirano2003efficient,li2018overlapWeights}. %
    This includes several doubly-robust estimators that utilize predictions of covariate-level outcomes {(i.e., {predictions of $\mu_t(x)\coloneqq\Ex[Y(t)\mid X{=}x]$} for $t=0$ and $t=1$)} to achieve smaller confidence intervals \citep{robins2005doublyRobust,chernozhukov2018doubleML,chernozhukov2022simple,foster2023orthognalSL}.
    However, the understanding of the $\eps > 0$ case is limited (see \Cref{sec:related_works} for some prior work on inaccurate propensity scores). %

    \paragraph{Issue II: Outliers.} %
    Root-mean-squared-error (RMSE) is the main quantity of interest of an estimator of the average treatment effect $\tau$; not only because it checks the consistency of the estimator but also because it controls the size of the resulting confidence interval at a fixed confidence level. It is well known that given $n$ samples, the RMSE of both IPW and doubly-robust estimators is proportional to $\sqrt{\frac{1}{n} \Ex_\cD\insquare{\frac{1}{e(X)(1-e(X))}}}$ \citep{imbens2015causal,farrell2015robust,chernozhukov2022simple,foster2023orthognalSL}. %

    An \textit{outlier} in this context is defined as a covariate $x$ whose propensity score $e(x)$ is close to $0$ or $1$. To be more concrete, let covariate $x$ be a $\beta$-outlier if $e(x)(1-e(x)) < \beta$. 
    One important issue with the aforementioned expression of RMSE is that it goes to $\infty$ if there exists an outlier with positive mass.
    Therefore, even a single $\beta$-outlier, with small $\beta$, can significantly affect the size of the confidence interval of $\tau$.

    A popular heuristic to increase robustness to such outliers is to remove or trim all $\beta$-{outliers} from the data for some $\beta = \Omega(1)$  \citep{crump2009dealing}.
    While this ensures that the standard deviation is at most $O\sinparen{\sfrac{1}{\sqrt{\beta n}}}$, it can increase the bias of the estimation to $\rho$, where $\rho$ is the probability mass of the $\beta$-outliers \citep{li2018overlapWeights}.
    Hence, this would result in an RMSE and, hence, confidence interval of size proportional to $\rho+O\inparen{\sfrac{1}{\sqrt{\beta n}}}$ which is finite, as opposed to the RMSE of IPW which is $\infty$ in this case, but still it does not converge to zero. Unfortunately, without further assumptions, it is information-theoretically impossible to find an estimator with RMSE that goes to zero as $n$ goes to infinity (for more details see \cref{thm:infoLB:alphaL}). This leads us to the following question.

    \begin{mdframed}
        \textbf{Q2.}
        If $\rho$-fraction of covariates are $\beta$-outliers, then are there natural assumptions on outliers that enable estimation of $\tau$ with an RMSE much smaller than $\rho+O\sinparen{\sfrac{1}{\sqrt{n\beta}}}$?
    \end{mdframed}
    {In this work, we study both questions \textbf{Q1} and \textbf{Q2}\textbf{.}
    When $\eps>0$, in the regime of \textbf{Q1}\textbf{,} none of the aforementioned estimators (i.e., IPW, doubly-robust, and trimmed-IPW estimators) guarantee RMSE that decays with $n$; unless there are no $\beta$-outliers for $\beta\gg \eps$.
    When $\rho=\omega\sinparen{\sfrac{1}{\sqrt{n}}}$, in the regime of \textbf{Q2}\textbf{,} then none of the aforementioned \mbox{estimators guarantee RMSE smaller than $\rho+O\sinparen{\sfrac{1}{\sqrt{n\beta}}}$.}}

    \subsection{Our Main Result: Estimation Robust to Inaccuracies and Outliers} \label{intro:main}
    
    Our goal is to efficiently find an estimator whose confidence intervals satisfy the properties desired in \textbf{Q1} and \textbf{Q2}\textbf{.} We make the following natural assumption on the outcomes.
    
    \begin{assumption}[{Lipschitzness}]\label{asmp:lipschitzness}
      {{The expected outcome with treatment $t$ conditioned on covariate $x$,} i.e., $\mu_t(x)\coloneqq \Ex[Y(t) \mid X{=}x]$}, is $L$-Lipschitz in $x$ (for some $\ell_p$-norm), i.e., $\abs{\mu_t(x_1)-\mu_t(x_2)}\leq L\cdot \norm{x_1-x_2}_p$ for all $x_1$ and $x_2$.
    \end{assumption}
    We expect Lipschitzness to hold in many settings such as when $\mu_t(x)$ belongs to a parametric family {(e.g., see \citet{wager2018estimation} who require Lipschitzness).}
    For instance, a common assumption in the literature is that $\mu_t$ has a linear parametric form, say, $\mu_t(x)=w_t^\top x +\eps$ \citep{imbens2015causal,hernan2023causal}.
    This implies that 
    $\mu_t$ is
    $\norm{w_t}_1$-Lipschitz and satisfies \cref{asmp:lipschitzness} when $w_t$ has a bounded norm.
    Lipschitzness of $\mu(x)$ alone, however, does \emph{not} address any of the issues that we mentioned before: all aforementioned estimators are independent of the geometry of the covariates, hence, we can always rearrange the covariates to satisfy Lipschitzness.
    Our estimators, on the other hand, will make use of the geometry of covariates. %

    Apart from the Lipschitzness of outcomes, we also need assumptions to exclude certain edge cases where it seems hard (if not impossible) to achieve a small RMSE. %
        For some small parameter $\alpha>0$, a number $k=O(1)$, we assume that $\beta$-outliers satisfy the following assumptions. %
        \begin{assumption}[{Sparsity}]\label{asmp:sparsity}
            There is no $\ell_p$-ball of diameter larger than $\alpha$ {such that} $\inparen{1-o(1)}$-fraction {the covariates within the ball are} $\beta$-outliers.
            (Where norm is the same as in \cref{asmp:lipschitzness}.)
        \end{assumption}
        \vspace{-7mm}
        \begin{assumption}[{Isolation}]\label{asmp:isolation}
            There are $k$ $\ell_p$-balls of diameter $\alpha>0$ whose centers are, pairwise, $3\alpha$ far in $\ell_p$-norm and that partition all $\beta$-outliers.
            (Where norm is the same as in \cref{asmp:lipschitzness}.)
        \end{assumption}
        If sparsity does not hold then, under mild conditions, there is a large ball $B$ most of whose mass is $\beta$-outliers (see \cref{fig:assumptions:counterExample:sparsity}).
        Since $B$ has a large radius, covariates well inside $B$ are far from any non-outlier and, hence, we cannot use Lipschitzness to estimate the expected outcomes (with and without treatment) at these covariates.
        Further, estimating these expected outcomes directly may require arbitrarily large samples as their propensity scores can be arbitrarily close to 0 or 1.

        If isolation does not hold, then either (i) we need a large number of (small) balls to cover the covariates {(see  \cref{fig:assumptions:counterExample:isolation})} or (ii) any cover with $O(1)$ balls has two close balls.
        In case (i), we need a huge number of samples to identify the good partition $\mathcal{S}$.
        Intuitively, this is because the VC-dimension of unions of $k$ balls is at least $\Omega(k)$, leading to a need for $\Omega(k)$ samples, which goes to $\infty$ as $k\to\infty$ (see \cref{lem:impossibility_of_learning_gl_partition}).
        In case (ii), the cover can be identified from finite samples, but how to identify it computationally efficiently is unclear. Under the assumptions above, we can find estimators with better guarantees than existing estimators.

        \begin{figure}[h!]
            \centering
            \subfigure[
                Sparse and Isolated Instance\label{fig:assumptions:example}]{     
                {\hspace{4mm}\includegraphics[scale=0.1,trim={22cm 2cm 20cm 3cm}]{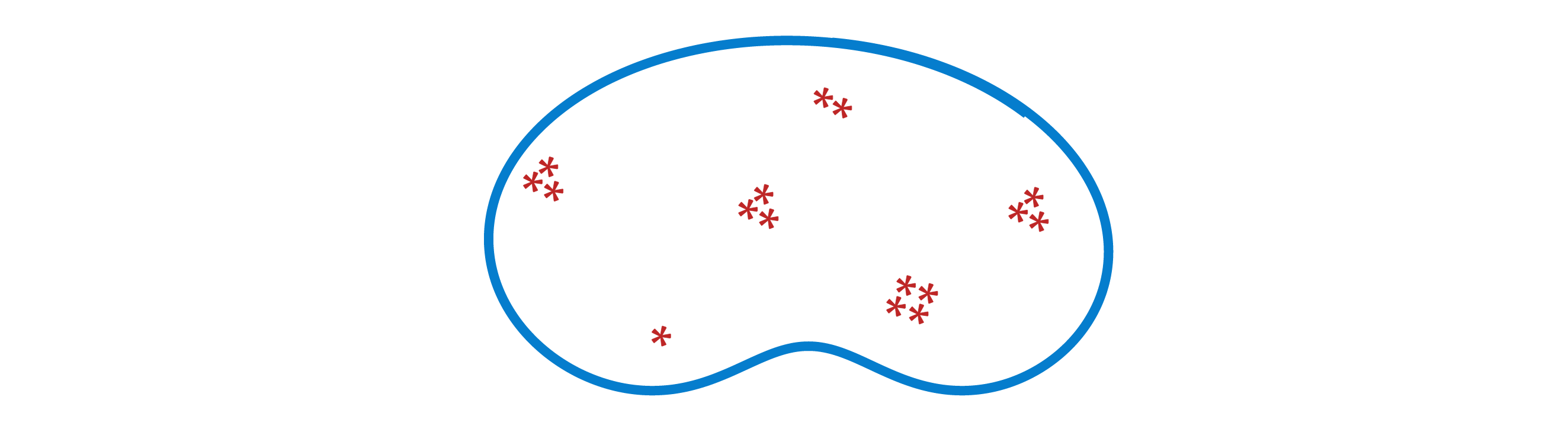}}
            }
            \quad
            \subfigure[
                Non-Sparse Instance\label{fig:assumptions:counterExample:sparsity}]{  
                \hspace{2.5mm}{\includegraphics[scale=0.1,trim={1cm 1cm 1cm 1.5cm}]{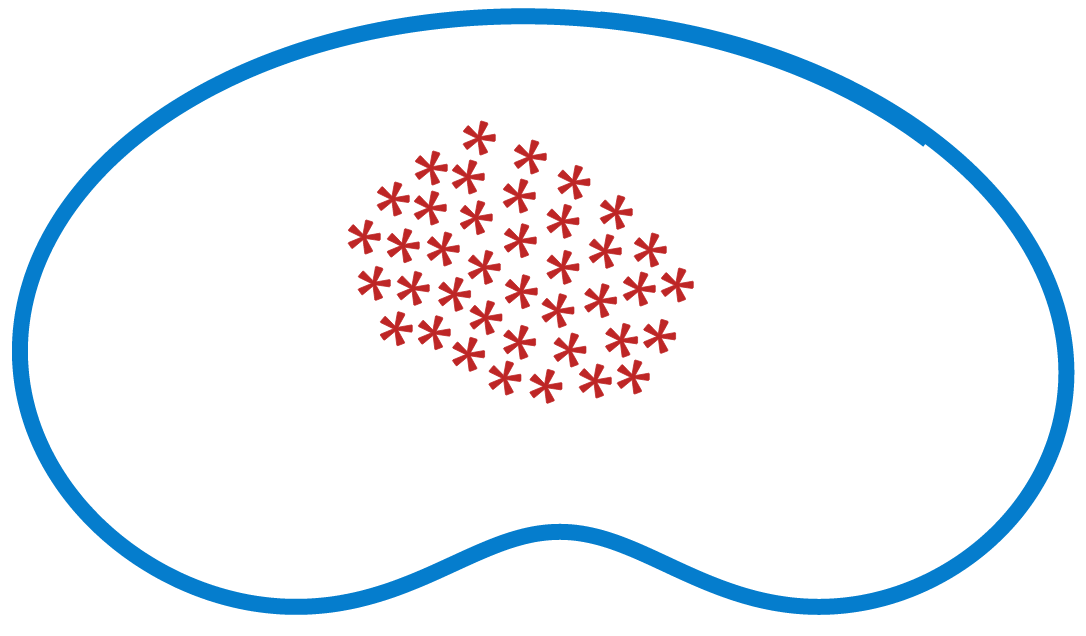}}
            }
            \qquad
            \subfigure[
                Non-Isolated Instance\label{fig:assumptions:counterExample:isolation}]{
                \hspace{2.5mm}{\includegraphics[scale=0.1,trim={1cm 1cm 1cm 1.5cm}]{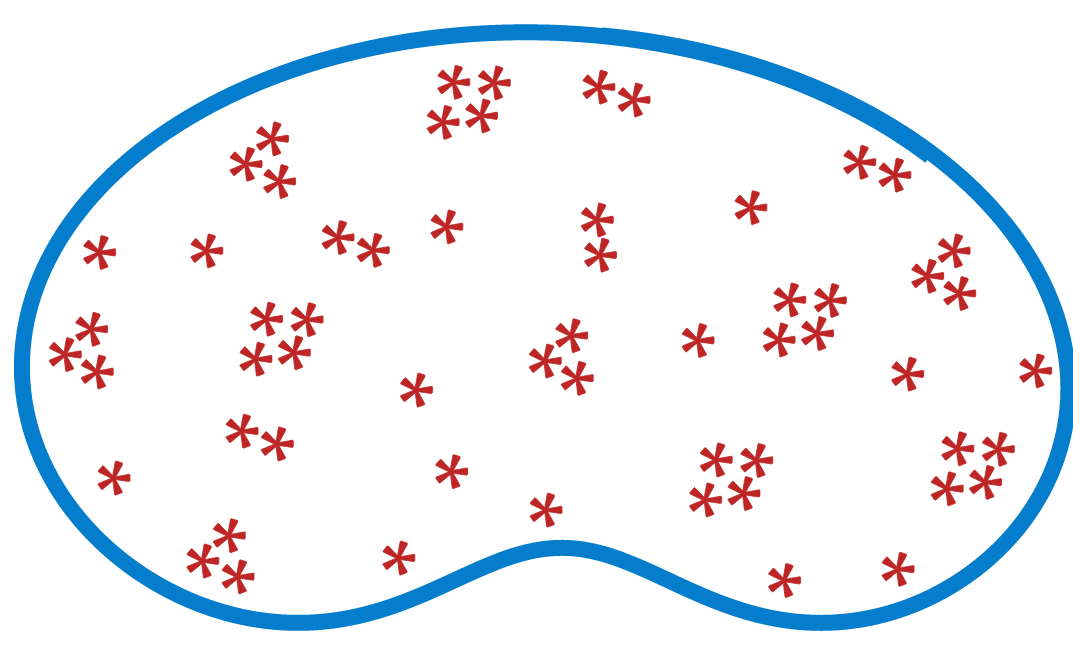}}
            }
            \vspace{-2mm}
            \caption{Illustrations of \cref{asmp:sparsity,asmp:isolation}. Outlier covariates are in \textcolor{niceRed}{red}. Inliers are hidden.}
            \label{fig:enter-label}
        \end{figure}

        \begin{inftheorem}\label{infthm:1}
            Suppose \cref{asmp:lipschitzness,asmp:sparsity,asmp:isolation} hold with $\alpha,\beta>0$ and $k,L=O(1)$ and assume that we are given estimated propensity scores $\wh{e}$ with $\norm{\wh{e}-e}_\infty \leq \eps$ and
            $n= \Omega(\sfrac{d}{(\rho\eps)^{2}})$ i.i.d.\ samples from an unconfounded distribution.
            There is an algorithm that 
            outputs a value $\tau_{A}$ in $\poly(n)$ time such that
            $\Ex[\abs{\tau-\tau_A}^2]^{\sfrac{1}{2}} = O(\eps+\alpha\rho+{\sfrac{1}{\sqrt{n\beta}}})$, where $\rho$ is the fraction of $\beta$-outliers.
        \end{inftheorem}
        This algorithm gives a unified approach for handling inaccuracies in propensity scores and outliers.
        Concretely, the algorithm, namely, \cref{algorithm}, answers both questions \textbf{Q1} and \textbf{Q2} in the affirmative:
        at one extreme, where there are no $\beta$-outliers (i.e., $\rho=0$) but the provided propensity scores have $\eps$-error, the algorithm achieves an RMSE $O(\eps)$ for $\beta \geq \sfrac{\rho}{d}$ and $O(\eps+{\sfrac{1}{\sqrt{n\beta}}})$ otherwise.
        At the other extreme, where accurate propensity scores are known, but $\rho$-fraction of the points are outliers, the algorithm achieves an RMSE $O(\alpha\rho+{\sfrac{1}{\sqrt{n\beta}}})$.
        This quantity is $\Omega\inparen{\sfrac{1}{\alpha}}$-times smaller than $\rho+O\sinparen{\sfrac{1}{\sqrt{n\beta}}}$ in the regime $\rho=\omega\sinparen{\sfrac{1}{\sqrt{n\beta}}}$ where trimmed-IPW performs {poorly}, and matches the guarantee of trimmed-IPW otherwise (when $\rho=O\sinparen{\sfrac{1}{\sqrt{n\beta}}}$).

        To put our algorithmic contribution in context, we verify that existing estimators have large confidence intervals under the same assumptions (see also \Cref{tab:RMSEBounds}).
        \begin{inftheorem}
        \label{prop:comparisonToBaselines}\label{infthm:2}
            For any $\eta,\eps>0$ and $n\geq 1$, there is an unconfounded distribution $\cD$ satisfying \cref{asmp:lipschitzness,asmp:sparsity,asmp:isolation} with parameters $(\alpha, \beta, k, L)=(\eta,\sfrac{1}{9},3,3)$ such that given inaccurate propensity scores $\wh{e}$ with $\norm{\wh{e}-e}_\infty\leq \eps$ and $n=\Omega(d/(\rho\eps)^2)$ samples %
            \begin{itemize}[itemsep=-3pt]
                \item[$\triangleright$] The RMSE of IPW and doubly robust estimators is $\Omega\left(\sfrac{1}{{\eta}} \right)$; %
                \item[$\triangleright$] The RMSE of $\eps$-Trimmed IPW estimator (which removes all $\eps$-outliers) is $\Omega(1)$; 
                \item[$\triangleright$] The RMSE of the Estimator in \cref{infthm:1} is $O(\eps+{\sfrac{1}{\sqrt{n}}})$.
            \end{itemize}
        \end{inftheorem} 

        \begin{table}[h!]
            \centering
                \caption{
                    $\eps$-Robust RMSE of different estimators where %
                    $(\alpha,\beta,L,k=O(1))$ are parameters of \cref{asmp:lipschitzness,asmp:sparsity,asmp:isolation},
                    $\rho$ is the mass of $\beta$-outliers, and
                    unconfoundedness holds for some unknown propensity scores.
                }
                \vspace{2mm} 
            \begin{tabular}{p{6.5cm} c c}
            \toprule
            {Estimator} & No Assumptions & \cref{asmp:lipschitzness,asmp:sparsity,asmp:isolation}\\ 
            \midrule
            IPW or Doubly-Robust Estimator & 
                $\infty$ & 
                    $\infty$ \\[1mm] 
            $\eta$-IPW Estimator (see Table \ref{tab:CIPWCapturesExistingEstimators})& 
                $O\inparen{\rho+\frac{1}{\sqrt{n\beta}}}$ & 
                    $O\inparen{{\frac{\eps}{\eta}} + \rho + \frac{1}{\sqrt{n\beta}}}$ \\[3mm] 
            \textbf{Our \cref{infthm:1}}  & 
                --- & 
                    $O\inparen{{\eps+\rho \alpha L} + \frac{1}{\sqrt{n\beta}}}$~\footnotemark{}
            \\[2mm] 
            \midrule
            Information-Theoretic Lower Bounds & 
                $\Omega\sinparen{\rho}$ (\cref{coro:infoLB:rho}) & 
                    $\Omega\sinparen{
                            \rho\cdot \min\inbrace{1, \alpha L} 
                        }$ (\cref{thm:infoLB:alphaL}) \\ 
            \bottomrule
            \end{tabular}
            \label{tab:RMSEBounds}
        \end{table} 

        \footnotetext{\cref{algorithm} does not take $L$ as input. If $L$ is provided, then $\alpha L$ in the RMSE can be improved to $\min{\sinbrace{1,\alpha L}}$: one can simply check if $\alpha L\leq 1$ and, if so, output the CIPW estimate, and, otherwise, use the $\beta$-Trimmed IPW estimate.}

        \subsection{Other Contributions} \label{sec:overview}

        In this section, we highlight some other important contributions of our work. In particular, in \cref{intro:cipw}, we explore the robustness of a natural extension of IPW estimators which we call \textit{coarse IPW estimators (CIPW)}, and, in \cref{intro:data}, we justify the need for data-dependent clustering.
        
        \subsubsection{Properties of CIPW Estimators} \label{intro:cipw} %
        
    We introduce a family of Coarse IPW (CIPW) estimators that captures almost all the existing IPW estimators in the literature and we explore their performance with respect to \textbf{Q1} and \textbf{Q2}. 
    
    \paragraph{CIPW Estimators.}
    Each estimator $\tau_{\cS, N}$ in this family is specified by a partitioning of the covariate space into sets $\cS=\inbrace{S_1, S_2, \dots}$ and an additional \textit{null} set $N$.
    Given $\inparen{\cS, N}$, the corresponding \CIPW{} estimator $\tau_{\cS, N}$ is defined as the IPW estimator over the coarse domain--where all covariates belonging to any one set $S\in \cS$ are merged into a single covariate and all covariates in the null set $N$ are ignored (see \cref{sec:cipw} for a formal definition).
    This family captures the IPW estimators as well as the aforementioned variants (see \cref{tab:CIPWCapturesExistingEstimators} below).
    (The only exception is doubly-robust estimators, which are captured by using doubly-robust estimators instead of IPW on the coarse domain.)
        
        \begin{table}[h!]
            \centering
                \caption{
                    \CIPW{} estimators capture existing IPW estimators (see \cref{sec:examplesIPW} for their definitions).
                }
                \vspace{2mm} 
                \small 
            \begin{tabular}{p{5cm} p{11cm}}
            \toprule
            \textbf{Estimator} &\textbf{Choice of $(\cS, N)$}\\ 
            \midrule
            IPW (aka Horvitz-Thompson)
                    & $\cS=\sinbrace{\inbrace{x}\colon x\in \R^d}$ and $N=\emptyset$\\%
            $\eta$-Trimmed IPW 
            & $N=\sinbrace{x\in \R^d\colon e(x)\not\in [\eta,1-\eta]}$ and $\cS=\sinbrace{\inbrace{x}\colon x\not\in N}$\\%
            Neyman Estimator
                    & $\cS=\R^d$ and $N=\emptyset$\\ %
            Balancing Score Estimator
                    & $\cS=\inbrace{\inbrace{x\colon b(x){=}z}\colon z\in \R}$ and $N=\emptyset$, where $b$ is the balancing score\\ %
            Blocking Estimator 
                    & $\cS$ is the partition formed by the blocks and $N=\emptyset$ \\ %
            \bottomrule
            \end{tabular}
            \label{tab:CIPWCapturesExistingEstimators}
        \end{table}

            \paragraph{Robust Root Mean Squared Error.} Since all estimators $E$ we consider, even the ones in \cref{intro:main}, are asymptotically normal, their confidence intervals scale with their root-mean-squared-error (RMSE), defined as $\sqrt{\Ex_{\dataset\sim \cD}\sinsquare{\sabs{E(\dataset; e)-\tau}^2}},$ where $\dataset$ is the dataset of size $n$ and $e$ is the propensity scores.
                Our interest is in the worse-case RMSE under additive perturbations of propensity scores: 
                for any $\eps>0$, let $B(e,\eps)$ be the $\ell_\infty$-ball of radius $\eps$ around $e$ and define 
                \begin{align*}
                    \RMSE_{\cD, \epsilon}{(E)} \coloneqq 
                    \max_{\wh{e} \in B(e,\epsilon)}
                    \sqrt{\mathbb{E}_{\dataset\sim \cD} \left[ |E\inparen{\dataset;\wh{e}} - \tau|^2\right]}\,.
                    \tagnum{$\eps$-Robust RMSE}{eq:robustRMSE}
                \end{align*}
                As mentioned in \cref{sec:intro}, for any $\eps>0$, the $\eps$-Robust RMSE of IPW and doubly-robust estimators can be arbitrarily larger than their respective non-robust RMSE (\cref{prop:IPW_has_bad_mse}).
                When $\eps=0$, we simplify the above notation and use $\RMSE_{\cD}{(E)}$.

            \paragraph{Computational Complexity.} %
                Since we care about a data-dependent way to design the CIPW estimator, a natural question that arises is to find the best CIPW estimator given 
                 a censored dataset $\dataset$ of size $n$ from $\cD$ and $\eps>0$. In particular, we have to find a partition $(\cS, N)$ such that $\tau_{\cS, N}$ has the minimum $\eps$-Robust RMSE among all \CIPW{} estimators.
                
                In other words, we want to solve the optimization problem: $\argmin_{\cS, N}\; 
                                            \RMSE_{\cD,\eps}{(\tau_{\cS, N})}.$
                We begin with the special case of this problem with $\eps=0$ and consider its decision version: given number $U$, verify whether $\min_{\cS, N}\; \RMSE_{\cD}{(\tau_{\cS, N})}\leq U$.
                We call this problem \problemRMSE{}, whose inputs are $U$ and (all relevant parameters of) $\cD$.
                
                Since our focus is on computational complexity, we further assume that for any fixed partition $(\cS, N)$, $\RMSE_{\cD}{(\tau_{\cS, N})}$ is efficiently computable.
                We show that even in this special case, the problem is $\mathsf{NP}$-hard.
                In fact, not only is it $\mathsf{NP}$-hard, but also notoriously hard to approximate.
                \begin{restatable}[Hardness of Approximation]{theorem}{hardnessApproximation}\label{thm:hardness:approximation}
                    If $\mathsf{P} \neq \mathsf{NP}$, then there is no exponential-factor approximation algorithm for \textbf{\problemRMSE{}}, i.e., there is no algorithm that, given an instance of \textbf{\problemRMSE{}} with bit-complexity $b$,\footnote{The bit complexity of $A$ is the number of bits required to encode $A$ using the standard binary encoding (which, e.g., maps integers to their binary representation, rational numbers as pair of integers, and vectors/matrices as a tuple of their entries) \cite[Section 1.3]{grotschel2012geometric}.} 
                    checks if there is a partition $(\cS, N)$ such that $\RMSE_{\cD}{(\tau_{\cS, N})}\leq e^{O(b)}\cdot U$. %
                \end{restatable}
                In the proof of \cref{thm:hardness:approximation} we construct an exponential-gap-inducing partition from \problemSS{} to \problemRMSE{}.
                We overview key ideas and present the complete proof %
                in \cref{proofOverivews:hardness}. 
            
            \paragraph{A Criterion Guaranteeing Small Robust RMSE.}
                Given the computational hardness of finding the CIPW estimator with the minimum RMSE (even approximately), we focus on finding a \toa{good} CIPW estimator that is robust to inaccuracies in propensity scores and outliers. 
                Toward this, we will analyze an upper bound on $\RMSE_{\cD}(\tau_{\cS, N})$ focusing, for now, on the $\eps=0$ case.
                The standard bias-variance decomposition implies that 
                $ 
                    \RMSE_{\cD}(\tau_{\cS, N})
                    = 
                    \bias_{\cD}(\tau_{\cS, N}) + {\variance_{\cD}(\tau_{\cS, N})}^{\sfrac{1}{2}}
                $.
                After algebraic manipulation, we derive explicit expressions for both quantities (\cref{lem:exp_of_mse}).
                While these expressions are complicated, the key takeaway is that with bounded outcomes and $L$-Lipschitzness, the bias and variance are upper bounded as follows 
                \begin{align}
                    \bias_\cD\inparen{\tau_{\cS, N}}
                    & \leq O(\cD(N)) + 
                     \frac{1}{1-\cD(N)} \sum_{S\in \cS} O(L)\cdot \diam{(S)}\cdot \cD(S)\,, \nonumber \\
                    \variance_\cD(\tau_{\cS, N})
                        & \leq 
                        \frac{1}{n} \cdot \frac{1}{1 - \cD(N)} \sum_{S \in \cS} 
                            \frac
                            {\cD(S)}
                            {
                                e(S) \inparen{1-e(S)}
                            }\,,
                        \label{eq:simpleUpperbound}
                \end{align}
                where $e(S)$ is the average propensity score over $S$, i.e., $e(S) = \frac{1}{\mathrm{vol(S))}}\int_S e(x) \text{d} x$ and $\diam(S) \coloneqq \max_{x,x'\in R} \norm{x-x'}_p$
                with the same $p$ as in \cref{asmp:lipschitzness} (Lipschitzness).
                While these formulas are still not very simple, they are at least independent of the distribution of the outcomes.
                Moreover, observe the following from \cref{eq:simpleUpperbound}:
                \begin{enumerate}
                    \item If the \textit{coarse} propensity score of each set $S\in \cS$ is $\Omega(1)$, i.e., $e(S)(1-e(S))=\Omega(1)$, and the mass of $N$ is bounded away from 1, e.g., $\cD(N)\leq 1-\Omega(1)$, then the variance is $O\inparen{\sfrac{1}{n}}$.
                    \item If the mass of $N$ is small $\cD(N)\leq \gamma$, then bias is at most $O(L)\cdot \max_{S\in \cS}\diam{(S)}+\gamma$. Hence, ensuring that the diameter of all (non-null) sets is small and $\cD(N)$ is small ensures that the bias is small.
                \end{enumerate}
                
                \noindent These arguments lead to the following simple criterion for a good partition for a CIPW estimator.
                \begin{definition}[Good-Local Partition]\label{def:goodlocal}
                    Given constants $\alpha, \beta, \gamma >0$, a partition $\inparen{\cS, N}$ is said to be an \goodlocal{\alpha,\beta,\gamma} if (1) $\cD\inparen{N}\leq \gamma$, (2) each $S\in \cS$ has diameter $\diam{\inparen{S}}\leq \alpha$, and (3) each $S\in \cS$ has coarse propensity score $e(S)(1-e(S))\geq \beta$.
                \end{definition}
                The aforementioned upper bounds (\cref{eq:simpleUpperbound}) imply that any CIPW estimator constructed from an \goodlocal{\alpha,\beta,\gamma} with $\gamma\leq \sfrac{1}{2}$ has a RMSE of at most $O\inparen{\alpha L+\gamma+\sfrac{1}{\sqrt{n \beta}}}$.
                Crucially, the criteria in \cref{def:goodlocal} are independent of the distribution of the outcomes which may not be known.
                Nevertheless, Lipschitzness is crucial to obtain such a criterion: for instance, without Lipschitzness, for any $\alpha,\beta,\gamma>0$ and any \goodlocal{\alpha,\beta,\gamma} different from IPW, \cref{lem:need_data_dependence} constructs a distribution $\cD$ where $\RMSE_{\cD}{\inparen{\tau_{\cS, N}(C)}}\geq \sfrac{1}{3}$.
                Apart from a small non-robust RMSE, importantly, we show that a good-local partition also implies a small robust RMSE.
                \begin{restatable}[Robust RMSE of Good-Local Partition]{lemma}{RobustMSEofLGpartition}\label{lem:RobustMSEofLGpartition}
                    Suppose \cref{asmp:lipschitzness} holds {with parameter $L>0$.}
                    An \goodlocal{\alpha,\beta,\gamma} for any $\alpha,\beta>0$ and $\gamma\in [0,\sfrac{1}{2}]$ satisfies
                    $\RMSE_{\cD, \epsilon}{\inparen{\tau_{\cS, N}}}
                            \leq 
                                O\sinparen{
                                    \alpha L
                                    + \inparen{\sfrac{\eps}{\beta}}
                                    + \gamma
                                    +\sfrac{1}{\sqrt{n\beta}}
                                }
                                $
                    for any $\eps\in [0, \sfrac{\beta}{2}]$. 
                    
                \end{restatable} 
                Hence, not only does a good-local partition have a small RMSE but, if the inaccuracy in the propensity scores $\eps$ is much smaller compared to the coarse propensity scores, then it also has a small $\eps$-Robust RMSE. %
                Contrast this with existing estimators: 
                    while $\eps$-Robust RMSE of IPW and doubly-robust estimators can be arbitrarily larger than their RMSE, the $\eps$-Robust RMSE and RMSE of a CIPW estimator based on a good-local partition are $\inparen{\sfrac{\eps}{\beta}}$-close.
                Intuitively, this is because in any \goodlocal{\alpha,\beta,\gamma} $(\cS, N)$, for all $S\in \cS$, the coarse propensity score $e(S)$ is bounded away from $0$ and 1.
                This is useful as, when propensity scores are bounded away from $0$ and 1, then $\eps$-additive errors in propensity scores also imply $\inparen{1\pm O(\eps)}$-\textit{multiplicative} errors---which are significantly easier to handle.  
                Finally, we note that CIPW estimators arising from any good-local partition are \emph{asymptotically normal} (\cref{sec:asymptoticNormality}), allowing us to obtain corresponding confidence intervals.

            \paragraph{Learning a Good-Local Partition.}
                We note that \cref{lem:RobustMSEofLGpartition} requires the partition to be independent of the dataset on which the estimator is evaluated.
                Hence, one must first learn a good-local partition $(\cS, N)$ using the dataset $\dataset$ and then evaluate the learned $\tau_{\cS, N}$ on a fresh dataset $\dataset'$.
                However, finding a good-local partition seems challenging even from a statistical perspective.
                For instance, to even test whether a given partition $(\cS, N)$ is good, we, at least, require a uniform convergence bound over all the possible sets that we might use, e.g., this is needed to even check whether $\cD(N) \leq \gamma$ or not.
                Our next result shows that, even assuming that the sets in the partition belong to a VC class, it is not possible to learn a good-local partition without further assumptions.
                \begin{lemma}[Informal; see \cref{sec:statisticalHardness}]\label{lem:impossibility_of_learning_gl_partition_intro}
                    Suppose \cref{asmp:lipschitzness} holds with $L = O(1)$ and there exists an \goodlocal{\alpha,\beta,\gamma} with all sets in a class with $O(1)$ VC dimension.
                    For any $n\geq 1$, there is an unconfounded distribution $\cD$ such that given a censored dataset $\dataset\sim \cD$ of size $n$, it is information-theoretically impossible to find an \goodlocal{\alpha,\beta,\gamma} with probability $>\sfrac{1}{8}$.
                \end{lemma}
                This lemma excludes the existence of any finite-sample algorithm that finds a good-local partition without further assumptions.
                Briefly, the proof uses that $\cS$ may contain a large number of subsets and reduces PAC-learning of a union of $\abs{\cS}$ intervals (of width at most $\alpha$) up to $\gamma$-error in the agnostic setting to finding an \goodlocal{\alpha,\beta,\gamma}.
                The result then follows by verifying that the VC dimension of this class is at least $\abs{\cS}$ and using lower bounds on the sample complexity of PAC learning in the agnostic setting \citep{shalev2014understanding}.

            This brings us to \cref{asmp:sparsity,asmp:isolation} (Sparsity and Isolation).
            These assumptions guarantee the existence of an \goodlocal{\alpha, \Omega(\beta), 0} $(\cS^\star, N^\star)$ where $\abs{\cS^\star}=k=O(1)$ and all sets of the partition belong to a family with VC-dimension $O(d)$. This allows us to build the algorithm promised in \cref{infthm:1}, although the resulting estimator is a generalized version of CIPW. We refer to \cref{sec:cipw,sec:algoOverview} for the technical details.

     \subsubsection{Need for Data Dependence} \label{intro:data}
                In this section we ask whether there exists some CIPW, different from IPW, that has better RMSE uniformly over all unconfounded distributions $\cD$.
                \begin{restatable}[Impossible to Weakly Beat IPW]{lemma}{needDataDependence}\label{lem:need_data_dependence}%
                    For any $n \geq 1$ and $\tau_{\cS, N}$ distinct from IPW, there is an unconfounded distribution $\cD$, such that $\RMSE_{\cD}(\tau_{\cS, N}) \geq \sfrac{1}{3}$ and $\RMSE_{\cD}(\ipw{}) = O(\sfrac{1}{\sqrt{n}}).$
                \end{restatable}
                Therefore, not only is it impossible to weakly beat the RMSE of the IPW estimator, but any \CIPW{} estimator different from IPW has at least a \textit{constant} RMSE for some distribution $\cD$ irrespective of the number of samples $n$.
                This negative result demonstrates that the true power of CIPW estimators can only be realized by designing the estimator in a \emph{data-dependent way} (see \cref{sec:data-dependent} for further discussion).

\subsection{Related Work}\label{sec:related_works}

    \textit{Sensitivity Analysis} in Causal Inference includes a large body of work studying the performance of estimators with inaccurate propensity scores \citep{rosenbaum2002observational}. 
    The literature on sensitivity analysis studies several error models, including pointwise multiplicative errors in propensity scores (which are implied by bounding the odds ratio, i.e., $\max_x \inparen{\nfrac{e(x)}{\wh{e}(x)}}\cdot\inparen{\nfrac{(1-\wh{e}(x)}{(1-e(x))}}$) \citep{tan2006distributional,kallus2019interval,kallus2021minimax}, errors with bounded expected value of the odds-ratio \citep{huang2023variancebased}, and fixed interval bounds on propensity scores \citep{aronow2012interval}.
    However, almost all of the works operate under the assumption that there are no $\beta$-outliers for some $\beta=\Omega(1)$.
    In contrast, we allow the existence of outliers and design new estimators that, under mild assumptions (\cref{asmp:lipschitzness,asmp:sparsity,asmp:isolation}), have a smaller RMSE than existing estimators (\cref{infthm:2}).

    \smallskip \noindent \textit{Robust Statistics} includes a growing number of works that, broadly speaking, solve estimation tasks amidst data inaccuracies (ranging from adversarial corruption \citep{huber2004robust} to interval censoring \citep{cohen1991truncated})--with a growing focus on computational efficiency \citep{diakonikolas2019recent}.
    The simplest task is perhaps estimating the mean of a Gaussian distribution when a \(\rho\)-fraction of samples is adversarially corrupted.
    Here, common estimators (e.g., the empirical mean) fail, but robust alternatives can estimate the true mean within an additive factor \(O(\rho)\).
    Moreover, without additional assumptions, it is information-theoretically impossible to achieve a better estimation of the mean.
    Such information-theoretic lower bounds also arise in our work.
    For instance, given a distribution $\cD$ with a mass $\rho$ of outliers, any estimator must have an RMSE of at least $\Omega(\rho)$ (\cref{sec:infoLB}).
    Under \cref{asmp:lipschitzness,asmp:sparsity,asmp:isolation}, however, one can find more robust estimators (e.g., \cref{infthm:1}). %

    \smallskip\noindent \textit{Learning propensity scores} from data is necessary as propensity scores are unknown in an observational study.
    A number of Machine Learning methods are used for estimating propensity scores: from simple methods such as logistic regression (see \citet{westreich2010estimation}), to more sophisticated methods such as regression trees, random forests, and boosted variates of these methods \citep{mccaffrey2004propensity}, to neural networks \citep{westreich2010estimation}.
    These methods, however, are susceptible to errors in the estimated propensity scores.
    In this work, we show that they can be combined with CIPW estimators to increase the robustness of the resulting estimators to these errors.

    \smallskip \noindent{
        \textit{Non-propensity-score-based estimators} of the average treatment effect predict the expected (potential) outcomes conditional on the covariates, i.e., they predict $\mu_t(x)\coloneqq \Ex[Y(t)\mid X{=}x]$ for each covariate $x$ and $t\in \zo$ \cite{chernozhukov2024applied}.
        While this quantity is identifiable under unconfoundedness,~\footnote{To see why the quantity is identifiable, observe that, under unconfoundedness, $\Ex[Y(t)\mid X{=}x] = \Ex[Y(t)\mid X{=}x, T{=}t]$ and $Y(t)$ is observed when $T=t$.} the estimations error scale with inverse propensity scores (concretely, with $\frac{1}{{e(x)(1-e(x))}}$) and, hence, can be large in the presence of outliers $x$ whose propensity scores $e(x)$ are close to 0 or 1 \cite{wager2018estimation}.
        It would be interesting to see if a similar approach as CIPW estimators can reduce the errors of estimators predicting the conditional expected outcomes $\mu_t(x)$ in the presence of outliers.
        }

\section{Preliminaries}
\label{sec:preliminaries}
    \paragraph{Causal Inference Setup.}
        An observational study involves several individuals or units (e.g., patients); each described by a vector of \textit{covariates} $X\in \R^d$ (encoding, e.g., medical history).
        In the study, each unit $x$ is assigned a treatment $T$ (e.g., medication)
            with some fixed but known probability $e(x)\in (0,1)$ \textit{independent} of all other units and we observe a %
            treatment-dependent outcome $Y(T)$ (e.g., the extent of the patient's symptoms).
        We focus bounded outcomes $Y(T)\in [-1,1]$ and binary treatments $T\in \zo$.\footnote{Our results extend to categorical treatments paying the standard degradation with the number of categories and to outcomes in $[-B, B]$ (for any $B>0$) with linear degradation in the RMSE with $B$.}
        The tuple $(X,Y(0),Y(1),T)$ has an unknown joint distribution $\cD$.
        The following parameters of $\cD$ are of interest: for each $x$ and $t\in \zo$:
        \[
            \cD(x) \coloneqq \Pr_\cD[X{=}x]\,,\qquad 
            \mu_t(x) \coloneqq \Ex_\cD[Y(t) \mid X{=}x]\,,\qquad 
            v_t(x) \coloneqq \var_\cD[Y(t) \mid X{=}x]\,.
        \]
        which denote the probability mass, conditional means, and conditional variances at $x$ respectively.
        Of specific interest will be outlier covariates: for any $\beta>0$ and (possibly inaccurate) propensity scores $e\colon \R^d\to (0,1)$, covariate $x\in \R^d$ is said to be a $\beta$-outlier with respect to $e$ if $e(x)(1-e(x)) < \beta$.
        We denote the set of $\beta$-outliers with respect to $e$ by 
        \[
            \white{.}\qquad\outlier{}(\beta; e)\coloneqq \inbrace{x\in \R^d\colon e(x)(1-e(x)) < \beta}\,.
            \tag{$\beta$-outliers w.r.t. $e$}
        \]
        
    \paragraph{Average Treatment Effect.}
        An important goal in causal inference is to estimate the effect of treatment on the average outcomes, which is known as the average treatment effect \citep{imbens2015causal,hernan2023causal}.
        Concretely, the average treatment effect is defined as 
        \[
            \white{.}\hspace{30mm}
            \tau \coloneqq 
            \Ex_\cD\insquare{Y(1) - Y(0)}.
            \tag{Average Treatment Effect}
        \]
        If one had access to independent samples from $\cD$, then we could estimate the above expectation using empirical averages.
        The main difficulty is that the samples are \emph{censored}: instead of observing sample $s=(X, Y(0), Y(1), T)$, only the censored sample $c=(X, Y(T), T)$ is observed.
        That is if the unit was assigned treatment, then $Y(0)$ is censored, and $Y(0)$ is censored otherwise.
        Due to this censoring, $\tau$ is unidentifiable without additional assumptions on $\cD$.\footnote{{One simple example demonstrating this is as follows: let $Y(0)$ be deterministically $0$ and $Y(1)=T\cdot A + (1-T)B$ for some random variables $A$ and $B$ independent of $T$. Here, $\tau=\Ex_\cD[Y(1)-Y(0)]=\Ex[T]\Ex[A]+\Ex[1-T]\Ex[B]$. However, since $Y(1)$ is only observed if $T=1$, one does not gain any information about $\Ex[B]$ and, hence, $\Ex[B]$ is non-identifiable from the censored data implying that $\tau$ is also non-identifable.}}
        Unconfoundedness is a standard assumption that enables identifiability.
        It requires that conditional on the covariates the outcomes are independent of the treatment, i.e., 
        \begin{equation}
            \white{.}\hspace{30mm}
            Y(0), Y(1) ~\bot~ T ~~\mid~~ X\,.
            \tagnum{Unconfoundedness}{eq:unconfounded}
        \end{equation}   
        Unconfoundedness is known by many other names: ignorability, conditional exogeneity, conditional independence, and selection on observables \citep{imbens2015causal,hernan2023causal}.

    \paragraph{Estimators.} %
        An estimator $E$ of $\tau$ is a function that takes a censored dataset $\dataset =\inbrace{\inparen{x_i, y_i, t_i}\colon 1\leq i\leq n}$ as input and outputs a scaler value $E(\dataset)$.
        Apart from $\dataset$, $E$ may also get some additional or \textit{nuisance} parameters $\eta$, such as the propensity scores, in which case, we overload the notation to $E(\dataset;\eta).$ 
        (See \cref{sec:examplesIPW} for several examples.)
        Ideally, we want $E(\dataset;\eta)$ to be equal to $\tau$.
        Since this is almost never possible, we settle for confidence intervals on $\tau$.
        To obtain confidence interval $E(\dataset;\eta)$ must have a known distribution, e.g., a standard normal distribution.
        This is guaranteed by asymptotic normality (which requires that $E(\dataset;\eta)$ approaches a normal distribution as $\abs{\dataset}\to \infty$).
        We show that CIPW estimators constructed from good-local partitions are asymptotically normal in \cref{sec:asymptoticNormality}.

\section{Coarse Inverse Propensity-Score Weighted Estimators}\label{sec:cipw}
    In this section, we formally define CIPW estimators, which are inspired by the success of tree-methods in estimating treatment effects \citep{mccaffrey2004propensity,wager2018estimation}.
    \begin{definition}[{CIPW Estimators}]
        Given a partition $\inparen{\cS=\inbrace{S_1,S_2,\dots},N}$ of $\R^d$, 
        where both $(S_i)_i$ and $N$ are subsets of $\R^d$,
        the corresponding CIPW estimator with input a censored dataset $\dataset = \{(x_i, y_i, t_i) : 1 \leq i \leq n\}$ and (possibly inaccurate) propensity scores $e\colon \R^d\to (0,1)$ is %
        \[
            \tau_{\cS, N}(\dataset; e)
                \coloneqq 
                \frac{1}{
                \abs{\sinbrace{i\in [n] \colon x_i \not\in N}}
                }
                \cdot 
                \sum_{S\in \cS} \sum_{i\colon x_i\in S}
                    \inparen{
                        \frac{t_i y_i}{e(S)}
                        - \frac{(1 - t_i) y_i}{1 - e(S)}
                    }\,,
            \yesnum\label{eq:CIPW_expression_on_data}
        \]
        where $e(S)$ is the \emph{coarse propensity score} over $S$, i.e., $e(S)\coloneqq \Pr_\cD\insquare{T=1\mid x\in S}.$ 
    \end{definition}
    As mentioned in \cref{sec:intro}, the above definition is quite general and captures several existing propensity score-based estimators (see \cref{tab:CIPWCapturesExistingEstimators}).
    The motivation of CIPW estimators comes from the use of decision trees and random forests for learning propensity scores \citep{mccaffrey2004propensity,wager2018estimation}: these methods also partition the covariate domain and learn a ``common'' propensity score for all covariates in the domain.
    We present expressions for the bias, variance, and RMSE of CIPW estimators in \cref{sec:expression_of_mse}.
    We show that $\tau_{\cS, N}$ is asymptotically normal whenever $\min_{S\in \cS} e(S)(1-e(S))$ is bounded away from 0 (as for good-local partitions)  in \cref{sec:asymptoticNormality}.

    \paragraph{Fractional CIPW Estimators.} If we restrict to the CIPW estimators defined above, then RMSE guarantee in \cref{infthm:1} deteriorates to $\eps+\alpha L+\myfrac{\sqrt{n\beta}}$.
        To achieve the improved guarantee of $\eps+\rho\cdot \alpha L+\myfrac{\sqrt{n\beta}}$, we need to consider \emph{fractional CIPW estimators}: %
            each fractional CIPW estimator is defined by a fractional partition $(\cS, N, w)$, where, for each $T\in \cS\cup \inbrace{N}$ and covariate $x$, there is a weight $w_T(x)\in [0,1]$.
            The weights satisfy the natural condition that they sum to 1 for any covariate $x$, i.e., $\sum_T w_T(x)=1$.

        Before defining $\tau_{\cS, N, w}$, we explain the connection between a fractional partition $(\cS, N, w)$ and a usual partition; they have many similarities. %
        Concretely, $(\cS, N, w)$ is a (usual) partition of an extension $\hypo{X}_{\cS, N, w}$ of the original covariate-domain $\hypo{X}=\R^d$: in the extended covariate-domain, each covariate $x\in \hypo{X}$ (with mass $\cD(x))$, is \toa{split} into multiple covariates $\inbrace{x_T\colon T\in \cS\cup\inbrace{N}}$ which have probability masses $\inbrace{\cD(x)\cdot w_T(x)\colon T\in \cS\cup \inbrace{N}}$ respectively.
        For any $(\cS, N, w)$, $\tau_{\cS, N, w}$ is defined as the CIPW estimator on the new domain $\hypo{X}_{\cS, N, w}\times  [-1,1] \times \zo$.
        One issue is that the dataset $\dataset$ we receive lies in the original domain and we would like to have samples in the extended domain, where the fractional CIPW operates.
        To address this issue we proceed as follows: for each $(x_i,y_i,t_i)\in \dataset$, we replace it with $((x_i)_T, y_i, t_i)$ with probability $\propto w_T(x_i)$ for each $T\in \cS\cup\inbrace{N}$.
        Note that such a re-sampling process is necessary, since if one naïvely added weights to \cref{eq:CIPW_expression_on_data}, then the terms will no longer be independent.
        
        To summarize, fractional CIPW estimators are defined on a fractional partition $(\cS, N, w)$, they are CIPW estimators on an extension of the domain, and given $\dataset$, we can generate samples from the extended domain and, hence, compute $\tau_{\cS, N, w}$.
        Further, we can compute the RMSE of fractional CIPW estimators by using the expression of RMSE of a CIPW estimator on the extended domain.
        Finally, a fractional partition is said to be \goodlocal{\alpha,\beta,\gamma} if the corresponding CIPW estimator on the extended domain is \goodlocal{\alpha,\beta,\gamma}.

    \paragraph{Computing Coarse Propensity Scores.}
    The expression of $\tau_{\cS, N}$ involves coarse propensity scores $\inbrace{e(S)\colon S\in \cS}$ (\cref{eq:CIPW_expression_on_data}).
        These coarse propensity scores can be estimated from data, akin to the usual propensity scores.
        For instance, under standard assumptions on 
            (i) propensity scores (e.g., finite fat shattering dimension \citep{alon1997scale}; as implied by common smoothness assumptions \citep{hirano2003efficient}) and
            (ii) sets in $\cS$ (e.g., finite VC dimension \citep{blumer1989learnability} and lower bounds on probability mass), 
            coarse propensity scores can be estimated from data:
            for instance, given $\wh{e}\colon \R^d\to (0,1)$ and $\dataset$ of size $n=\Omega\inparen{\eps^{-2}}$, 
            for any $S\in \cS$,
            the empirical average $\sfrac{\inparen{\sum_{i}\wh{e}(x_i) \mathds{I}[x_i\in S]}}{\inparen{\sum_{i}\mathds{I}[x_i\in S]}}$ 
            is additively $\eps$ close to $e(S)$ with high probability.
        For simplicity of exposition, we henceforth assume access to an oracle that given a set $S$ (from a class of finite VC dimension), outputs a value $v\in e(S)\pm \eps$. %

\section{Algorithmic Result and Overview of Our Algorithm}\label{sec:algoOverview}
        In this section, we present our main algorithmic result: an efficient algorithm that, given $\eps$-accurate propensity scores, outputs a prediction of $\tau$ that has a small $\eps$-Robust RMSE.
        \begin{restatable}[\textbf{Main Algorithmic Result}]{theorem}{mainTheorem}
            Suppose \cref{asmp:lipschitzness,asmp:sparsity,asmp:isolation} hold with parameters $\alpha,\beta,k,L$ and
            fix any $0\leq \eps\leq \beta/10$ and any $\delta > 0$.
            \label{thm:main}
            There is an algorithm (\cref{algorithm}) that, 
                given
                an estimate of propensity scores $\wh{e} \in B(e,\eps)$,
                constants $\alpha,\beta,\eps$, and
                a censored dataset $\dataset$ of size $n=\wt{\Omega}\inparen{\frac{1}{\rho^2\eps^{2}}\cdot {\inparen{k^3d+\log{(1/\delta)}}}{}}$, %
            outputs a value $\tau_A$ in $O(n^3)$ time such that, with probability at least $1-\delta,$ 
            \[
                    \Ex\insquare{\abs{\tau-\tau_A}^2}^{\sfrac{1}{2}}
                ~\leq~ 
                    O\inparen{
                        {\eps + \rho \alpha L} + \frac{1}{\sqrt{n\beta}}
                    }
                \qquadand \frac{\tau_A - \Ex[\tau_A]}{\sqrt{\var[\tau_A]}}~\xrightarrow{n\to\infty}~
                    \cN\inparen{0, 1}
                \,.
            \]
        \end{restatable}
        In particular, the estimator's RMSE only increases by a small additive amount due to inaccuracies in propensity scores and outliers.
        As shown by \cref{prop:comparisonToBaselines}, this is significantly better than existing estimators which, under the same assumptions, can have $\Omega(1)$ RMSE.
        Moreover, the RMSE is close to the information-theoretic lower bound of $\Omega\inparen{\rho\min\inbrace{1, \alpha L}}$. %
        The algorithm does not use $L$ or $k$.
        If it is given $L$ as input, then we can bring the RMSE closer to the lower bound by using \cref{thm:main} if $\alpha L<1$ and otherwise using $\beta$-Trimmed IPW Estimator.
        Finally, the estimator in \cref{thm:main} is {asymptotically normal} and, hence, enables \mbox{one to generate confidence intervals on $\tau$.}
        \smallskip
        
        The pseudocode of the algorithm is in \cref{algorithm} and the proof of \cref{thm:main} appears in \cref{sec:proofof:thm:main}.

        \subsection{Technical Overview of \cref{algorithm}}
        \cref{algorithm} splits $\dataset$ into two equal parts $\dataset_1$ and $\dataset_2$; $\dataset_1$ is used to learn a (fractional) good-local partition $(\cS, N, w)$ and $\dataset_2$ is used to evaluate $\tau_{\cS, N, w}$.
        Our algorithm contains multiple components that handle different challenges. For this reason, we divide the description of the algorithm into three stages, where each stage handles a more general setting. %
        \cref{algorithm} presents the final algorithm.

        \paragraph{Stage 1 - Accurate Propensity Scores ($\eps=0$).}
                In this case, the algorithm then constructs an \goodlocal{2\alpha, \Omega(\beta), 0} of covariates in $\dataset_1$ in $O(n^3)$ time.
                To do this, the algorithm covers all $\beta$-outliers in $\dataset_1$ by balls of diameter $2\alpha$; the outliers are identifiable as $\eps=0$.
                To do this efficiently, the algorithm utilizes that the $k$-balls in the cover of $\beta$-outliers are far (due to \cref{asmp:isolation}). %

                Let this partition be $(\cS, N)$.
                We want $(\cS, N)$ to also generalize to $\dataset_2$.
                Toward this, one can use standard uniform convergence bounds to show that all sets in $\cS\cup\inbrace{N}$ satisfy the conditions required by a good-local partition (this is where we use $\abs{\cS}=k=O(1)$).
                The difficulty is that there can be many covariates in $\dataset_2$ that are not in $\dataset_1$ and, hence, may not be covered by $(\cS, N)$. 
                In fact, there may be a large mass of them since we consider a continuous domain and, hence, have a 0 probability of sampling the same covariate twice.
                We show that \cref{asmp:isolation} implies that the set of $\beta$-outliers not covered by $(\cS, N)$ (whether in $\dataset_2$ or not) have a mass of at most $O\inparen{\sfrac{1}{n}}$ (\cref{lem:algorithm:step2}).
                This enables us to extend $(\cS, N)$ to a \goodlocal{2\alpha, \Omega(\beta), O(1/n)} of the whole domain $(\cS', N')$ (where we add all uncovered outliers to $N'$).

        \begin{algorithm}[t!]
            \DontPrintSemicolon  %
            \small 
            
            \KwIn{A censored dataset $\dataset$, propensity scores $\wh{e}\in B(e,\eps)$, and constants $\alpha,\beta>0$}
            \KwOut{An estimate $\tau_{\cS,N}$ of $\tau$}
            
            \SetKwFunction{FMain}{FindGoodLocalParition}
            \SetKwProg{Fn}{Function}{:}{}
             
\Fn{$\mathsf{FindGoodPartition}${($\dataset, \wh{e}, \alpha, \beta$})}{

                Split $\dataset$ into two datasets of equal size $\dataset_1$ and $\dataset_2$
                
                Initialize  $\cS=\emptyset$, $N=\emptyset$, and, $w=\nll$
                
                \For{~$(x,y,t)\in \dataset_1$ \textnormal{such that} ${\wh{e}(x)(1-\wh{e}(x))} < \sfrac{\beta}{3}$~}
                {
                    \If{~\textnormal{(the outlier)} $x$ \textnormal{has not already been covered}~}{
                        
                        Define an $\ell_\infty$-ball around $x:~$ $S=\inbrace{z\in \R^d\colon \norm{x-z}_\infty \leq \alpha}$
                        
                        Collect the non-outliers of $S$ {in $\dataset_1$:}~ $\hypo{T}=\inbrace{\inbrace{z}\colon z \in S{\cap \dataset_1} \text{ ~ and ~ } \wh{e}(z)(1-\wh{e}(z)) \geq \sfrac{\beta}{3}}$
                        
                        Update the weight $w$ based on:~ $\mathsf{UpdateWeight}$($S, \hypo{T}, \dataset_1, \wh{e}, w, \beta$)

                        Update the partition by adding $S$ and $\hypo{T}:~$ $\cS=\cS\cup\inbrace{S}\cup \hypo{T}$

                        Mark all covariates in the following set as covered $\inbrace{z\colon (z,y,t)\in \dataset_1 \text{ and } z\in S}$
                    }
                }

                \For{~\textnormal{all} $(x,y,t)\in \dataset_2$ \textnormal{where} $x$ \textnormal{has not been covered}~}{

                    \textbf{if}~~ $\wh{e}(x)(1-\wh{e}(x)) \geq \sfrac{\beta}{3}$ ~~\textbf{then}~~ add (the non-outlier) $\inbrace{x}$ to $\cS$\\ {\textbf{else}~~ add (the outlier) $x$ to $N$}
                }
                
                \KwRet $\tau_{\cS, N, w}(\dataset_2;\wh{e})$\;
                
            }
            
            \vspace{2mm}

            \Fn{$\mathsf{UpdateWeight}${($S, \hypo{T}, \dataset, \wh{e}, w, \beta$)}}{
                $\#$ The set $S$ is an $\ell_\infty$-ball around some outlier in $\dataset$
                
                $\#$ The set $\hypo{T}$ contains all singleton sets each corresponding to a non-outlier inside $S$
                \;

                Compute the fraction of outliers in ${\dataset}:~$ $\eta = \frac{1}{{\abs{\dataset}}}\sum_{(x,y,t)\in {\dataset}} \mathds{I}\insquare{\wh{e}(x)(1-\wh{e}(x)) < \beta/3}$
                \;
                
                Compute the fraction of outliers in ${S\cap \dataset}:~$ $\eta_S = \frac{1}{{\abs{S\cap \dataset}}}\sum_{(x,y,t)\in {S\cap \dataset }} \mathds{I}\insquare{\wh{e}(x)(1-\wh{e}(x)) < \beta/3}$
                \;

                \For{~\textnormal{all} $(x,y,t)\in \dataset$ \textnormal{such that} $x\in S$~}{
                    \uIf{~$\wh{e}(x)(1-\wh{e}(x)) < \beta/3$~}{
                        Set $w_S(x)=1,$ and, $w_T(x)=0$ for each $T\in \hypo{T}$
                    }
                    \Else{
                        \mbox{Set $w_S(x) = {{\small \min\hspace{0mm}}\inbrace{\hspace{-0.5mm}\frac{\eta_S+k^{-1}\eta}{1-\eta_S+k^{-1}\eta}, 1\hspace{-0.5mm}}}$,
                        \footnotemark{} and, $w_T(x) = 1{-}w_S(x)$ for the only set $T \in \hypo{T}$ containing $x$} 
                    }
                    
                }

                $\#$ Each $T \in \hypo{T}$ has zero weighted-mass of outliers

                $\#$ The set $S$ has a near-equal weighted mass of outliers and non-outliers.
                
                \KwRet $w$
            }

            \vspace{4mm}
            
            \caption{Algorithm From \cref{thm:main} (and \cref{infthm:1})}\label{algorithm}
            \end{algorithm}

        \paragraph{Stage 2 ($0\leq \eps\leq \sfrac{\beta}{10}$):}
                Stage 1 used propensity scores to (i) identify $\beta$-outliers and (ii) evaluate $\tau_{\cS', N'}(\dataset_2; e)$.
                Implementing (ii) with inaccurate propensity scores is not an issue as the $\eps$-Robust MSE of good-local partitions is well-behaved (\cref{lem:RobustMSEofLGpartition}).
                To implement (i), we show the following (\cref{fact:outliers}):
                for any $\eps\leq \sfrac{\beta}{10}$, 
                \[
                \forall\ {\wh{e}\in B(e, \eps)}\,,\quad 
                 \frac{\beta}{9}\text{-outliers}
                 ~~\subseteq~~ 
                 \inbrace{x\in \R^d\colon \wh{e}(x)(1-\wh{e}(x))\leq \frac{\beta}{3}} 
                 ~~\subseteq~~
                 \beta\text{-outliers}\,.
            \] 
            Hence, given any $\eps$-inaccurate propensity score $\wh{e}$, one can approximately identify the outliers, which we show is sufficient to ensure guarantees of Stage 1 up to constant factors (\cref{lem:algorithm:step3}).

            \footnotetext{We shortly mention that the number of $\ell_\infty$-balls $k$ used in our cover is not given to the algorithm, but our $\mathsf{UpdateWeight}$ routine uses it. However, there is an easy adaptation of this algorithm that finds $k$ which can be run as a pre-processing step: we run $\mathsf{FindGoodPartition}$ by omitting the calls to $\mathsf{UpdateWeight}$ and then $k$ will correspond to the number of iterations of the (first) outer \textbf{for} loop. Then we execute our algorithm with $k$ as input.}

        \paragraph{Stage 3 (Improving Guarantee on Bias from $\alpha L$ to $\rho\alpha L$):}
                Since $\cD(N')=O(1/n)$ and the diameter of all sets in $\cS'$ is $2\alpha$, an upper bound on the bias is as follows (see \cref{eq:simpleUpperbound}) 
                \[ 
                    \bias_\cD(\tau_{\cS', N'})
                    \leq 
                    O\inparen{\frac{1}{n}}
                    +O(\alpha L)\cdot \sum_{S\in \cS'\colon \diam(S)>0}\cD(S) %
                    \,.  
                \]
                To get an upper bound of $O(\rho \alpha L+(\sfrac{1}{{n}}))$, we need to ensure that the total mass of sets $S$ for which $\diam{(S)}>0$ is at most $\rho$.
                However, even under \cref{asmp:lipschitzness,asmp:sparsity,asmp:isolation}, there may be no partition satisfying this guarantee.
                    To see this suppose there is a single $\beta$-outlier,  $x_{\rm out}$, with mass $\cD(x_{\rm out})=\rho$.
                    Let $x_{\rm in}$ be a non-outlier close to $x_{\rm out}$ with mass $\cD(x_{\rm in})\gg \rho$.
                Since $x_{\rm in}$ is a point-mass, any $S$ covering $x_{\rm out}$ either covers $x_{\rm in}$ or does not cover $x_{\rm in}$.
                In the former case, then 
                $\cD(S)=\cD(x_{\rm in})+\cD(x_{\rm out})\gg \rho,$ 
                and, hence, we do not satisfy $\cD(S)\leq \rho$. %
                In the latter case, $e(S)(1-e(S))\leq \beta$ and, hence it cannot be a part of a good-local partition.
                We side-step the issue by considering \textit{weighted} partitions.
                Intuitively, in this case, one can construct two sets $S,S'$ and set $w_S(x_{\rm out})=1$ and $w_S(x_{\rm in})=\sfrac{\rho}{\cD(x_{\rm in})}$.
                This corresponds to the $\mathsf{UpdateWeight}$ routine in \cref{algorithm}.
                
    \subsection*{Additional Remarks}
        {Next, we present some additional remarks on \cref{asmp:lipschitzness,asmp:isolation,asmp:sparsity}.

        Our algorithm for estimating $\tau$ requires the knowledge of parameters in \cref{asmp:lipschitzness,asmp:isolation,asmp:sparsity}.
        Some of these parameters (those in \cref{asmp:isolation,asmp:sparsity}) can be estimated from data: given $\epsilon$ accurate propensity scores $\wh{e}$, one can check if \cref{asmp:isolation,asmp:sparsity} hold for all $x$ with $\wh{e}(x)(1-\wh{e}(x)) < 2\beta+\epsilon$. 
        This is sufficient as \cref{fact:propensityScores} implies that any $\beta$-outlier must satisfy $\wh{e}(x)(1-\wh{e}(x)) < 2\beta+\epsilon$.
        Estimating the Lipschitzness constant (i.e., testing \cref{asmp:lipschitzness}) requires access to $\mu_t(x)$ which may be unavailable in practice.
        However, if $\mu_t(x)$ is a parametric function of $x$ (a common assumption in practice \citep{imbens2015causal,hernan2023causal}), then \cref{asmp:lipschitzness} holds under mild assumptions on the underlying parameter; see discussion after \cref{asmp:lipschitzness}.

        \cref{asmp:lipschitzness} is required to ensure that CIPW estimators defined by good-local partitions have a small robust RMSE.
        \cref{asmp:sparsity} is a \textit{necessary} condition to ensure the existence of a good-local partition.
        Given \cref{asmp:sparsity} holds, \cref{asmp:isolation} ensures (i) the existence of a fractional good-local partition with (possibly) overlapping partitions and (ii) an efficient algorithm to find it. 
        While \cref{asmp:isolation} is not necessary to ensure the existence of a (fractional) good-local partition, to be able to identify the good-local partition from finite samples, some assumption in addition to \cref{asmp:sparsity} is necessary (see \cref{lem:impossibility_of_learning_gl_partition_intro}).
        That said, alternatives to \cref{asmp:isolation} may also be able to ensure (i) and (ii).
        For instance, if the outliers lie in \textit{known} affine subspaces and non-outliers have a ``sufficient'' mass close to these subspaces, then any covering of these subspaces by $\ell_p$-balls of diameter $\alpha$ is a good-local partition, and the corresponding CIPW estimator satisfies \cref{thm:main}'s guarantee.}

\section{Expressions of Bias and Variance of CIPW Estimators}\label{sec:expression_of_mse}
    In this section, compute the expressions of the bias and variance of CIPW estimators.
    These, in turn, imply expressions for the RMSE and MSE of CIPW estimators. %
    \begin{restatable}[{Bias and Variance of CIPW Estimators}]{theorem}{expOfMSE}\label{lem:exp_of_mse}%
            Fix any $n\geq 1$ and assume that Equation~\eqref{eq:unconfounded} holds (unconfoundedness). 
            For any partition $(\cS,N)$ of $\R^d$,  censored dataset $\dataset$ of size $n$ generated from $\cD$ with feature marginal $\cD_X$,  propensity scores $e = \{e(x) \}_{x \in \R^d}$ and conditional means/variances $\mu,v = \{\mu_0(x),\mu_1(x),v_0(x),v_1(x)\}_{x \in \R^d}$
            it holds that
            \[
                \Ex_{\dataset}\insquare{
                    \tau_{\cS, N}(\dataset;(e,\mu)) 
                }
                    = 
                    \sum_x
                        \cD_X(x) 
                        \left(
                        \frac{e(x)\mu_1(x)}{e(S_x)}
                        - \frac{\inparen{1-e(x)}\mu_0(x)}{1-e(S_x)}
                        \right
                                  )
                        \,,
            \]
            where $S_x$ is the unique $S\in \cS$ containing $x$.
            Moreover, it holds that
            the bias term $\mathrm{Bias}(\tau_{\cS, N}(\dataset;(e,\mu))) $ is equal to
            \begin{align*}
                    \abs{
                        \tau - 
                        \Ex_{\dataset}\insquare{
                            \tau_{\cS, N}(\dataset;(e,\mu)) 
                         }
                    }\,,
                \end{align*}
                and the variance term
                $\mathrm{Var}(\tau_{\cS, N}(\dataset;(e,\mu,v)))$
                is equal to
                \begin{align*}
                    \frac{1}{n}
                    \sum_x
                        \cD_X(x)\inparen{
                            \frac{e(x)\inparen{v_1(x)+\mu_1(x)^2}}{e(S_x)^2}
                            + 
                            \frac{\inparen{1 - e(x)}\inparen{v_0(x)+\mu_0(x)^2}}{\inparen{1 - e(S_x)}^2}
                        }
                    - {
                        \Ex\insquare{
                            \tau_{\cS, N}(\dataset;(e,\mu)) 
                         }
                      }^2\,.
            \end{align*}
        \end{restatable}
    \noindent 
    Let us comment on the above expressions.
    The propensity scores $e(x)$ and $e(S)$ are observed quantities.
    While $\mu_0(x),\mu_1(x),v_0(x),$ and $v_1(x)$ are not observed due to censoring, under unconfoundedness, given a sufficiently large dataset, they can be computed to arbitrary accuracy (
    e.g., by using that $\mu_1(x)=\Ex\insquare{Y(1)}=\Ex[Y(1)\mid T=1]$).
    Hence, under unconfoundedness, \cref{lem:exp_of_mse} gives us expressions for bias, variance, and MSE in terms of quantities that can be estimated from a censored dataset.
    Furthermore, using the Central Limit Theorem, whenever $e(S)(1-e(S))$ is bounded away from 0, one can also obtain confidence intervals for the bias and variance and, hence, RMSE. (This analysis is analogous to that in \cref{sec:asymptoticNormality} which establishes asymptotic normality of $\tau_{\cS, N}$ when the coarse propensity scores $e(S)$ for each $S\in \cS$ are bounded away from $0$.)

    In the remainder of this section, we prove \cref{lem:exp_of_mse}.

            \label{sec:proofof:lem:exp_of_mse}

                \medskip
                
                \begin{proof}\proofof{\cref{lem:exp_of_mse}}
                    Recall that, given $n$ samples $C\coloneqq \inbrace{\inparen{x_i,y_i,t_i}}_i$ and a partition $\cP=(\cS, N)$
                \[
                    \tau_{\cP} = \frac{1}{n} \sum_{i=1}^n \phi(x_i,y_i,t_i)\,,
                    \quadtext{where}
                    \phi(x,y,t) = \frac{t y}{e(X; \cP)} - \frac{(1-t)y}{1-e(X; \cP)}\,.
                \]
                Since $\var(X) = \Ex[X^2]- \Ex[X]^2$ for a random variable $X$, simple algebraic manipulation gives that
                \begin{align*}
                    \Ex_{C\gets (\cD,A)^n}\insquare{\inparen{\tau_{\cP}-\tau}^2}
                    &= \inparen{\tau - \Ex_{C\gets (\cD,A)^n}\insquare{\tau_{\cP}}}^2
                    + \var_{C\gets (\cD,A)^n}\insquare{\tau_{\cP}}.
                \end{align*}
                Since $(\cD,A)$ is individualistic
                \begin{align*}
                    \Ex_{C\gets (\cD,A)^n}\insquare{\tau_{\cP}} = \Ex_{(X,Y,T)\gets (\cD,A)}\insquare{\phi(X,Y,T)}
                     \quadand  
                    \var_{C\gets (\cD,A)^n}\insquare{\tau_{\cP}} = \frac{1}{n}\var_{(X,Y,T)\gets (\cD,A)}\insquare{\phi(X,Y,T)}\,.
                \end{align*}
                Moreover, due to unconfoundedness, for any $z$ in the feature domain,
                \[
                    \Ex_{(X,Y,T)\gets (\cD,A)}\insquare{\phi(X,Y,T)\mid X=z} = \psi(z) \coloneqq \frac{e(z)\mu_1(z)}{e(z; \cP)} - \frac{(1-e(z))\mu_0(z)}{1-e(z; \cP)}\,.
                    \yesnum\label{eq:expectationWRTX}
                \]
                Recall that, in the above, we take $\mu_1(z) = \Ex\insquare{YT\mid X = z} = \Ex\insquare{Y^1\mid X=z}$ and $\mu_0(z) = \Ex[Y(1-T) \mid X=z] = \Ex\insquare{Y^0 \mid X=z}$, where the equalities are due to unconfoundedness.
                
                Henceforth, we omit the subscript $(X, Y, T)\gets (\cD, A)$ when it is clear from context. 
    
                Given the above expressions, we can now compute $\var\insquare{\phi(X, Y, T)}$ as follows. First, via the law of total variance, i.e., using that $\var(B)=\var[\Ex(B\mid A)]+\Ex[\var(B\mid A)]$, we have that
                \[
                \var\insquare{\phi(X,Y,T)} =
                \var_X\insquare{\Ex_{Y,T}\inparen{\phi(X,Y,T) \mid X}}
                    +\Ex_X\insquare{\var_{Y,T}\inparen{\phi(X,Y,T) \mid X}}\,.
                \]
                We next consider the second term.
                Using that $\var(A+B)=\var(A)+\var(B)+2\cov(A,B)$ with 
                $A = TY/e(X; \cP)$ and $B = (T-1)Y/(1-e(X; \cP))$, \mbox{we get that $\Ex_X\insquare{\var_{Y,T}\inparen{\phi(X,Y,T) \mid X}}$ is equal to}
                \[
                \Ex_X\insquare{
                        \var_{Y,T}\inparen{\frac{TY}{e(X; \cP)} \mid X}
                        +\var_{Y,T}\inparen{\frac{{-}(1-T)Y}{1-e(X; \cP)} \mid X}
                        +2\cov_{Y,T}\inparen{\frac{TY}{e(X; \cP)},\; \frac{{-}(1-T)Y}{1-e(X; \cP)} \mid X}}\,.
                \]
                Now by definition, it holds that
                $\cov(A,B)=\Ex[AB]-\Ex[A]\Ex[B]$ and since $T(1-T)=0$, we can simplify the covariance expression as
                \[
                \cov_{Y,T}\inparen{\frac{TY}{e(X; \cP)},\; \frac{{-}(1-T)Y}{1-e(X; \cP)} \mid X}
                =
                \Ex_{Y,T}\inparen{\frac{TY}{e(X; \cP)} \mid X}
                        \Ex_{Y,T}\inparen{\frac{(1-T)Y}{1-e(X; \cP)} \mid X}\,.
                \]
                Hence, in total, we have that $\var\insquare{\phi(X,Y,T)} =
            \var_X\insquare{\Ex_{Y,T}\inparen{\phi(X,Y,T) \mid X}} + F$, where
            \[
                \smallmath{
                F = \Ex_X\insquare{
                        \var_{Y,T}\inparen{\frac{TY}{e(X; \cP)} \mid X}
                        +\var_{Y,T}\inparen{\frac{(1-T)Y}{1-e(X; \cP)} \mid X}
                        {+}2\Ex_{Y,T}\inparen{\frac{TY}{e(X; \cP)} \mid X}
                        \Ex_{Y,T}\inparen{\frac{(1-T)Y}{1-e(X; \cP)} \mid X}}\,.
                }
            \]            
    We are now ready to compute the three terms of $F$. For this, we will make use of the fact that $T \in \{0,1\}$ and unconfoundedness.
    Note that for $W \in \zo$, 
    $\var[W B] = \Ex[W B^2]-\inparen{\Ex[W B]}^2$ and this means that, since $e(X) = \Pr[T=1 | X]$ and, e.g., $\Ex[(Y^1)^2 | X] = v_1(X) + \mu_1(X)^2$,
    we get
    \begin{align*}
        \smallmath{\Ex_X\insquare{
            \var_{Y,T}\inparen{\frac{TY}{e(X; \cP)} \mid X}}}
            &\smallmath{=\Ex_X\insquare{
            \frac{e(X)(v_1(X)+\mu_1(X)^2) 
            -
            e(X)^2\mu_1(X)^2}{e(X; \cP)^2} }}
            \,,\\
        \smallmath{\Ex_X\insquare{
            \var_{Y,T}\inparen{\frac{(1-T)Y}{1-e(X; \cP)} \mid X}}} 
            &
            \smallmath{=\Ex_X\insquare{\frac{(1-e(X))(v_0(X)+\mu_0(X)^2) - (1-e(X))^2\mu_0(X)^2}{(1-e(X; \cP))^2} }}
            \,,\,\text{\small and},\\
        \smallmath{\Ex_X\insquare{
    \Ex_{Y,T}\inparen{\frac{TY}{e(X; \cP)} \mid X}
                        \Ex_{Y,T}\inparen{\frac{(1-T)Y}{1-e(X; \cP)} \mid X}}}
                        &\smallmath{=
                        \Ex_X\insquare{
                        \frac{e(X)\mu_1(X)}{e(X; \cP)} 
                        \frac{(1-e(X))\mu_0(X)}{1-e(X; \cP)}}}\,.
    \end{align*}
    We now aggregate the above terms and have that $F$ reads as
    \begin{align*}
        F 
        &= 
        \Ex_X\insquare{
                        \frac{e(X)(v_1(X)+\mu_1(X)^2)}{(e(X; \cP))^2} } + \Ex_X\insquare{\frac{(1-e(X))(v_0(X)+\mu_0(X)^2)}{(1-e(X; \cP))^2} }\\
        &\quad 
        - 
                        \Ex_X\insquare{
                        \inparen{
                            \frac{e(X)\mu_1(X)}{e(X; \cP)} - \frac{(1-e(X))\mu_0(X)}{1-e(X; \cP)}
                        }^2
                    }\,.
                    \yesnum\label{eq:second_expression_of_f}
    \end{align*}
    Recall that $\var\insquare{\phi(X,Y,T)} =
            \var_X\insquare{\Ex_{Y,T}\inparen{\phi(X,Y,T) \mid X}} + F$ and $\psi(z)=\Ex_{Y,T}\insquare{\phi(X,Y,T) \mid X=z}$ for all $z$ in the feature domain and, hence, $\var\insquare{\phi(X,Y,T)} =
            \var_X\insquare{\psi(X)} + F$.
        Subsequently, observing that the first term of $F$ in \cref{eq:second_expression_of_f} is $-\Ex_X\insquare{\psi(X)^2}$ it follows that
        \begin{align*}
            \var\insquare{\phi(X,Y,T)} 
            &=\var_X\insquare{\psi(X)}
        +\Ex_X\insquare{
                        \frac{e(X)(v_1(X)+\mu_1(X)^2)}{(e(X; \cP))^2} } + \Ex_X\insquare{\frac{(1-e(X))(v_0(X)+\mu_0(X)^2)}{(1-e(X; \cP))^2} }\\
        &\quad 
        - 
                        \Ex_X\insquare{
                        {
                            \psi(X)
                        }^2
                    }\,.
        \end{align*}
        Now, since $\var[A]=\Ex[A^2]-\inparen{\Ex[A]}^2$, we can simplify the above as 
        
        \begin{align*}
            \var\insquare{\phi(X,Y,T)} 
            &= \Ex_X\insquare{
                            \frac{e(X)(v_1(X)+\mu_1(X)^2)}{(e(X; \cP))^2} } + \Ex_X\insquare{\frac{(1-e(X))(v_0(X)+\mu_0(X)^2)}{(1-e(X; \cP))^2} } - \inparen{\Ex_X\insquare{\psi(X)}}^2 
            \,.
        \end{align*}
        This completes the computation.
    \end{proof}
    
    \begin{remark}[{Recovering the Standard IPW Estimator}]
        As a sanity check, by picking $f$  to be 1-1, we can recover the variance of the IPW estimator. For simplicity, we will work under the assumption that, for all $z$ in the feature domain, 
            $v_1(z)=v_0(z).$
            Under this assumption, \citep{wager2020notes} claims that the variance of the IPW estimator is 
            \[
                \var\insquare{\ipw{}}
                = \var_X\insquare{\tau(X)}
                + \Ex_X\insquare{\frac{v_0(X)}{e(X)(1-e(X))}}
                +
                \Ex_X\insquare{
                    \frac{
                        \inparen{\mu_1(X)-e(X)\inparen{
                            \mu_1(X)-\mu_0(X)
                        }}^2
                    }{e(X)(1-e(X))}
                }\,.
            \]
            To simplify the above expression, observe that 
            ${\mu_1(z)-e(z)\inparen{
                            \mu_1(z)-\mu_0(z)
                        }}=\mu_1(z)(1-e(z))+\mu_0(z)e(z),$ and, hence,
                        for all $z$ in the feature domain
            \begin{align*}
                \frac{
                        \inparen{\mu_1(z)-e(z)\inparen{
                            \mu_1(z)-\mu_0(z)
                        }}^2
                    }{e(z)(1-e(z))}
                &=
                    \frac{\mu_1(z)^2(1-e(z))}{e(z)}
                    +\frac{\mu_0(z)^2e(z)}{1-e(z)}
                    +2\mu_1(z)\mu_0(z)\\
                &=
                    \frac{\mu_1(z)^2}{e(z)}
                    +\frac{\mu_0(z)^2}{1-e(z)}
                    -\inparen{\mu_1(z) - \mu_0(z)}^2.
            \end{align*}
            Therefore, it follows that
            \[
                \var\insquare{\ipw{}}
                = \var_X\insquare{\tau(X)}
                + \Ex_X\insquare{\frac{v_0(X)}{e(X)(1-e(X))}}
                +
                \Ex_X\insquare{ 
                    \frac{\mu_1(X)^2}{e(X)}
                    +\frac{\mu_0(X)^2}{1-e(X)} 
                }
                -\Ex_X\insquare{\inparen{\mu_1(X) - \mu_0(X)}^2}\,.
            \] 
            Under the above assumption, we should get the same value for $\var\insquare{\phi(X,Y,T)}$ when $f$ is 1-1.
            Note that in this case $e(z; \cP)=e(z)$ and $v_0(z)=v_1(z)$ for all $z$ in the feature domain.
            Moreover, these imply that $\psi(z)=\tau(z)$ for all $z$ in the feature domain.
            Therefore
            \[
                \var\insquare{\phi(X,Y,T)}
                =
                \Ex_X\insquare{
                    \frac{v_0(X)}{e(X)(1-e(X))}
                }
                +\Ex_X\insquare{
                    \frac{\mu_1(X)^2}{e(X)}
                    +\frac{\mu_0(X)^2}{1-e(X)}
                }
                - \inparen{\Ex_X\insquare{\tau(X)}}^2\,.
            \] 
            It follows that
            \begin{align*}
                \var\insquare{\phi(X,Y,T)}
                - \var\insquare{\ipw{}}
                &=
                -\inparen{\Ex_X\insquare{\tau(X)}}^2
                -\var_X\insquare{\tau(X)}
                +\Ex_X\insquare{\inparen{\mu_1(X)-\mu_0(X)}^2}\,.
            \end{align*}
            This is zero as $\mu_1(z)-\mu_0(z)=\tau(z)$ and $\var[W]=\Ex[W^2]-(\Ex[W])^2$.
    \end{remark}

\section{Formal Statements and Proofs of Hardness Results}
\subsection{Computational Hardness of Finding Minimum-RMSE CIPW Estimator}\label{proofOverivews:hardness}
        {In this section, we consider the computational complexity of solving $\argmin_{\cS, N}\; 
                                            \RMSE_{\cD}{(\tau_{\cS, N})}$.}
        Concretely, we consider the following decision version of it which considers distributions $\cD$ supported over $m$ covariates $x_1,x_2,\dots,x_m.$
        \begin{mdframed}[backgroundcolor=gray!5]
            \textbf{\problemRMSE{}$(U, n, e, \cD, \mu, v)$.}
            \begin{itemize}[leftmargin=15pt]
                \item \textbf{Input:} 
                Threshold $U \geq \Q_{+}$, 
                Integers $n, m \geq 1$, and for each $i\in[m]$, 
                    rationals $e(x_i),
                    \cD(x_i) \in \Q_+$ (s.t., $\sum_{i=1}^m \cD(x_i)=1$), 
                    $\mu_0(x_i),\mu_1(x_i)\in [-1,1]\cap \Q$, and 
                    $v_0(x_i),v_1(x_i)\in [0,1] \cap \Q$.
                \item \textbf{Output:} \textbf{\texttt{Yes}} if $\min_{\cS, N} \RMSE_{\cD}{(\tau_{\cS, N})}^2 = \min_{\cS, N} \MSE_{\cD}{(\tau_{\cS, N})} \leq U$ and \textbf{\texttt{No}} otherwise.
            \end{itemize}
        \end{mdframed}
        \noindent {Since it is more convenient to work with the MSE instead of the RMSE, above we phrase the condition on the square of the RMSE.}
        Under unconfoundedness all of the inputs to \problemRMSE{} are identifiable and, hence, can be estimated to arbitrary precision given sufficiently many samples.
        Since we are interested in the computational complexity, we assume all inputs are known exactly for now.
        Our main hardness result is as follows.
        
        \hardnessApproximation*
        
        \noindent The proof of \cref{thm:hardness:approximation} appears in \cref{sec:proofof:thm:hardness}.
        We continue with a technical overview below. %

        \subsubsection{Technical Overview} 
        {The hardness result follows from a reduction of \problemSS{} to the special case of \problemRMSE{} where $\mu_0(x)=v_0(x)=0$ and $v_1(x)+\mu_1(x)^2=C>0$ for all $x$ and some large constant $C$.
    
        Consider an instance $a_1,a_2,\dots,a_m$ and $V$ of \problemSS{}. The first idea in the hardness result is to select $n$ and $U$ such that such that any solution $(\cS, N)$ of $\problemRMSE{}$ must ensure that $\tau_{\cS, N}$ is unbiased and $N=\emptyset$.
        To see why this is relevant, consider the following equality, which follows from the expression of CIPW estimators bias (see \cref{lem:exp_of_mse})
        \[
                \mathbb{E}_{\dataset}[\tau_{\cS, N}(\dataset;e)] =  \sum_x
                            \frac{\cD_X(x)}{1-\cD(N)}
                            \left(
                            \frac{e(x)\mu_1(x)}{e(S_x)}
                            - \frac{\inparen{1-e(x)}\mu_0(x)}{1-e(S_x)}
                            \right
                                      )
                            \,,
                \]
        where $S_x$ is the unique $S\in \cS$ containing $x$. %
        This implies that $\tau_{\cS, N}$ is unbiased if and only if
        \[
              \tau\cdot (1-\cD(N))=\sum_{x\not\in N}
                        \cD_X(x) \tau_{\cS, N}(x; (e,\mu))\,,
        \]
        where $\tau_{\cS, N}(x; (e,\mu) ) = 
                            \frac{e(x)\mu_1(x)}{e(S_x)}
                            - \frac{\inparen{1-e(x)}\mu_0(x)}{1-e(S_x)}.$ 
        If we can ensure $N=\emptyset$, then this simplifies to
        \[
             \tau=\sum_{x}
                        \cD_X(x)\cdot  \tau_{\cS, N}(x; (e,\mu))\,.
        \] 
        Hence, selecting $\tau\propto V$ and $\cD_X(x_i)\tau_{\cS, N}(x_i; (e,\mu)) \propto a_i$ for all $1\leq i\leq m$, implies that $\tau_{\cS, N}$ is unbiased if and only if $\sum_{x\in [m]\backslash N} a_i=V.$
        
        The difficulty is that $\tau$ is dependent on the values of $\cD_X(x)$ and $\tau_{\cS, N}(x; (e,\mu))$ and, hence, cannot take an arbitrary value.
        To avoid this, one idea is to introduce a \toa{dummy} element $x_{m+1}$ which must be included in $N$ (the $\nll$ set) in any solution of \problemRMSE{}: this is useful as (1) we can set $\tau$ to be an arbitrary value by selecting an appropriate $\cD_X(x_{m+1})\tau_{\cS, N}(x_{m+1}; (\eps, \mu))$ and (2) since $x_{m+1}\in N$, it is irrelevant for the sum $\sum_{x\not\in N}
                        \cD_X(x) \tau_{\cS, N}(x; (\eps, \mu))$ which shows up in the bias.
        To be more precise, since $N\neq\emptyset$, we need to account for the factor $(1-\cD(N))$. We can side-step this by ensuring $1-\cD(N)\approx 1 - \cD_X(x_{m+1})\approx 1$.
    
        One idea to ensure that in any feasible solution satisfies $x_{m+1}\in N$ is to set a very propensity score for $x_{m+1}$, say, $e(x_{m+1})=\eps\ll 1$.
        This ensures that if $x_{m+1}\not\in N$ and if $x_{m+1}$ is also not merged with any other point, then the variance is at least $\Omega\inparen{\sfrac{1}{\eps}}$.
        However, it allows for $x_{m+1}\not\in N$, if $x_{m+1}$ is merged with other points.
        We ensure that merging $x_{m+1}$ with other points is not desirable by adding another dummy element $x_{m+2}$.
        We construct $x_{m+2}$ so that: 
        \begin{enumerate}
            \item $\mu_1(x_{m+2})-\mu_2(x_{m+1})$ is large (hence, merging $x_{m+1}$ and $x_{m+2}$ leads to large bias); and 
            \item $\cD_X(x_{m+2}) \gg \cD_X(x_{m+1})\gg \cD_X(x)$ for any other $x$ (this ensures that if $x_{m+1}$ is not merged with $x_{m+2}$, then the variance remains $\Omega(\poly(\sfrac{1}{{\eps}}))$).
        \end{enumerate}
        Finally, to prove the hardness of approximation, we extend the above reduction.
        Given any instance $\instanceSS{}$ of \problemSS{} with bit-complexity $b$ and a number $\eta>0$, we construct an instance \instanceMSE{} of \problemRMSE{} such that:
        \begin{enumerate}
            \item \instanceMSE{}'s bit complexity at most $b+O(\log\inparen{\sfrac{1}{\eta}})$ (i.e., polynomial in the input $\inangle{\instanceSS{}, \eta}$); and
            \item The MSE of the optimal and non-optimal solutions are separated by a multiplicative factor of $\sfrac{1}{\eta}$: concretely, let $(\cS^\star, N^\star)$ be any minimizer of $\MSE{}$ and let $(\cS', N')$ by any non-optimal solution, then we show that $\MSE{(\tau_{\cS^\star, N^\star})}\leq O(\eta)\cdot \MSE{(\tau_{\cS', N'})}$. 
        \end{enumerate}}

    \subsubsection{Proof of \cref{thm:hardness:approximation}}\label{sec:proofof:thm:hardness}
        In this section, we prove \cref{thm:hardness:approximation}.
        
    \subsubsection*{Reduction}
    To prove \cref{thm:hardness:approximation}, we construct a gap-inducing reduction from \problemSS{} to the special case of \problemRMSE{} where 
    \[
        \mu_0(x)=v_0(x)=0 \quad 
        \text{for all } x\,.
    \]
    In this special case of \problemRMSE{}, the expression of the MSE (see \cref{lem:exp_of_mse}) simplifies to:
    \begin{align*}
        \MSE{(\tau_{\cS, N})} 
        &= \inparen{
            \sum_{S\in \cS} \frac{\sum_{x\in S} e(x)\cD(x)\mu_1(x)}{
                e(S)\inparen{1-\cD(N)}
            } 
            - \sum_x \cD(x)\mu_1(x)
        }^2\\
        &+
        \frac{1}{n}\sum_{S\in \cS}
                \frac{
                    \sum_{x\in S} e(x) \cD(x)\inparen{v_1(x)+\mu_1(x)^2}
                }{
                    e(S)^2\inparen{1-\cD(N)}
                }
        -\frac{1}{n}\inparen{\sum_{S\in \cS} \frac{\sum_{x\in S} e(x)\mu_1(x)\cD(x)}{e(S)(1-\cD(N))}}^2
        \,.
    \end{align*}
    Where $e(S)$ is the average propensity score on $S$ and $\cD(S)$ is the total probability mass on $S$:
    \[
        e(S)\coloneqq \frac{\sum_{x\in S} e(x)\cD(x)}{\cD(S)}
         \quadand  
        \cD(S)\coloneqq \sum_{x\in S} \cD(x)\,.
    \]
    Thus, $\min_{\cS, N}\MSE_\cD(\tau_{\cS, N})$ is equivalent to the following program
    \[
        \text{\small$\min_{1\leq k\leq m}\quad 
        \min_{
            \substack{
                N\subseteq [m],\\
                \cS = \inbrace{S_1,\dots,S_k} \text{ partitions}\\
                \text{$[m]\backslash N$ into $k$ non-empty sets}
            }
        }\quad 
        \left(\begin{array}{c}
             \inparen{
                \sum_{S\in \cS} \frac{\sum_{x\in S} e(x)\cD(x)\mu_1(x)}{
                    e(S)\inparen{1-\cD(N)}
                } 
                - \sum_x \cD(x)\mu_1(x)
            }^2\\ 
            +
            \frac{1}{n}\sum_{S\in \cS}
                \frac{
                    \sum_{x\in S} e(x) \cD(x) \inparen{v_1(x)+\mu_1(x)^2}
                }{
                    e(S)^2\inparen{1-\cD(N)}
                }
            -\frac{1}{n}\inparen{\sum_{S\in \cS} \frac{\sum_{x\in S} e(x)\mu_1(x)\cD(x)}{e(S)(1-\cD(N))}}^2
        \end{array}\right)
        .$}
    \]

    \paragraph{Reduction.}
        Recall that \problemSS{} is defined as follows and is known to be $\mathsf{NP}$-hard \citep{gareyjohnson1990}.
        \begin{mdframed}[backgroundcolor=gray!5]
            \textbf{\problemSS{}.}
            \begin{itemize}
                \item Input: Positive integers $a_1,a_2,\dots,a_k$ and $T$
                \item Output: \texttt{Yes} if there exists $S\subseteq[k]$ such that $\sum_{i\in S}a_i = T$ and \texttt{No} otherwise.
            \end{itemize}
        \end{mdframed}
    Given an instance $\cI_{SS}(a,T)$ of \problemSS{}, define constants 
    \[
        A \coloneqq \sum_i a_i\,,\quad 
        {\eps \ll \frac{1}{(Am)^4}\,,} \quad
        \alpha = \eps^3\,,\quad 
        \beta = \eps^5\,,\quadand
        \Delta = \frac{1}{1+\eps+k\eps^3}\,.
    \]
    We construct an instance $\cI_{\rm MSE}(m,n,e,\mu,v,\cD,U)$ of \problemRMSE{} with 
    \[
        m = k+2\,,\quad 
        n = \eps^{-7}\,,\quadand
        U = {8A^2\eps^7}\,. %
    \]
    The first $k$ out of $k+2$ points are defined as follows: 
    for each $1\leq i\leq k$
    \[
        \cD(x_i) = \Delta\alpha\,,\quad 
        e(x_i) = 1\,,\quad 
        \mu_1(x_i) = a_i\,, \quadand  
        v_1(x_i) = 4A^2 - \mu_{1}(x_i)^2\,.
        \yesnum\label{eq:reduction:construction:firstkPt}
    \]
    The last two points are defined as follows:
    \begin{align*}
        &\cD(x_{m-1}) 
            = \Delta\,,\  
        e(x_{m-1}) = 1\,,\ 
        \mu_1(x_{m-1}) = 2A\alpha\,,
        \ \text{and}\  
        v_1(x_{m-1}) = {4A^2 - \mu_{1}(x_{m-1})^2}\,;
        \yesnum\label{eq:reduction:construction:secondLastPt}\\
        &\cD(x_m) 
            = \Delta\eps\,,\  
        e(x_m) = {\beta}\,,\ 
        \mu_1(x_m) = \inparen{\frac{T+2A}{\Delta} - 3A}\frac{\alpha}{\eps}\,, \ \text{and}\ 
        v_1(x_m) = 4A^2 - \mu_{1}(x_m)^2\,.
        \yesnum\label{eq:reduction:construction:lastPt}
    \end{align*}
    Recall that for all points $x$, we set $\mu_0(x)=v_0(x)=0$.
    This completes the construction of \instanceMSE{}.
    Clearly, \instanceMSE{} can be constructed in time polynomial in the bit complexity of \instanceSS{}.
    For \instanceMSE{} to be a valid instance of \problemRMSE{}, we need to ensure that $\mu_1,v_1\in\insquare{-1,1}$, which can be done by dividing $\mu_1$ by $A$ and $v_1$ and $U$ by $A^2$. 
    (Intuitively, this does not change the proof because $\MSE{(\tau_{\cS, N})}/U$ is invariant to scaling of the outcomes and $\sqrt{U}$.)

    \subsubsection*{Soundness} 
        Next, we prove that the reduction is sound.
    \begin{lemma}[{Soundness}]\label{lem:reduction:soundness}
        Consider any instance $\cI_{SS}(a,T)$ of \problemSS{} and the corresponding instance $\cI_{\rm MSE}(m,n,e,\mu,v,\cD,U)$ of \problemRMSE{}.
        For any partition $\inparen{\cS,N}$, the following holds
        \begin{enumerate}
            \item If there is $1\leq i\leq k$ such that $S_i = \inbrace{x_{m}}$, then $\MSE{\inparen{\tau_{\cS,N}}} > \frac{U}{\eps}$;
            \item If there is $1\leq i\leq k$ such that $x_m\in S_i$, $\abs{S_i}\geq 2$, and $\MSE{\inparen{\tau_{\cS,N}}} \leq \frac{U}{\sqrt{\eps}}$, then \instanceSS{} is a \texttt{Yes} instance;
            \item If there is $1\leq i\leq k$ such that $x_m\in N$ and $\MSE{\inparen{\tau_{\cS,N}}} \leq \frac{U}{\sqrt{\eps}}$, then \instanceSS{} is a \texttt{Yes} instance.
        \end{enumerate}
    \end{lemma}
    To see how this implies soundness, observe that due to the first part of the above lemma, if \instanceMSE{} is a \texttt{Yes} instance with certificate $\inparen{\cS, N}$, then a special element $x_m$ satisfies: $x_m\in N$ or $x_m\in S_i$ with $\abs{S_i}\geq 2$.
    The second and third parts imply that if $x_m$ satisfies either of these conditions, then soundness holds.
    Therefore
    \[
        \instanceMSE{} \text{ is a \texttt{Yes} instance}
        \quad\implies\quad
        \instanceSS{} \text{ is a \texttt{Yes} instance}\,.
    \]
    \begin{proof}\proofof{\cref{lem:reduction:soundness}}
        The proof is divided into three cases, corresponding to the three parts of the statement. 

        \bigskip 

        \noindent \textit{Part 1 ($S_i=\inbrace{x_m}$):}
            Recall that $\MSE{\inparen{\tau_{\cS,N}}}\geq \variance\inparen{\tau_{\cS,N}}$.
            Our goal is to lower bound $\variance\inparen{\tau_{\cS,N}}.$
            We begin with the following equality:
            \begin{align*}
                \variance\inparen{\tau_{\cS,N}}
                &= \frac{1}{n} \frac{
                    \sum_{x\in S_i} e(x) \cD(x)\inparen{v_1(x)+\mu_1(x)^2}
                }{
                    e(S_i)^2\inparen{1-\cD(N)}
                }
                -
                \frac{1}{n}\inparen{\sum_{S\in \cS} \frac{\sum_{x\in S} e(x)\mu_1(x)\cD(x)}{e(S)(1-\cD(N))}}^2\,.
                \yesnum\label{eq:hardness:part1:tmp}
            \end{align*}
            Since $e(x)=1$ for all $x\not\in S_i$ and $S_i=\inbrace{x_m}$, the first term satisfies the following: %
            \begin{align*}
                \inparen{\sum_{S\in \cS} \frac{\sum_{x\in S} e(x)\mu_1(x)\cD(x)}{e(S)(1-\cD(N))}}^2
                &= \frac{1}{\inparen{1-\cD(N)}^2}
                \inparen{
                        \Delta\eps\mu_1(x_m)
                        + 
                    \sum_{S\neq S_i} {\sum_{x\in S} \mu_1(x)\cD(x)}
                    }^2\,.
            \end{align*}
            Further as $\cD(N)\leq 1-\cD(x_m)\leq 1-\Delta\eps$ and $\alpha=\eps^3$
            \begin{align*}
                \inparen{\sum_{S\in \cS} \frac{\sum_{x\in S} e(x)\mu_1(x)\cD(x)}{e(S)(1-\cD(N))}}^2
                &\leq 81A^2\eps^4\,. %
            \end{align*}
            Next, we lower-bound the first term in \cref{eq:hardness:part1:tmp} as follows
            \begin{align*}
                \frac{1}{n} \frac{
                    \sum_{x\in S_i} e(x) \cD(x)\inparen{v_1(x)+\mu_1(x)^2}
                }{
                    e(S_i)^2\inparen{1-\cD(N)}
                }
                &~~\stackrel{(S_i=\inbrace{x_m})}{=}~~ \frac{1}{n} \frac{
                    \cD(x_m) %
                    \inparen{v_1(x_m)+\mu_1(x_m)^2}
                }{
                    e(x_m)\inparen{1-\cD(N)}
                }
                \geq \frac{
                    4A^2 \Delta\eps 
                }{
                    n\beta
                }\,.
            \end{align*}
            Since $\beta=\eps^5$ and $n=\eps^{-7}$, it follows that
            \[
                \MSE{(\tau_{\cS, N})}\geq \variance{(\tau_{\cS, N})}\geq \Omega\inparen{A^2\eps^3}
                \,.
            \]
            This is a contradiction to $\MSE{(\tau_{\cS, N})}\leq \sfrac{U}{\sqrt{\eps}}= 8A^2\eps^{6.5}$. %

        \bigskip

        \noindent\textit{Part 2 ($x_m\in S_i$ and $\abs{S_i}\geq 2$):}
            Without loss of generality let $i=1$.
            Since $x_m\not\in S_2\cup S_3\cup \dots \cup S_k$, all items in $S_2\cup S_3\cup \dots \cup S_k$ have the same propensity score and hence, the partition $\inparen{\hypo{T}\coloneqq \inbrace{S_1, S_2\cup S_3\cup \dots \cup S_k}, N}$ has the same $\MSE{}$ as $\inparen{\cS,N}$:
            \[
                \MSE{\inparen{\cS,N}} = \MSE{\inparen{\tau_{\hypo{T}, N}}}\,.
            \]
            Consider two cases.

            \begin{itemize}
                \item \textbf{Case A ($x_{m-1}\in S_1$):}
                    Observe that $\cD(S_1) 
                            \in \Delta\inparen{1+\eps} \pm k\Delta\alpha$ and $\sum_{i\in S_1} e(x)\cD(x)
                            \in \Delta\inparen{1+\eps\beta}\pm k\Delta\alpha$, and, hence, 
                    \[
                        e(S_1) 
                        \in 1\pm O(\eps)\,.
                    \]
                    Where we also use that $k\alpha=O(\eps)$ and $\beta\leq 1$.
                    Moreover, we have that 
                    \[
                        \cD(N) \leq k\Delta\alpha = O(\eps)\,.
                    \]
                    Substituting the above in the expression of $\bias\inparen{\tau_{\hypo{T}, N}}$ implies that 
                    \begin{align*}
                        \bias\inparen{\tau_{\hypo{T}, N}}
                        &= \inparen{
                                \frac{\sum_{x\in S_1} e(x)\cD(x)\mu_1(x)}{
                                        \inparen{1\pm O(\eps)}\inparen{1-O(\eps)}
                                    } 
                                + \frac{\sum_{x\in S_2\cup\dots\cup S_k} e(x)\cD(x)\mu_1(x)}{
                                        e(S_2\cup\dots\cup S_k)\inparen{1-O(\eps)}
                                    } 
                                    - \sum_x \cD(x)\mu_1(x)
                            }^2\,.
                    \end{align*}
                    Further, as $e(x)=1$ for all $x\in S_2\cup\dots\cup S_k$, we get that 
                    \begin{align*}
                        \bias\inparen{\tau_{\hypo{T}, N}}
                        &= \inparen{
                                \frac{\sum_{x\in S_1} e(x)\cD(x)\mu_1(x)}{
                                        \inparen{1\pm O(\eps)}\inparen{1-O(\eps)}
                                    } 
                                + \frac{\sum_{x\in S_2\cup\dots\cup S_k} e(x) \cD(x)\mu_1(x)}{
                                        \inparen{1-O(\eps)}
                                    } 
                                    - \sum_x \cD(x)\mu_1(x)
                            }^2 \,.
                    \end{align*}
                    Define
                    \[
                        R = \inparen{S_1\cup S_2\cup\dots S_k} \cap [k]\,.
                    \]
                    Observe that  
                    \begin{align*}
                        \bias\inparen{\tau_{\hypo{T}, N}}
                            &= \inparen{
                                \frac{
                                        2A\Delta\alpha
                                        + \inparen{-3A+\nfrac{(T+2A)}{\Delta}}\Delta\eps\beta
                                        + \sum_{x\in R} e(x)\cD(x)\mu_1(x)
                                    }{
                                        1\pm O(\eps)
                                    } 
                                    - \sum_x \cD(x)\mu_1(x)
                            }^2\\
                            &= \inparen{
                                \frac{
                                        2A\Delta\alpha
                                        + \inparen{-3A+\nfrac{(T+2A)}{\Delta}}\Delta\eps\beta
                                        + \sum_{x\in R} e(x)\cD(x)\mu_1(x)
                                    }{
                                        1\pm O(\eps)
                                    } 
                                    - (T+2A)\alpha 
                            }^2\,.
                    \end{align*}
                    Since $\Delta\leq 1$ and $\Delta\in 1\pm O(\eps)$, the above expression simplifies to:
                    \begin{align*}
                        \bias\inparen{\tau_{\hypo{T}, N}}
                        =
                        \inparen{
                                \frac{
                                        \sum_{x\in R} e(x)\cD(x)\mu_1(x)
                                    }{
                                        1\pm O(\eps)
                                    } 
                                    - T\alpha 
                                    \pm O(A\eps(\alpha+\beta))
                            }^2\,.
                    \end{align*}
                    Substituting the values of $e(x)$, $\cD(x)$, and $\mu_1(x)$, we get 
                    \begin{align*} 
                        \bias\inparen{\tau_{\hypo{T}, N}}
                        &= \inparen{
                                \frac{
                                        \Delta\alpha \sum_{x\in R} a_i
                                    }{
                                        1\pm O(\eps)
                                    } 
                                    - T\alpha 
                                    \pm O(A\eps(\alpha+\beta))
                            }^2
                            \,.
                    \end{align*}
                    Using $\beta\leq \alpha$ and simplifying we get that 
                    \begin{align*}
                        \bias\inparen{\tau_{\hypo{T}, N}}
                        &= \alpha^2\inparen{
                                    \sum_{x\in R} a_i - T
                                    \pm O\inparen{A\eps}
                            }^2
                            \,.
                            \yesnum\label{eq:hardness:part2:caseb}
                    \end{align*}
                    Recall that $\bias\inparen{\hypo{T},N}\leq \MSE{}\inparen{\hypo{T},N}\leq O(\alpha^2\sqrt{\eps})$.
                    This combined with \cref{eq:hardness:part2:caseb} implies that $\sum_{x\in R} a_i = T$ as if not, then $\abs{\sum_{x\in R} a_i - T}^2\geq 1$ and, hence, $\bias\inparen{\hypo{T},N}\geq \alpha^2$.
                    Therefore, $R$ is a certificate that \instanceSS{} is a \texttt{Yes} instance.
                \item \textbf{Case B ($x_{m-1}\not \in S_1$):}
                    Observe that $\sum_{x\in S_1} e(x)\cD(x) \leq \Delta\alpha\inparen{\beta + k}$ and $\sum_{x\in S_1}\cD(x)\geq \Delta\inparen{\eps - k\alpha}$, and, hence
                    \[
                        e(S_1) 
                        \leq \frac{k\alpha+\alpha\beta}{\eps-k\alpha}
                        \leq \frac{k\alpha}{\eps} 
                            \inparen{
                                    1 
                                    +  O\inparen{\frac{k\alpha}{\eps}}
                                    +  O\inparen{\frac{\beta}{k}}
                            }
                        \leq \frac{3k\alpha}{\eps}\,.
                        \yesnum\label{eq:hardness:part2:caseb:ub_on_propensity}
                    \]
                    Where the last inequality holds because $k\alpha\leq 2\eps$ and $\beta\leq k$.
                    Recall that 
                    \begin{align*}
                        n\cdot \MSE{\inparen{\tau_{\cS,N}}}
                        \geq \variance\inparen{\tau_{\cS,N}}
                        = \frac{
                            \sum_{x\in S_i} e(x) \cD(x)\inparen{v_1(x)+\mu_1(x)^2}
                        }{
                            e(S_i)^2\inparen{1-\cD(N)}
                        }
                        -
                        \inparen{\sum_{S\in \cS} \frac{\sum_{x\in S} e(x)\mu_1(x)\cD(x)}{e(S)(1-\cD(N))}}^2\,.
                    \end{align*}
                    Since $1-\cD(N)\leq 1$ and each term in the sum is non-negative, the first term satisfies: %
                    \begin{align*}
                        \frac{1}{n}\frac{
                                \sum_{x\in S_1} e(x) \cD(x) \inparen{v_1(x)+\mu_1(x)^2}
                            }{
                                e(S_1)^2\inparen{1-\cD(N)}
                            }
                        &\geq \frac{1}{n}\frac{
                                \sum_{x\in S_1\backslash \inbrace{x_m}} e(x) \cD(x) \inparen{v_1(x)+\mu_1(x)^2}
                            }{
                                e(S_1)^2
                            }\,.%
                    \end{align*}
                    Further, \cref{eq:hardness:part2:caseb:ub_on_propensity} and $\sum_{x\in S_1\backslash\inbrace{x_m}}e(x)\cD(x)\inparen{v_1(x)+\mu_1(x)^2}\leq 4A^2\Delta\alpha k$ imply that 
                    \begin{align*}
                         \frac{1}{n}\frac{
                                \sum_{x\in S_1} e(x) \cD(x) \inparen{v_1(x)+\mu_1(x)^2}
                            }{
                                e(S_1)^2\inparen{1-\cD(N)}
                            }
                        \geq \frac{1}{n}\frac{4\eps A^2\Delta k}{9k^2\alpha}
                        &
                        \quad
                        \stackrel{(\Delta\geq \frac{1}{2},\ n=\eps^{-7})}{\geq}
                        \quad {\frac{A^2\eps^9}{2k^2\alpha}}\,. 
                        \yesnum\label{eq:hardness:part2:case2:lowerbound}
                    \end{align*}
                    Further, as $e(x)=1$ for all $x\not\in S_1$,  second term satisfies
                    \begin{align*}
                        \frac{1}{n}\inparen{\sum_{S\in \cS} \frac{\sum_{x\in S} e(x)\mu_1(x)\cD(x)}{e(S)(1-\cD(N))}}^2
                        &= 
                        \frac{1}{n}
                        \inparen{
                            \frac{\sum_{x\in S_1} e(x)\mu_1(x)\cD(x)}{e(S_1)(1-\cD(N))}
                            + \sum_{S\neq S_1} \frac{\sum_{x\in S} \mu_1(x)\cD(x)}{1-\cD(N)}
                        }^2\,.
                    \end{align*}
                    To simplify this, observe that:
                    First, $\sum_{S_1}e(x)\cD(x)=\Delta\eps\beta+\inparen{\abs{S_1}-1}\Delta\alpha$ and $\sum_{S_1}\cD(x)=\Delta\eps+\inparen{\abs{S_1}-1}\Delta\alpha$ and, hence, $e(S_1)=\frac{\eps\beta+\alpha\inparen{\abs{S_1}-1}}{\eps+\alpha\inparen{\abs{S_1}-1}}$.
                    Second, $\cD(N)\leq 1-\cD(x_m)=1-\Delta\eps$.
                    Combining these two implies that 
                    \begin{align*}
                        \frac{1}{n}\inparen{\sum_{S\in \cS} \frac{\sum_{x\in S} e(x)\mu_1(x)\cD(x)}{e(S)(1-\cD(N))}}^2
                        &\leq 
                        \frac{1}{n\Delta^2\eps^2}
                        \inparen{%
                            \begin{array}{c}
                                 \sum_{x\in S_1} e(x)\mu_1(x)\cD(x)\cdot 
                                 \frac{\eps+\alpha\inparen{\abs{S_1}-1}}{\eps\beta+\alpha\inparen{\abs{S_1}-1}}\\
                                 + \sum_{S\neq S_1} \sum_{x\in S} \mu_1(x)\cD(x)
                            \end{array}
                        }^2
                \end{align*}
                Substituting the values for different quantities implies that 
                \begin{align*}
                        \frac{1}{n}\inparen{\sum_{S\in \cS} \frac{\sum_{x\in S} e(x)\mu_1(x)\cD(x)}{e(S)(1-\cD(N))}}^2 %
                        &\leq 
                        \frac{9A^2\alpha^2}{n\eps^2}
                        \inparen{%
                             \frac{
                                \inparen{2\beta+\abs{S_1}-1}\inparen{\eps+\alpha\inparen{\abs{S_1}-1}}
                            }{\eps\beta+\alpha\inparen{\abs{S_1}-1}} 
                             + 1
                        }^2\\
                        &\leq 
                        \frac{9A^2\alpha^2}{n\eps^2}
                        \inparen{%
                             \frac{
                                \inparen{\eps+\alpha k}
                                \inparen{2\beta+\abs{S_1}-1}
                            }{\eps\beta+\alpha\inparen{\abs{S_1}-1}} 
                             + 1
                        }^2\,.
                        \tag{as $\abs{S_1}\leq k+1$}
                \end{align*}
                Since $\frac{z+x}{y+x}$ is an increasing function of $x$ for $y>z$ on $x>-y$ and $\abs{S_1}\leq k+1$
                \begin{align*}
                        \frac{1}{n}\inparen{\sum_{S\in \cS} \frac{\sum_{x\in S} e(x)\mu_1(x)\cD(x)}{e(S)(1-\cD(N))}}^2
                        &\leq 
                        \frac{9A^2\alpha^2}{n\eps^2}
                        \inparen{%
                             \frac{
                                \inparen{\eps+\alpha k}\inparen{2\beta+k}
                             }{
                                \eps\beta+\alpha k
                             } 
                             + 1
                        }^2 \\
                        &\leq 225A^2\eps^9\,.
                        \tagnum{since $\alpha\leq \eps$ and $n=\eps^{-7}$}{eq:hardness:part2:case2:upperbound}
                    \end{align*}
                    Combining Equations~\eqref{eq:hardness:part2:case2:lowerbound} and \eqref{eq:hardness:part2:case2:upperbound}, and using that $\alpha=\eps^3$ and $\eps\ll \sfrac{1}{(Ak)^2}$, implies that 
                    \[ 
                        \MSE{(\tau_{\cS, N})}
                        \geq \frac{A^2\eps^9}{2k^2\alpha} - 225 A^2\eps^9
                        \geq \frac{A^2\eps^6}{3k^2}\,.
                    \]
                    
                    Which is a contradiction since on the one hand $\MSE{(\tau_{\cS, N})}\leq \sfrac{U}{\sqrt{\eps}}= 8A^2\alpha^2\sqrt{\eps}$ and on the other hand $\MSE{(\tau_{\cS, N})}\geq \sfrac{A^2\eps^6}{(3k^2)}$.
                    (To see the contradiction, observe that as $\alpha=\eps^3$ and $\eps\ll k^{-4}$, $8A^2\alpha^2\sqrt{\eps}<\sfrac{A^2\eps^6}{(3k^2)}$.)
            \end{itemize}

        \bigskip 

        \noindent \textit{Part 3 ($x_m\in N$)}   
        We consider two cases. 
        \begin{itemize}
            \item \textbf{Case A ($x_{m-1}\not\in N$):}
                Since $x_m\in N$, all elements in $S_1\cup\dots\cup S_k$ have the same propensity score and, hence, 
                \begin{align*}
                    \bias\inparen{\tau_{\cS,N}}
                    = \inparen{
                            \frac{\sum_{x\not\in N} \mu_1(x)\cD(x)}{1-\cD(N)} 
                            - \sum_x \mu_1(x)\cD(x)
                        }^2
                    &= \inparen{
                            \frac{\sum_{x\not\in N} \mu_1(x)\cD(x)}{1-\cD(N)} 
                            - (T+2A)\alpha
                        }^2.
                \end{align*}
                Moreover, $1-\cD(N) \in 1 \pm \inparen{\eps + k\alpha}$.
                Therefore 
                \begin{align*}
                    \bias\inparen{\tau_{\cS,N}}
                    &= \inparen{
                            \frac{\sum_{x\not\in N} \mu_1(x)\cD(x)}{1 \pm \inparen{\eps + k\alpha}} 
                            - (T+2A)\alpha
                        }^2.
                \end{align*}
                Define
                \[
                    R\coloneqq \inparen{S_1\cup S_2\cup\dots\cup S_k}\cap [k]\,.
                \]
                It follows that 
                \begin{align*}
                    \bias\inparen{\tau_{\cS,N}}
                    &= \inparen{
                            \frac{
                                2A\Delta\alpha 
                                +
                                \sum_{i\in R} \mu_1(x_i)\cD(x_i)
                            }{1 \pm \inparen{\eps + k\alpha}} 
                            - (T+2A)\alpha
                        }^2 = \alpha^2 \inparen{
                                \sum_{i\in R} a_i - T
                                \pm O(A\inparen{\eps + k\alpha})
                        }^2.
                \end{align*}
                Since $\bias\inparen{\tau_{\cS, N}}=O(\alpha^2\sqrt{\eps})$ and if $\sum_{i\in R}a_i\neq T$ then $\abs{\sum_{i\in R}a_i - T}\geq 1$, it must be that $\sum_{i\in R} a_i = T.$
                Therefore, $R$ is a certificate that \instanceSS{} is a \texttt{Yes} instance. 
            \item \textbf{Case B ($x_{m-1} \in N$):}
                Define
                \[
                    R \coloneqq S_1\cup S_2\cup \dots \cup S_k\,.
                \]
                Since $x_m\in N$ all elements in $R=S_1\cup S_2\cup \dots \cup S_k$ have the same propensity score 
                \begin{align*}
                    \MSE\inparen{\tau_{\cS, N}}
                        &= \inparen{
                                \frac{\sum_{x\in R} \cD(x)\mu_1(x)}{
                                        1-\cD(N)
                                    } 
                                    - \sum_x \cD(x)\mu_1(x)
                            }^2.
                \end{align*}
                Moreover $1-\cD(N)=\cD(R) = \abs{R}\Delta\alpha$ and $\sum_{x\in R}\cD(x)\mu_1(x)=\Delta\alpha \sum_{i\in R} a_i$.
                Therefore 
                \begin{align*}
                    \MSE\inparen{\tau_{\cS, N}}
                        &= \inparen{
                                \frac{
                                    \sum_{i\in R} a_i
                                }{
                                    \abs{R}
                                } 
                                - \sum_x \cD(x)\mu_1(x)
                            }^2 = \inparen{
                                \frac{
                                    \sum_{i\in R} a_i
                                }{
                                    \abs{R}
                                } 
                                - O(A\alpha)
                            }^2\,.
                \end{align*}
                Since $\abs{R}\leq m$, $R$ is non-empty and $a\geq 1$, it follows that 
                \begin{align*}
                        \MSE\inparen{\tau_{\cS, N}}
                        &= \Omega\inparen{\frac{1}{m^2}}\,.
                \end{align*}
                This is a contradiction since $\MSE\inparen{\tau_{\cS,N}}\leq \sfrac{U}{\sqrt{\eps}}=O(\alpha^2\sqrt{\eps})=\eps^{6.5}$ and $\Omega\inparen{m^{-2}} \gg \eps.$
        \end{itemize}

    \end{proof}
    \subsubsection*{Completeness}\label{sec:proofof:thm:hardness:completeness}
    Finally, we prove completeness. 
    \begin{lemma}[{Completeness}]\label{lem:reduction:completeness}
        Consider any instance $\instanceSS{}(a,T)$ of \problemSS{} and the corresponding instance $\instanceMSE{}(m,n,e,\mu,v,\cD,U)$.
        If there is a subset $R\subseteq [k]$ such that $\sum_{i\in R}a_i=T$, then $\MSE\inparen{\tau_{\cS, N}}\leq U$ where $\cS = \inbrace{R\cup \inbrace{x_{m-1}}}$ and $N=\inbrace{x_{m}}\cup [k]\backslash R.$
    \end{lemma}
    \begin{proof}\proofof{\cref{lem:reduction:completeness}}
        Given a certificate $R$ that \instanceSS{} is a \texttt{Yes} instance, define 
        \[
            \cS = \inbrace{S = R\cup \inbrace{x_{m-1}}}
             \quadand 
            N=\inbrace{x_{m}}\cup [k]\backslash R\,.
        \]
        Since $x_m\in N$, all elements in $S$ have the same propensity score and, hence,
        \[
            \bias\inparen{\tau_{\cS, N}} = \inparen{
               \frac{\sum_{x\in S}\mu_1(x)\cD(x)}{1-\cD(N)} - \sum_x \mu_1(x) \cD(x)
            }^2\,.
        \]
        Moreover $1-\cD(N) = \cD(R)+\cD\inparen{x_{m-1}}\in \Delta \pm k\Delta\alpha$.
        Therefore 
        \begin{align*}
            \bias\inparen{\tau_{\cS, N}} 
                &= \inparen{
                   \frac{
                        \Delta\alpha \sum_{i\in R}a_i
                        + 2A\Delta\alpha
                    }{
                        \Delta\pm k\Delta\alpha
                    } 
                    - (T+2A)\alpha 
                }^2
                = O(A^2k^2\alpha^4)\,. \tag{as $\sum_{i\in R}a_i=T$}
        \end{align*}
        Further, since all $x\in S$ have $e(x)=1$
        \[
            \variance{\inparen{\tau_{\cS, N}}}
            = \frac{1}{n} \frac{
                    \sum_{x\in S} e(x) \cD(x) \inparen{v_1(x)+\mu_1(x)^2}
                }{
                    e(S)^2\inparen{1-\cD(N)}
                }
            = \frac{1}{n} \frac{
                    \sum_{x\in S} \cD(x) \inparen{v_1(x)+\mu_1(x)^2}
                }{
                    \inparen{1-\cD(N)}
                }\,.
        \]
        Further 
        \begin{align*}
            \variance{\inparen{\tau_{\cS, N}}}
            &\leq \frac{1}{n} \frac{
                    4A^2 \sum_{x\in S} \cD(x)
                }{
                    1-\cD(N)
                }
            = \frac{4A^2}{n}
            = 4A^2\eps^7\,.
        \end{align*}
        Thus, it follows 
        \begin{align*}
            \MSE\inparen{\tau_{\cS, N}}
            = \bias{\inparen{\tau_{\cS, N}}}^2 + \variance{\inparen{\tau_{\cS, N}}}
            \leq O(A^2k^2\alpha^4) + 4A^2\eps^7
            \leq 8A^2\eps^7\,.
        \end{align*}
    \end{proof}

\subsection{Statistical Hardness of Learning Good-Local Partitions} %
\label{sec:statisticalHardness}

    In this section, we prove \cref{lem:impossibility_of_learning_gl_partition_intro}.
    The formal version of \cref{lem:impossibility_of_learning_gl_partition_intro} is as follows.

            \begin{restatable}[{Impossibility of Learning a Good-Local Partition}]{lemma}{LemmaImpossibility}\label{lem:impossibility_of_learning_gl_partition}
                Fix $\alpha,\beta>0$ and $\gamma<\sfrac{1}{16}$. %
                Suppose \cref{asmp:lipschitzness} holds and there exists an \goodlocal{\alpha,\beta,\gamma} $(\cS, N)$.
                Further assume that $\cS\subseteq\hyH$ and $N\in \hyH$, where $\hyH$ is a hypothesis class.
                There is an $\hyH$ with $\vc{}(\hyH)=O(1)$ and, for any $n\geq 1$, an unconfounded distribution $\cD$, such that, no algorithm, given a censored dataset $\dataset\sim \cD$ of size $n$, outputs an \goodlocal{\alpha,\beta,\gamma} with probability $>\sfrac{1}{8}$.\footnote{We remind the reader of the definition of the VC dimension.
                    \begin{definition}[{VC Dimension}]
                        Consider any collection of subsets $\cH$ of $\R^d$.
                    	Given a finite set $S$, define the collection of subsets $\cH_S\coloneqq \{H\cap S\mid H\in \cH\}$.
                    	We say that $\cH$ shatters a set $S$ if $|\cH_{S}|=2^{|S|}$.
                    	The VC dimension of $\cH$, $\vc{}(\cH)\in \N$,  is the largest integer such that there exists a set $S$ of size $\vc{}(\cH)$ that is shattered by $\cH$.
                    \end{definition}}
            \end{restatable}
            As mentioned in \cref{sec:overview}, the proof \cref{lem:impossibility_of_learning_gl_partition} utilizes the fact that $\cS$ can contain a large number of subsets.
            Using this, it reduces the problem of PAC-learning a union of $\abs{\cS}$ intervals (of width at most $\alpha$) up to $\gamma$-error in the agnostic setting to finding an \goodlocal{\alpha,\beta,\gamma}.
            The result follows that this class has a VC-dimension of at least $\abs{\cS}$ and, hence, if $\abs{\cS}\to\infty$, then it cannot be learned with any finite number of samples.

        \subsubsection{Technical Overview}
            In this section, we overview the proof of \cref{lem:impossibility_of_learning_gl_partition}.
            The complete proof appears in \cref{sec:statisticalHardness:proof}.

            Define $k \coloneqq \abs{\cS}$; this is a parameter in our proof and we select $k>n$.
            It will suffice to consider one-dimensional covariates, i.e., $d=1$.
            Let $\hyH$ be the set of all intervals of width at most $\alpha$.
            Let $\hyG$ to be the union of $k$ sets from $\hyH$, i.e., $\hyG \coloneqq \inbrace{H_1\cup H_2\cup \dots\cup H_k\colon H_1,H_2,\dots,H_k\in \hyH}$
            It is straightforward to see that $\hyG$ can shatter of set of $k$ points $X=\inbrace{x_1,x_2,\dots,x_k}\subseteq \R$  such that any pair of points in $X$ is at least $4\alpha$ apart.
            Hence, in particular, using well-known lower bounds, it follows that $\hyG$ cannot be PAC-learned with $k$ samples.\footnote{Recall that a set of $p$ points $X=\inbrace{x_1,x_2,\dots,x_p}\subseteq \R$ are shattered by a hypothesis class $\hyG\subseteq \zo^\R$ if for every vector $b\in \zo^{p}$, there is a set $G\in \hyG$ such that $x_i\in G$ if and only if $b_i=1$.}

            \paragraph{PAC-Learning Instance.} Fix any $X$ of size $k$ shattered by $\hyG$ such that points in $X$ are pairwise $4\alpha$ apart.
            Consider any distribution $\cP$ on $X\times \zo$.
            Let ${\rm neg}(X)=\inbrace{x\in X\colon \Pr_{(x,y)\sim \cP}[y=1] <\sfrac{1}{2}}$ and ${\rm pos}(X)=X\backslash {\rm neg}(X)$.
            Observe that any $G\in \hyG$ with error $\eps+\opt{}$ on $\cP$ must label all but $\eps$-fraction of points in ${\rm neg}(X)$ negatively and all but $\eps$-fraction of points in ${\rm pos}(X)$ positively.
        
            \paragraph{Reduction.}
            Given $X$ and $\cP$, we construct an unconfounded distribution $\cD=\cD_{\alpha,\beta,\gamma,\eps}$ such that any \goodlocal{\alpha,\beta,\gamma} w.r.t. to $\cD$ can be used to find a hypothesis with $\eps+\opt{}$ error on $\cP$ as follows:
                        given $(\cS, N)$, let $\cS_{>1}\subseteq \cS$ be the subset of $\cS$ where each set $S\in \cS_{>1}$ has more than one point, and let $G_\cS$ be the hypothesis that labels all points from $X$ in $\cS_{>1}$ negatively and all remaining points positively.

                \paragraph{Construction.}
                At a high level, start with covariates $X$ and add $\abs{{\rm neg}(X)}$ additional covariates $Z$:
                    for each $x_i\in {\rm neg}(X)$, we add a covariate $z_i$ very close to it.
                We select propensity scores so that each $x\in X$ has $e(x)=\sfrac{1}{2}$ and each $z\in Z$ is a $\eta$-outlier for $\eta\ll \beta$.
                The remainder of the reduction is independent of the distribution of the outliers and, hence, we can choose it to satisfy unconfoundedness and \cref{asmp:lipschitzness}.
                Further, observe that the above construction already satisfies \cref{asmp:sparsity,asmp:isolation} with constants $\alpha,\beta$.

                \paragraph{Overview of Soundness and Completeness.}
                For simplicity suppose $\gamma=0$.
                Now we show that any \goodlocal{\alpha,\beta,0} $(\cS, N)$ must cluster each $z_i\in Z$ with the only point within $\alpha$-distance of it, namely, $x_i\in {\rm neg}(X)$.
                If not, then $e(S) \leq \eta$ for the partition $S$ containing $z_i$. ($z_i$ cannot be clustered with another point as all other points at further than $\alpha$-away and $\diam{(S)}\leq\alpha$.)
                Hence, we can \textit{identify} all samples in ${\rm neg}(X)$ given $(\cS, N)$.
                When $\gamma>0$, then we can identify all but $O(\gamma)$ fraction of points in ${\rm neg}(X)$, which we show is sufficient to achieve $(\gamma\eps)$-accurate.

        \subsubsection{Proof of \cref{lem:impossibility_of_learning_gl_partition}}\label{sec:statisticalHardness:proof}

            \begin{proof}\proofof{\cref{lem:impossibility_of_learning_gl_partition}}
                Define $k \coloneqq \abs{\cS}$; this is a parameter in our proof and we select 
                \[
                    k > n\,.
                \]
                It will suffice to consider one-dimensional covariates, i.e., $d=1$.
                We will select $\hyH$ to be the set of all intervals of width at most $\alpha$, i.e., $\hyH = \inbrace{[a,b]\colon \abs{a-b} \leq \alpha}$.
                Note that any $S\in \hyH$ has $\diam{(S)}\leq \alpha$ as required by the definition of an \goodlocal{\alpha,\beta,\gamma}.

                \paragraph{Proof outline (Connection to Agnostic PAC Learning).}
                Recall that a set of $p$ points $X=\inbrace{x_1,x_2,\dots,x_p}\subseteq \R$ are shattered by a hypothesis class $\hyG\subseteq \zo^\R$ if for every vector $b\in \zo^{p}$, there is a set $G\in \hyG$ such that $x_i\in G$ if and only if $b_i=1$.
                It is well known that if a hypothesis class can shatter sets of size $p$ then it cannot be learned using $o(p)$ samples.
                We will use the following version of this fact in our proof. 
                \begin{fact}[Section 28.2.2 of \citet{shalev2014understanding}]\label{fact:lower_bound_sample_complexity}
                    For any $\eps>0$, 
                        hypothesis class $\hyG\subseteq\zo^\R$, and 
                        $p$ points $X=\inbrace{x_1,x_2,\dots,x_p}\subseteq\R$ that are shattered by $\hyG$, 
                    there is a distribution $\cP$ over $X\times \zo$, such that, 
                        no algorithm, given $8\eps^{-2}p$ samples from $\cP$ outputs a hypothesis $G\in \hyG$ such that the error of $G$, ${\rm Err}_\cP(G)$ is at most $\eps+\opt{}$.
                    Where ${\rm Err}_\cP(G) \coloneqq \Pr_{(x,y)\sim \cP}\insquare{y\neq \mathds{1}_{x\in G}}$ and $\opt{}\coloneqq \min_{G\in \hyG}{\rm Err}_\cP(G)$ is the minimum error of a hypothesis in $\hyG$.
                \end{fact}
                We will choose $\hyG$ to be the union of $k$ sets from $\hyH$, i.e., 
                \[
                    \hyG \coloneqq \inbrace{H_1\cup H_2\cup \dots\cup H_k\colon H_1,H_2,\dots,H_k\in \hyH}\,.
                \]
                Observe that for any \goodlocal{\alpha, \beta, \gamma} $(\cS, N)$, $\cS\in \hyG$.
                Given any $\cS\subseteq \hyG$, divide it into two parts:
                $\cS$ into two parts $\cS_{1}$ such that each $S\in \cS_1$ contains exactly one point from $X$ and $\cS_{>1}$ where each $S\in \cS_{>1}$ contain at least 2 points from $X$.
                Let $G_\cS$ be the hypothesis that labels all points from $X$ in $\cS_{>1}$ negatively and all remaining points positively.
                In the remainder of the proof, we construct distribution $\cD$ where the following implication holds:
                \[
                    \text{if $(\cS, N)$ is an \goodlocal{\alpha,\beta,\gamma}}\,,\quadtext{then}
                    {\rm Err}(G_\cS) \leq \eps + \opt{}\,.
                    \yesnum\label{eq:sample_complexity:implication}
                \]
                Hence, any algorithm for learning an \goodlocal{\alpha,\beta,\gamma} implies an algorithm for finding a hypothesis, namely $G_\cS\in \hyG$, with near-optimal error.
                The result then follows because the sample complexity of finding a hypothesis with near-optimal error is at least $k$ is at least $k>n$ by \cref{fact:lower_bound_sample_complexity}.

                One way to see why the sample complexity of finding a hypothesis with near-optimal error is at least $k$, is to construct a set of $2k$ points shattered by $\hyG$.
                While such sets exist, more interestingly for us, there is a set of $k$ points $X=\inbrace{x_1,x_2,\dots,x_k}\subseteq \R$ that is shattered by $\hyG$ and has the property that pair of points in $X$ is at least $4\alpha$ apart, i.e., 
                \[
                    \min_{x,x'\in X;\ x\neq x'} \abs{x-x'} > 4\alpha\,.
                \]
                It suffices to select $X=\inbrace{5\alpha,10\alpha, \dots, 5k\alpha}$.
                Fix this $X$ for the remainder of the proof.
                
                In the remainder of the proof, we construct a distribution $\cD$ (based on the distribution $\cP$ promised to exist in \cref{fact:lower_bound_sample_complexity}) and prove \cref{eq:sample_complexity:implication}.

                \paragraph{Construction of unconfounded distribution $\cD$.}
                Fix the set $X=\inbrace{0,5\alpha, 10\alpha, \dots, 5k\alpha}$, and hypothesis classes $\hyH$ and $\hyG$ described in the overview.
                Let $\cP$ be the distribution promised \cref{fact:lower_bound_sample_complexity} supported on $X\times \zo$.
                Divide the samples in $X$ into two parts:
                \begin{align*}
                    {\rm neg}(X) &\coloneqq 
                        \inbrace{
                            x\in X\colon \Pr_{(x,y)\sim \cP}[y=1\mid x=x_i]< \nfrac{1}{2}
                        }\,,\\
                    {\rm pos}(X) &\coloneqq 
                        \inbrace{
                            x\in X\colon \Pr_{(x,y)\sim \cP}[y=1\mid x=x_i]\geq  \nfrac{1}{2}
                        }\,.
                \end{align*}
                A useful rule of thumb, which we formalize later is that any hypothesis with near-optimal error must label most of the points in ${\rm neg}(X)$ negatively.
                For each of notation, let $m=\abs{{\rm neg}(X)}$ and
                \[
                    {\rm neg}(X) = \inbrace{x_{i(1)}, x_{i(2)}, \dots, x_{i(m)}}\,.
                \]
                We will introduce a set of $m$ new points $Z=\inbrace{z_{i(1)},z_{i(2)},\dots,z_{i(m)}}$ that satisfy:
                \[
                        \text{for each $1\leq \ell\leq m$}\,,\quad 
                            0<\abs{z_{i(\ell)}-x_{i(\ell)}}<\alpha\,.
                \]
                (For instance, it suffices to let $z_{i(\ell)}=5i(\ell)\cdot \alpha - (\sfrac{\alpha}{2})$ for each $1\leq \ell\leq m$).
                Define $\cD_\cX$ to be the uniform mixture of $\cP_X$ (i.e., the projection of $\cP$ on $X$) and the uniform distribution on $Z$, i.e., 
                \[
                    \cD_\cX = \frac{1}{2}\inparen{\cP_X + \text{Unif}(Z)}\,.
                \]
                The distribution of outcomes at each $x$ is irrelevant for this proof.
                This is because the definition of a good-local partition (\cref{def:goodlocal}) does not depend on outcomes.
                Crucially, due to this, one can choose outcomes that satisfy \cref{asmp:lipschitzness}.
                
                To complete the construction of $\cD$, we select the distribution of treatments such that each point has the following propensity score:
                    for each $z\in Z$ and $x\in X$
                    \[
                        e(z)=\frac{1}{2}
                        \quad \text{and}\quad 
                        e(x) = \begin{cases}
                            \frac{\beta}{2} & \text{if } x\in {\rm neg}(X)\\
                            \frac{1}{2} & \text{otherwise}
                        \end{cases}\,.
                    \]
                In the above construction, all negative samples $x\in {\rm neg}(X)$ are outliers (i.e., have $e(x)(1-e(x))<\beta$) and vice versa.

                We claim that the following is a sufficient condition to ensure that the hypothesis $G\colon X\to \zo$ has error at most ${\rm Err}_\cP(G)\leq 16\eps\gamma + \opt{}$:
                \begin{enumerate}
                    \item For all but $\gamma k$ values of $x\in {\rm pos}(X)$, $G(x)=1$;
                    \item For all but $\gamma k$ values of $x\in {\rm neg}(X)$, $G(x)=0$.
                \end{enumerate}
                To see this, observe that the hypothesis $G^\star$ with error $\opt{}$ labels all points in ${\rm neg}(X)$ negatively and all points in ${\rm pos}(X)$ positively.
                Any hypothesis $G$ satisfying the above claim differs from $G^\star$  on at most $2\gamma k$ points.
                Since $\cP$ is guaranteed to satisfy $\Pr[y=1\mid x=x_i]\in \inbrace{\frac{1-8\eps}{2}, \frac{1+8\eps}{2}}$ and each point has mass $\frac{1}{m+k}$, the error of $G$ and $G^\star$ differ by at most $8\eps\times 2\gamma k\times \frac{1}{m+k}\leq 16\eps\gamma$.

                Now, we are ready to show that any \goodlocal{\alpha,\beta,\gamma} is sufficient to identify a hypothesis with error at most $8\eps\gamma+\opt{}$.
                Consider any \goodlocal{\alpha,\beta,\gamma} $(\cS, N)$.
                Divide $\cS$ into two parts $\cS_{1}$ such that each $S\in \cS_1$ contains exactly one point from $X$ and $\cS_{>1}$ where each $S\in \cS_{>1}$ contain at least 2 points from $X$.
                Since $(\cS, N)$ is an \goodlocal{\alpha,\beta,\gamma}, $\cD(N)\leq \gamma$ and, hence, $N$ contains at most $\gamma\inparen{k+m}\leq 2\gamma k$ points from $X\cup Z$ and, hence, at most $2\gamma k$ points from $X$.

                We make two observations.
                \begin{enumerate}
                    \item \textit{(All $x\in {\rm neg}(X)\backslash N$ are covered by $\cS_{>1}$)}
                        Any point $x\in {\rm neg}(X)$ that is not in $N$, must be covered by a set in $\cS_{>1}$: if not, then the set $S\in \cS$ covering $x$ must be a singleton and, hence, $e(S)=e(x)=\sfrac{\beta}{2}$, which contradicts the fact that $e(S)(1-e(S))\geq \beta$ as required by the definition of a good-local partition.
                    \item \textit{(All $x\in {\rm pos}(X)\backslash N$ are covered by $\cS_{1}$)}
                        Any point $x\in {\rm pos}(X)$ that is not in $N$, must be covered by a set in $\cS_{1}$ as for all $S\in \cS$, $\diam{(S)}\leq \alpha$ but there is no other point that is $\alpha$-close to $x$.
                \end{enumerate}
                Now consider the hypothesis $G$ that negatively labels all points covered by $\cS_{>1}$ and positively labels all remaining points.
                The above observations imply that 
                \begin{align*}
                    \abs{{\rm neg}(X)\cap \inbrace{x\in X\colon G(x)=1}}
                    \leq \abs{N}\leq \gamma k\quadand
                    \abs{{\rm pos}(X)\cap \inbrace{x\in X\colon G(x)=0}}
                    \leq \abs{N}\leq \gamma k\,.
                \end{align*}
                Therefore, from our claim above, it follows that ${\rm Err}(G)\leq 2\gamma k+\opt{}$.
                Thus, given an \goodlocal{\alpha,\beta,\gamma}, one can construct a hypothesis whose error is $(2\gamma k)$-close to the optimal.

                Finally, observe that $n$ draws from $\cP$ can be used to generate $\Theta(n)$ draws from $\cD$ (as $\cD$ is just a uniform mixture of $\cP$ and another distribution).
                Since no algorithm, given $n$ samples from $\cP$, can output a hypothesis with no error with significant probability, no algorithm given $2n$ samples from $\cD$, can output a good local partition either. 
            \end{proof}

\section{Properties of CIPW Estimators Based on Good-Local Partitions}\label{sec:asymptoticNormality}

    \subsection{Asymptotic Normality of CIPW Estimators Based on Good-Local Partitions}
        In this section, we show that CIPW estimators arising from good-local partitions are asymptotically normal.
        
        \begin{theorem}[{Asymptotic Normality}]\label{thm:asymptoticNormality}
            Fix any $\alpha>0$, $\beta, \gamma\in (0,1)$, and $0\leq \eps <\beta/2$.
            For any unconfounded distribution $\cD$, an \goodlocal{\alpha,\beta,\gamma} $(\cS, N)$ for $\cD$, and any (possibly inaccurate) propensity scores $\wh{e}\in B(e, \eps)$
            \[
                \inparen{
                    \tau_{\cS, N}(\dataset; \wh{e})  
                    - \Ex_\cD\insquare{
                        \tau_{\cS, N}(\dataset; \wh{e})
                        }
                }\cdot \sqrt{
                    \frac{n}{
                        \var_\cD\insquare{
                            \tau_{\cS, N}(\dataset; \wh{e}) 
                        }
                    }
                }
                ~\xrightarrow{\abs{\dataset}\to \infty}~
                \cN\inparen{0, 1}\,.
            \]
            where $\dataset$ is generated from $\cD$.
        \end{theorem}
        Thus, CIPW estimators arising from good-local partitions imply confidence intervals on $\tau$.
        
        \medskip
        
        \begin{proof}\proofof{\cref{thm:asymptoticNormality}}
            Recall the definition of $\tau_{\cS, N}(\dataset; \wh{e}):$
            \[
                \tau_{\cS, N}(\dataset; \wh{e})
                = 
                \frac{1}{
                \abs{\sinbrace{i\in [n] \colon x_i \not\in N}}
                }
                \cdot 
                \sum_{S\in \cS} \sum_{i\colon x_i\in S}
                    \inparen{
                        \frac{t_i y_i}{\wh{e}(S)}
                        - \frac{(1 - t_i) y_i}{1 - \wh{e}(S)}
                    }\,.
            \]
            Rewrite this as 
            \[
                \tau_{\cS, N}(\dataset; \wh{e})
                = 
                \frac{1}{
                \abs{\sinbrace{i\in [n] \colon x_i \not\in N}}
                }
                \cdot 
                \sum_{S\in \cS} \sum_{i\colon x_i\in S}
                    \inparen{
                        \frac{t_i y_i(1 - \wh{e}(S)) - (1-t_i)y_i \wh{e}(S)}{\wh{e}(S)(1 - \wh{e}(S))}
                    }\,.\yesnum\label{eq:asympNormality:exp}
            \]
            As explained in \cref{sec:cipw}, since $\wh{e}\in B(e,\eps)$, under mild assumptions on the propensity score and sets, one can derive ensure that $\abs{\wh{e}(T) - e(T)}\leq 1.5\eps$ for each $T\in \cS\cup\inbrace{N}$ for sufficiently large $\abs{\dataset}$.
            Further, since $(\cS, N)$ is an \goodlocal{\alpha,\beta,\gamma}, $\min_{S\in \cS} e(S)(1-e(S))\geq \beta$ and, hence 
            \[
                \min_{S\in \cS}~ \wh{e}(S)(1-\wh{e}(S))
                \geq 
                \frac{\beta}{4}  \,.
            \]
            Where we also used that $0\leq \eps<\beta/2$.
            Consequently, the numerator of each term in \cref{eq:asympNormality:exp} is lower bounded by $\sfrac{\beta}{4}$ and, since the numerator is upper bounded by $2$, it follows that $\tau_{\cS, N}(\dataset; \wh{e})$ is a sum of $m\coloneqq \abs{\sinbrace{i\in [n] \colon x_i \not\in N}}$ terms each of which lies in $\insquare{-\sfrac{8}{\beta}, \sfrac{8}{\beta}}$.
            Furthermore, since $\cD(N)\leq \gamma < 1$, it follows that as $\abs{\dataset}\to \infty$, $m\to \infty$.
            Now the result follows by using Liapounov's Central Limit Theorem for Bounded Random variables \cite{chen2010normal}.
        \end{proof}

    \subsection{Robust RMSE of CIPW Estimators Based on Good-Local Partitions}
    \label{sec:proofof:lem:RobustMSEofLGpartition}
            In this section, we upper bound the robust RMSE of any CIPW estimator based on a good-local partition.
            Concretely, we prove the following result.

            \RobustMSEofLGpartition*

            \noindent Fix any pair of true and inaccurate propensity scores $e,\wh{e}\colon \R^d\to (0,1)$ where $\wh{e}\in B(e,\eps)$.
            Recall %
            \[
                \RMSE{\inparen{\tau_{\cS, N}(\wh{e})}} = \sqrt{\Ex_\dataset\insquare{\inparen{\tau_{\cS, N}(\dataset;  \wh{e}) - \tau}^2}}\,.    
            \]
            We will prove the following upper bound on $\RMSE{\inparen{\tau_{\cS, N}(\wh{e})}}$: if $\cD(N)\leq \frac{1}{2}$, it holds that 
            \[
                \RMSE{(\tau_{\cS, N}(\wh{e}))}^2
                \leq %
                \inparen{
                    4L\sum_{S\in \cS}\cD(S)\diam\sinparen{S} + 8\cD(N)
                }^2
                +\frac{1}{n}
                \sum_{S\in \cS}
                    \frac{\cD(S)}{\wh{e}(S)(1-\wh{e}(S))}
                    \cdot 
                    \frac{1}{1 - \cD(N)}\,.
                \yesnum\label{eq:robustMSE_upperbound:term1}
            \]
            \cref{lem:RobustMSEofLGpartition} follows from the definition of an \goodlocal{\alpha,\beta,\gamma}.

            {For the remainder of the proof, it would be more convenient to work with the MSE and then use $\RMSE{(\tau_{\cS, N}(\wh{e}))}=\sqrt{\MSE{(\tau_{\cS, N}(\wh{e}))}}$.}
            Substituting inaccurate propensity scores $\wh{e}$ and repeating the proof of \cref{lem:exp_of_mse} implies that
            \begin{align*}
                \MSE{(\tau_{\cS, N})(\wh{e})}
                &=
                \begin{array}{l}
                     \inparen{\Ex_X\insquare{\zeta(X)\mid X\not\in N}-\tau}^2 
                     ~~+~~
                     \frac{1}{N} {\Ex_X\insquare{
                        \frac{e(X)(v_1(X)+\mu_1(X)^2)}{(\wh{e}(X; \cS, N))^2} + \frac{(1-e(X))(v_0(X)+\mu_0(X)^2)}{(1-\wh{e}(X; \cS, N))^2} 
                    }}\\
                    -\frac{1}{N}\inparen{ 
                        \Ex_X\insquare{ \zeta(X) \mid X\not\in N}
                    }^2\,.
                \end{array}
            \end{align*}
            Where $\zeta(x)\coloneqq \frac{e(x)\mu_1(x)}{\wh{e}(X; \cS, N)} - \frac{(1-e(x))\mu_0(x)}{1-\wh{e}(X; \cS, N)}$.
            To simplify this, we use the fact that since $(\cS, N)$ is an \goodlocal{\alpha,\beta,\gamma}, for all $S\in \cS$, $e(S), \wh{e}(S)\geq \beta-\eps$ and, hence, for all $S\in \cS$
            \[
                \frac{1}{\wh{e}(S)} 
                \in \frac{1}{1\pm \nfrac{\eps}{\beta}}\cdot \frac{1}{{e}(S)}\,. %
            \]
            Substituting this in the above expression of the MSE implies that $\MSE{(\tau_{\cS, N}(\wh{e}))}\cdot \inparen{1\pm \nfrac{\eps}{\beta}}$ is upper bounded by the following
            \begin{align*}
                &\underbrace{\inparen{\Ex_X\insquare{\psi(X)\mid X\not\in N}-\tau}^2}_{\circled{1}}
                +\frac{1}{N} \underbrace{\Ex_X\insquare{
                    \frac{e(X)(v_1(X)+\mu_1(X)^2)}{(\wh{e}(X; \cS, N))^2} + \frac{(1-e(X))(v_0(X)+\mu_0(X)^2)}{(1-e(X; \cS, N))^2} 
                }}_{\circled{2}}\\
                &\quad \underbrace{-\frac{1}{N}\inparen{ \Ex_X\insquare{ \psi(X) \mid X\not\in N}
                }^2}_{\circled{3}}.
            \end{align*}
            Where $\psi(x)\coloneqq \frac{e(x)\mu_1(x)}{e(X; \cS, N)} - \frac{(1-e(x))\mu_0(x)}{1-e(X; \cS, N)}$.

            We upper bound $\circled{3}$ by 0.
            We prove upper bounds on \circled{1} and \circled{2} below.
            Before proceeding to the result, let $\cS=\inbrace{S_1,S_2,\dots}$ and for each $S_j\in \cS$
            define its \toa{bias} $b(j)$ as follows 
            \[
                b(j) \coloneqq \max_{r\in \zo}\max_{x_1,x_2\in S_j} \abs{\mu_r(x_1)-\mu_r(x_2)}.
                \yesnum\label{def:clusterBias}
            \]

            \paragraph{Upper bound on $(1)$.}
                Our proof is divided into the following two steps:
                \begin{enumerate}
                    \item First, we will show that if clusters $S_1,S_2,\dots,S_m$ have small bias, then \[\Ex_X\insquare{\psi(X)\mid X\not\in N}\approx \Ex_X\insquare{Y^1-Y^0\mid X\not\in N}.\]
                    \item Second, we show that, if $\cD(N)$ is small, then \[\Ex_X\insquare{Y^1-Y^0\mid X\not\in N}\approx \Ex_X\insquare{Y^1-Y^0}.\]
                \end{enumerate}
                Let $\overline{\cD}(x)\coloneqq \Pr[X=x\mid X\not\in N]$.
                Toward the first step, observe that 
                \begin{align*}
                    &\Ex_X\insquare{\psi(X)\mid X\not\in N}\\
                    &= 
                    \int_{x\colon X\not\in N} \inparen{\frac{\overline{\cD}(x) \mu_{1}(x)e(x)}{e(X; \cS, N)} 
                    - \frac{\overline{\cD}(x) \mu_{0}(x)(1-e(x))}{1-e(X; \cS, N)}} dx\\
                    &= 
                    \sum_{S\in \cS}
                    \int_{x\in S} 
                    \inparen{\frac{\overline{\cD}(x) \mu_{1}(x)e(x)}{e(X; \cS, N)} 
                    - \frac{\overline{\cD}(x) \mu_{0}(x)(1-e(x))}{1-e(X; \cS, N)}} dx\\
                    &= 
                    \sum_{S\in \cS}
                    \inparen{\frac{
                    \int_{x\in S} \overline{\cD}(x) \mu_{1}(x)e(x)dx}{e(S; \cS, N)} 
                    - \frac{
                    \int_{x\in S} \overline{\cD}(x) \mu_{0}(x)(1-e(x))dx}{1-e(S; \cS, N)}} 
                    \tag{using that $e(X; \cS, N)$ is constant over any $S$}\\
                    &\in 
                    \sum_{S\in \cS}
                    \inparen{
                    \frac{
                        \inparen{\Ex[Y^1\mid X\in S]\pm b(i)}
                        \int_{x\in S} \overline{\cD}(x) e(x)dx
                    }{e(S; \cS, N)} 
                    - \frac{
                        \inparen{\Ex[Y^0\mid X\in S]\pm b(i)}    \int_{x\in S} \overline{\cD}(x) (1-e(x))dx
                    }{1-e(S; \cS, N)}} 
                    \tag{using \cref{def:clusterBias}}\\
                    &= 
                    \sum_{S\in \cS}
                    \inparen{
                    \inparen{\Ex[Y^1\mid X\in S]\pm b(i)}
                    \int_{x\in S} \overline{\cD}(x)dx
                    -
                    \inparen{\Ex[Y^0\mid X\in S]\pm b(i)}    \int_{x\in S} \overline{\cD}(x)dx
                    }
                        \tag{using that $e(S; \cS, N)
                        =\frac{\int_{x\in S} e(x)\cD(x)dx}{\int_{x\in S}\cD(x)dx}
                        =\frac{\int_{x\in S} e(x)\overline{\cD}(x)dx}{\int_{x\in S}\overline{\cD}(x)dx}$}\\
                    &= 
                        \sum_{S\in \cS}
                        {
                            \inparen{\Ex[Y^1-Y^0\mid X\in S] \pm 2b(i) }\int_{x\in S}\overline{\cD}(x)dx
                        }\\
                    &= 
                        \sum_{S\in \cS}
                        {
                            \inparen{\Ex[Y^1-Y^0\mid X\in S]\pm 2b(i) }\cdot \frac{\cD(S)}{1-\cD(N)}
                        }\\
                    \\
                    &\in 
                        \inparen{\Ex[Y^1-Y^0\mid X\not\in N]
                        \pm 
                        2\sum_{S\in \cS}
                            { b(i) }\cD(S)
                        }\frac{1}{1-\cD(N)}\\
                    &\in 
                        \inparen{\Ex[Y^1-Y^0\mid X\not\in N]\inparen{1\pm 2\cD(N)}
                        \pm 
                        4\sum_{S\in \cS} { b(i) }\cD(S)
                        }\,.
                        \tag{using that $\cD(N)\leq \frac{1}{2}$}
                \end{align*}
            Therefore, using that $-1\leq Y^0,Y^1\leq 1$
            \[
                \abs{
                    \Ex\insquare{Y^1-Y^0\mid X\not\in N}
                    - \Ex_X\insquare{\psi(X)\mid X\not\in N}
                }
                \leq 
                4\cD(N)+4\sum_{S\in \cS} { b(i) }\cD(S)\,.
                    \yesnum\label{eq:clustering:step1_1}
            \]
            Next, toward the second step, observe that 
            \begin{align*}
                \Ex\insquare{Y^1-Y^0}
                    &= \Ex\insquare{Y^1-Y^0\mid X\not\in N}\Pr[X\not\in N]+\Ex\insquare{Y^1-Y^0\mid X\in N}\cD(N)\\
                    &\in \Ex\insquare{Y^1-Y^0\mid X\not\in N}\Pr[X\not\in N]\pm 2\cD(N)\tag{using that $-1\leq Y^0,Y^1\leq 1$}\\
                    &= \Ex\insquare{Y^1-Y^0\mid X\not\in N}(1-\cD(N)) \pm 2\cD(N)\,.
            \end{align*}
            Using that $\Ex\insquare{Y^1-Y^0\mid X\not\in N}\leq 2$, it follows that
            \begin{align*}
                \abs{\Ex\insquare{Y^1-Y^0} - \Ex\insquare{Y^1-Y^0\mid X\not\in N}}
                \leq 4\cD(N)\,.
                \yesnum\label{eq:clustering:step1_2}
            \end{align*}
            Therefore, \cref{eq:clustering:step1_1,eq:clustering:step1_2} imply that
            \[
                \sqrt{\circled{1}} 
                =
                \abs{
                    \Ex_X\insquare{\psi(X)\mid X\not\in N}
                    - \Ex\insquare{Y^1-Y^0}
                }
                \leq 8\cD(N)+4\sum_{S\in \cS} b(i)\cD(S)\,.
            \]

        \paragraph{Upper bound on $(2)$.}
            Term \circled{2} can be upper bounded as follows.
            \begin{align*}
                \circled{2}
                &=\Ex_X\insquare{
                    \frac{e(X)(v_1(X)+\mu_1(X)^2)}{(e(X; \cS, N))^2} + \frac{(1-e(X))(v_0(X)+\mu_0(X)^2)}{(1-e(X; \cS, N))^2} 
                }\\
                &\leq \Ex_X\insquare{
                    \frac{e(X)}{(e(X; \cS, N))^2} + \frac{1-e(X)}{(1-e(X; \cS, N))^2} 
                }\tag{using that $\abs{Y^0},\abs{Y^1}\leq 1$}\\
                &= \sum_{S\in \cS}\Ex_X\insquare{
                    \frac{e(X)}{(e(X; \cS, N))^2} + \frac{1-e(X)}{(1-e(X; \cS, N))^2} \mid X\in S
                }
                    \Pr[X\in S\mid X\not\in N]\\
                &= \sum_{S\in \cS}\Ex_X\insquare{
                    \frac{1}{e(X; \cS, N)} + \frac{1}{1-e(X; \cS, N)} \mid X\in S
                }
                    \Pr[X\in S\mid X\not\in N]\\
                &= \sum_{S\in \cS}\frac{2}{e(S)(1-e(S))}\frac{\cD(S)}{1-\cD(N)}\,.
            \end{align*}

\section{Proofs of Algorithmic Results}\label{sec:proofof:thm:main}

    \subsection{Proof of Main Algorithmic Result}
        In this section, we prove \cref{thm:main}, our main algorithmic result: %

        \mainTheorem*
        
        \noindent {As mentioned in \cref{sec:algoOverview}, the algorithm divides the given dataset $\dataset$ into two parts: $\dataset_1$ and $\dataset_2$.
        Using the $\dataset_1$, it constructs a (fractional) {\goodlocal{\alpha,\Omega(\beta),\eps}} $\inparen{\cS, N,w}$.
        Then, it outputs $\tau_{\cS, N, w}(\dataset_2; \wh{e})$, i.e., the value of the CIPW estimator $\tau_{\cS, N, w}$ on the second half of the dataset.
        See \cref{sec:cipw} for a discussion on fractional CIPW estimators, how to compute them, and the definition of a fractional good-local partition.}
        
        Before proceeding to this, in the next subsection, we present a few properties of propensity scores that are used to prove correctness.

        \subsubsection{Properties of Propensity Scores}
            The first result shows that inaccurate propensity scores are sufficient to approximately identify outliers with respect to the true propensity scores.
            Given $\beta>0$ and $\wh{e}\colon \R^d\to (0,1)$, let $\outlier{}(\beta; \wh{e})$ be the set of all points $\beta$-outliers with respect to $\wh{e}$, i.e., 
            \[
                \outlier{}(\beta; \wh{e})
                \coloneqq 
                \inbrace{x\in \R^d \colon \wh{e}(x)(1-\wh{e}(x)) \leq \beta}\,.
            \]
            
            \begin{fact}[{Outlier Identification from Inaccurate Propensity Scores}]\label{fact:outliers}
                Let $\eps<\sfrac{\beta}{9}$.
                Fix any two propensity scores $e,\wh{e}\colon \R^d\to (0,1)$ such that $\wh{e}\in B(e,\eps)$.
                It holds that 
                \[
                    \outlier{(\sfrac{\beta}{9}; e)} 
                    \subseteq \outlier{(\sfrac{\beta}{3}; \wh{e})}
                    \subseteq \outlier{(\beta; e)}\,.
                \]
            \end{fact}
            The next fact shows that the coarse propensity scores of a good-local partition are bounded away from 0 and 1.
            \begin{fact}[{Coarse Propensity Score of Good Partition is Bounded}]\label{fact:coarsePropensity}
                For any set $S\subseteq\R^d$ and $c_1,c_2,\beta>0$, 
                \[
                    \text{if}\quad 
                    \cD(S\backslash \outlier{(\sfrac{\beta}{c_1})}) \geq c_2 \cdot \cD\inparen{S}\,,\quad 
                    \text{then}\,,\quad 
                    \frac{c_2}{c_1} \beta
                    \leq e(S) \leq 
                    1 - \frac{c_2}{c_1} \beta\,.
                    \yesnum\label{eq:fact:coarsePropensity:condition}
                \]
            \end{fact}
            Our final fact enables us to infer the propensity scores of outliers and identify outliers from (inaccurate) propensity scores.
            It is used in the proofs of all the facts in this subsection.
            \begin{fact}[{Facts about Propensity Scores}]\label{fact:propensityScores}
                Fix any $0\leq \eps < \beta \leq \frac{1}{4}$.
                Fix any two propensity scores $e, \wh{e}\colon \R^d\to (0,1)$ such that $\wh{e}\in B(e, \eps)$.
                For all $x\in \R^d$, the following hold.
                \begin{enumerate}
                    \item If ${e(x)(1-e(x))}\geq \beta$, then $e(x)\in\insquare{\beta,  1-\beta}$ and else $e(x)\not\in\insquare{2\beta,  1-2\beta}$.
                    \item If ${e(x)(1-e(x))} \geq \beta$, then ${\wh{e}(x)(1-\wh{e}(x))} \geq \frac{3}{4}{(\beta-\eps)}$ and else ${\wh{e}(x)(1-\wh{e}(x))}\geq 2\beta+\eps$.
                    \item By symmetry, if ${\wh{e}(x)(1-\wh{e}(x))} \geq \beta$, then ${{e}(x)(1-{e}(x))}\geq \frac{3}{4}{(\beta-\eps)}$ and else ${{e}(x)(1-{e}(x))}\leq 2\beta+\eps$.
                \end{enumerate}
            \end{fact}
            
            \noindent Next, we present the proofs of all facts in this subsection.

            \bigskip
            
            \begin{proof}\proofof{\cref{fact:outliers}}
                Consider any $x\in \outlier{(\nfrac{\beta}{9}, e)}$.
                \cref{fact:propensityScores}~(2)  and $\eps< \sfrac{\beta}{9}$ imply that
                \[
                    \wh{e}(x)(1-\wh{e}(x)) 
                    \leq \frac{2\beta}{9}+\frac{\beta}{9}
                    =\beta\,.
                \]
                Hence, $x\in \outlier{(\nfrac{\beta}{3}, \wh{e})}$.
                Next, consider any $x\in \outlier{(\nfrac{3}{\beta}; \wh{e})}$.
                \cref{fact:propensityScores}~(3)  and $\eps< \nfrac{\beta}{3}$ imply that
                \[
                    {e(x)(1-e(x))}
                    \leq \frac{2\beta}{3} +\frac{\beta}{3}
                    = \beta\,.
                \]
                Hence, $x\in \outlier{(\beta, e)}$.
            \end{proof}
            \begin{proof}\proofof{\cref{fact:coarsePropensity}}
                Observe that 
                \begin{align*}
                    e(S) 
                    = \frac{\int_{S}e(x)\cD(x)dx}{\cD(S)}
                    &= \frac{
                            \int_{S\backslash\outlier{(\beta/c_1)}}e(x)\cD(x)dx
                            + \int_{S\cap \outlier{(\beta/c_1)}}e(x)\cD(x)dx
                        }{
                            \cD\inparen{S} %
                        }\,.
                \end{align*}
                Further, as $e(x)\geq 0$ for all $x\in \R^d$ and $e(x)\geq {\beta/c_1}$ for all $x\not\in \outlier{(\beta/c_1)}$, it follows that
                \begin{align*}
                    e(S)\geq \frac{\beta}{c_1}\cdot \frac{
                            \cD\inparen{
                                S\backslash\outlier{(\beta/c_1)}
                            }
                        }{
                            \cD\inparen{S} 
                        }
                        \geq \frac{c_2}{c_1}\beta\,.
                \end{align*}
            \end{proof}
            \begin{proof}\proofof{\cref{fact:propensityScores}}
                Fix any $0\leq \eps < \beta \leq \frac{1}{4}$.
                We divide the proof into three parts corresponding. 
                
                \medskip\noindent\textit{Proof of Part 1.} 
                    Since $e(x)\in [0,1]$, $e(x), 1-e(x)\geq e(x)(1-e(x))$ and, hence, if $e(x)(1-e(x))\geq \beta$, then $e(x)\in [\beta, 1-\beta]$.
                    Next, suppose that $e(x)(1-e(x)) < \beta$.
                \begin{itemize}
                    \item If $e(x)\leq \frac{1}{2},$ then $e(x) \leq 2{e(x)(1-e(x))} <2\beta$. 
                    \item Similarly, if $e(x)\geq \frac{1}{2},$ then ${1-e(x)}\leq 2{e(x)(1-e(x))} < 2\beta$.
                \end{itemize}
                Hence, if ${e(x)(1-e(x))} <\beta,$ then $e(x) \not\in \insquare{2\beta, 1-2\beta}$.

                \medskip\noindent\textit{Proof of Part 2.} 
                    Suppose ${e(x)(1-e(x))}\geq \beta$ and, hence, Part 1 implies $e(x)\in \insquare{\beta, 1-\beta}$.
                    Since $\norm{e-\wh{e}}_\infty\leq\eps$, $\wh{e}(x)\in \insquare{\beta-\eps, 1-\beta+\eps}$.
                    Further, since for any $\alpha\in (0,1)$, $\min_{\alpha\leq z\leq 1-\alpha}{z(1-z)}={\alpha(1-\alpha)}$, it follows that
                    \[
                        {\wh{e}(x)(1-\wh{e}(x))}
                        \geq {\inparen{\beta-\eps}\inparen{1-\beta+\eps}}
                        ~~\stackrel{(\beta\geq \eps)}{\geq}~~
                        (\beta-\eps)(1-\beta)\,.
                    \]
                    Since $\beta\leq \nfrac{1}{4}$, %
                    \[
                        {\wh{e}(x)(1-\wh{e}(x))}
                        \geq \frac{3}{4}\inparen{\beta-\eps}
                        \,.
                    \]
                    Next, suppose ${e(x)(1-e(x))} <\beta$.
                    Part 1 implies that $e(x)\not\in \insquare{2\beta, 1-2\beta}$ and, therefore, $\wh{e}(x)\not\in \insquare{2\beta+\eps, 1-2\beta-\eps}$. 
                    This implies that
                    \[
                        {\wh{e}(x)(1-\wh{e}(x))}
                        < 
                        {
                            \inparen{2\beta+\eps}
                            \inparen{1-2\beta-\eps}
                        }
                        \leq 2\beta+\eps\,.
                    \]

                \medskip\noindent\textit{Proof of Part 3.}
                    This follows by swapping $e$ and $\wh{e}$ in Part 2.
            \end{proof}

        \subsubsection{Proof of \cref{thm:main}}\label{sec:algorithm}
        
                We divide the proof the following five steps.
                \begin{itemize}
                    \item \textit{Step 1:} There exists an \goodlocal{\alpha,\beta,0} $(\cS^\star, N^\star)$ satisfying useful properties (\cref{lem:algorithm:existence}).
                    \item \textit{Step 2:} 
                        With high probability, the partition $(\cS, N)$ constructed by \cref{algorithm} is such that 
                        $\cS$ covers all outliers covered by $\cS^\star$ except those of mass $\gamma$, which are covered by $N$ (\cref{lem:algorithm:step2}).
                    \item \textit{Step 3:}
                        With high probability, the sets in $\cS$ are disjoint and, for each $S\in \cS$, $\diam{(S)}\leq 2\alpha$ and $e(S)(1-e(S))\geq \Omega(\beta)$  (\cref{lem:algorithm:step3}). 
                    \item \textit{Step 4:}  
                        It holds that $\sum_{S\in \cS\colon \diam(S)>0} \cD_w(S)\leq 5\rho$, where $\cD_w(S)=\sum_{x\in S} w_S(x)\cdot \cD(x)$ (\cref{lem:algorithm:step4}).
                    \item \textit{Step 5:}
                        The estimator $\tau_{\cS, N, w}$ is asymptotically normal.
                \end{itemize}
                {To see how these results imply \cref{thm:main}, observe that the upper bounds in \cref{eq:robustMSE_upperbound:term1} applied to fractional partitions imply that:
                    \begin{align*}
                        \bias_{\cD}(\tau_{\cS, N, w})
                            &\leq {
                                    4L\sum_{S\in \cS}\cD_w(S)\diam_w\sinparen{S} + 8\cD_w(N)
                                }\,,\\
                        \variance_\cD(\tau_{\cS, N, w})
                            &\leq 
                                    \frac{1}{n}
                                    \sum_{S\in \cS}
                                        \frac{\cD_w(S)}{\wh{e}(S)(1-\wh{e}(S))}
                                        \cdot 
                                        \frac{1}{1 - \cD_w(N)}
                                        \,.
                    \end{align*}
                Where for each set $T$, $\cD_w(T)$ is the weighted mass of $T$, i.e., $\cD_w(T)=\sum_{x\in T} w_T(x)\cdot \cD(x)$, and $\diam_w{(T)}$ is the diameter of the set of all covariates $x$ with a positive mass $w_T(x)>0$.
                Substituting $\cD_w(N)\leq \eps$, and, for each $S\in \cS$, $\diam_{w}(S)\leq 2\alpha$, $\wh{e}(S)(1-\wh{e}(S)) \geq e(S)(1-e(S))-O(\eps)\geq \beta-\eps\geq \sfrac{\beta}{2}$, implies that 
                    \begin{align*}
                        \bias_{\cD}(\tau_{\cS, N, w})
                            &\leq {
                                    8\eps
                                    + 8L\alpha \sum_{S\in \cS\colon \diam_w\sinparen{S}>0} \cD_w\sinparen{S}
                                }
                                \qquadand
                        \variance_\cD(\tau_{\cS, N, w})
                            \leq 
                                    \frac{2}{n\beta(1-\eps)}
                                        \,.
                    \end{align*}
                    Finally, by Step 4, $\sum_{S\in \cS: \diam_w\sinparen{S}>0} \cD_w\sinparen{S}\leq 5\rho$ and, hence the result follows.
                }

            \medskip\noindent\textit{\underline{Step 1 (Existence of good-local partition $(\cS^\star, N^\star)$)}}
                In this step, we prove the following lemma.
                \begin{lemma}\label{lem:algorithm:existence}
                    Suppose \cref{asmp:sparsity,asmp:isolation} hold with constants $\alpha,\beta>0$ and $k\geq 1$.
                    There exists an \goodlocal{\alpha,\beta,0} $(\cS^\star, N^\star)$ that satisfies the following:
                    \begin{enumerate}
                        \item $\cS^\star$ consists of $k$ $L_\infty$ balls $B_1,B_2,\dots,B_k$ of diameter $\alpha$ and singleton sets, and $N^\star$ is empty.
                        \item The balls cover all $\beta$-outliers, i.e., $\outlier{}(\beta; e)\subseteq \bigcup_{1\leq i\leq k} B_i$.
                        \item The centers of any pair of balls $B_i$ and $B_j$ are at least $3\alpha$ apart.
                    \end{enumerate}
                \end{lemma}
                \begin{proof}
                    Let $B_1,B_2,\dots,B_k$ be the $k$ $L_\infty$ balls of diameter $\alpha$ that cover $\outlier{}(\beta)$ as promised in \cref{asmp:isolation}.
                    Let $\cS^\star=\sinbrace{B_1,B_2,\dots,B_k}\cup \sinbrace{\sinbrace{x}\colon x\not \in \bigcup_{1\leq i\leq k} B_i}$ and $N^\star=\emptyset$.
                    The three properties are straightforward to verify.
                    Property 1 is satisfied by the construction of $\cS^\star$ and $N^\star$.
                    Property 2 is satisfied as $B_1,B_2,\dots,B_k$ is guaranteed to cover $\outlier{}(\beta, e)$.
                    Property 3 is satisfied by \cref{asmp:isolation}.

                    It remains to show that $(S^\star,N^\star)$ is an \goodlocal{\alpha,\beta,0}.
                    The only remaining part is to show that for all $S\in \cS^\star$, $e(S)(1-e(S))\geq \beta$.
                    We divide the proof into two cases.
                    The first case is when $S=B_i$ for some $i$.
                    Due to \cref{asmp:sparsity}, we know that $e(S\backslash\outlier{}(\beta, e))\geq \Omega(e(S))$ and, hence \cref{fact:coarsePropensity} implies that $e(S)\in [\Omega(\beta), 1-\Omega(\beta)]$.
                    \mbox{Further, \cref{fact:propensityScores}~(1) implies that $e(S)(1-e(S))\geq\Omega(\beta)$.}
                \end{proof}         
                
            \medskip\noindent\textit{\underline{Step 2 (WHP $\cS$ covers most outliers in $\cS^\star$)}}
                In this step, we prove the following lemma.
                \begin{lemma}\label{lem:algorithm:step2}
                    Suppose \cref{asmp:sparsity,asmp:isolation} hold with constants $\alpha,\beta>0$, $k\geq 1$, and $n=\Omega\inparen{\eps^{-2}\inparen{dk+\log(\nfrac{1}{\delta})}}$.
                    Let $\cB\subseteq \inbrace{B_1,B_2,\dots,B_k}$ be the set of balls that are not covered by any $S\in \cS$:
                    \[
                        \cB\coloneqq 
                            \inbrace{B_i\colon \not\exists S\in \cS\,,~~\st,~~S\supseteq B_i,\ 1\leq i\leq k}\,.
                    \]
                    With probability at least $1-\delta$, 
                    \[
                        \cD\inparen{\bigcup\nolimits_{B\in \cB} B\cap \outlier{}(\nfrac{\beta}{3}, \wh{e})}\leq \eps
                         \quadand  
                        N = \bigcup\nolimits_{B\in \cB} B\cap \outlier{}(\nfrac{\beta}{3}, e) \,.
                    \]
                \end{lemma}
                \begin{proof}
                    It will be sufficient to restrict our attention to only the sets in $\cS$ that were added in the first for-loop and contain at least 1 outlier in $\outlier{}(\nfrac{\beta}{3};\wh{e})$.
                    Let the collection of these sets be $\cS_1$.
                    (In particular, $\cS_1$ includes the $\ell_\infty$-balls $S$ around outliers, but not the collection $\hypo{T}$ of non-outliers in $S$.)
                    Observe that each set $S\in \cS_1$, is defined by a covariate $x$ and contains all points $\alpha$-close to $x$.

                    We claim that, for all $1\leq i\leq k$, if there is any covariate $x\in \outlier{}(\beta, e)\cap B_i$ that appears in $\dataset_1$, then there is a set $S\in \cS_1$ such that $S \supseteq B_i$.
                    To see this, fix any $B_i$ and consider any $x\in \outlier{}(\nfrac{\beta}{3}, \wh{e})\cap B_i$.
                    Since $x\in \outlier{}(\nfrac{\beta}{3}, \wh{e})$, it must be covered by the end of the first for-loop.
                    We consider two cases.
                    \begin{itemize}
                        \item \textbf{Case A (There is a set $S_x\in \cS$):}
                            If there is a set $S_x\in \cS$, then $S_x$ is centered at $x$ and contains all points $\alpha$ close to $x$ and, since $\diam{(B_i)}=\alpha$ and $x\in B_i$, it follows that $S_x$ contains all points in $B_i$.
                        \item \textbf{Case B (There is no set $S_x\in \cS$):}
                            Since $S_x\not\in \cS$ and $x$ was covered at the end of the first for-loop, there must be some other $z\in \outlier{}(\nfrac{\beta}{3}, \wh{e})$ such that $x\in S_z$ and, hence, $\norm{x-z}\leq \alpha$.
                            Since $z\in \outlier{}(\nfrac{\beta}{3}, \wh{e})$, $z\in \outlier{}(\beta, e)$ by \cref{fact:outliers} and, hence $z$ must belong to one of the balls $B_1,B_2,\dots,B_k$.
                            Moreover, since $\norm{x-z}\leq \alpha$ and the centers of any two balls is at least $3\alpha$ apart, it must hold that $z$ belongs to the same ball as $x$, i.e., $z\in B_i$.
                            Thus, we have found a point in $S\cap B_i$ which was covered in the first for-loop, we already checked this case above. %
                    \end{itemize}
                    The above observation implies that $\cS_1$ does not contain any points from $\bigcup\nolimits_{B\in \cB} B\cap \outlier{}(\nfrac{\beta}{3}, \wh{e})$ (even if $S\in \cS_1$ contained one point form $B\cap \outlier{}(\nfrac{\beta}{3}, \wh{e})$, it would contain the entire $B$ and, hence $B$ would not have been in $\cB$).
                    Further, since all points in $\bigcup\nolimits_{B\in \cB} B\cap \outlier{}(\nfrac{\beta}{3}, \wh{e})$ are in $\outlier{}(\nfrac{\beta}{3}, \wh{e})$, they are all included in $N$ in the second for loop and, hence, $N=\bigcup\nolimits_{B\in \cB} B\cap \outlier{}(\nfrac{\beta}{3}, \wh{e})$.

                    It remains to show that the mass of $\bigcup\nolimits_{B\in \cB} B\cap \outlier{}(\nfrac{\beta}{3}, \wh{e})$ is at most $\eps$ with high probability.
                    Let ${\rm BAD}$ be the set of all subsets of $\inbrace{B_1, B_2, \dots, B_k}$ such that $\bigcup_{B\in {\rm BAD}} B$ has mass at least $\eps$.
                    It suffices to show that the following event happens with probability at least $1-\delta$:
                    \[
                        \forall_{\cB\in {\rm BAD}}\,,\quad
                        \inparen{\bigcup\nolimits_{S\in \cS_1} S}
                        \cap 
                        \inparen{\bigcup\nolimits_{B\in \cB} B}
                        = \emptyset\,.
                    \]
                    Fix any $\cB\in {\rm BAD}$.
                    It holds that 
                    \[
                        \Pr\insquare{
                            \inparen{\bigcup\nolimits_{S\in \cS_1} S}
                        \cap 
                        \inparen{\bigcup\nolimits_{B\in \cB} B}
                        }
                        \leq 
                        \inparen{1-\cD{\inparen{\bigcup\nolimits_{B\in \cB} B}}}^{{n/2}}
                        \leq \inparen{1-\eta}^{{n/2}}
                        \,.
                    \]
                    Where we used the fact that $\abs{\dataset_1}=\sfrac{n}{2}.$
                    Taking the union bound over all $\cB\in {\rm BAD}$, it follows that 
                    \[
                        \Pr\insquare{
                            \forall_{\cB\in {\rm BAD}}\,,\quad
                            \inparen{\bigcup\nolimits_{S\in \cS_1} S}
                            \cap 
                            \inparen{\bigcup\nolimits_{B\in \cB} B}
                            = \emptyset
                        }
                        \leq 2^k \inparen{1-\eps}^{{n/2}}
                        <\delta\,.
                    \]
                    Where the final inequality follows because $n=\Omega\inparen{\eps^{-2}\inparen{k+\log(\nfrac{1}{\delta})}}$.
                \end{proof}

            \medskip\noindent\textit{\underline{Step 3 ($\cS$ is disjoint and, all $S\in \cS$, have small diameter and non-extreme propensity scores):}}
                In this step, we prove the following lemma.
                \begin{lemma}\label{lem:algorithm:step3}
                    Suppose \cref{asmp:sparsity,asmp:isolation} hold with constants $\alpha,\beta>0$ and $k\geq 1$.
                    All sets in $(\cS, N)$ are disjoint.
                    Further, for each $S\in \cS$, $\diam{}(\cS)\leq 2\alpha$ and $e(S)(1-e(S))\geq \Omega(\beta)$.
                \end{lemma}
                \begin{proof}
                    First, by construction, it holds that, for each $S\in \cS$, $\diam{}(\cS)\leq 2\alpha$.
                    Next, to see that $e(S)(1-e(S))\geq \Omega(\beta)$, we consider three cases:
                    
                \begin{itemize}
                    \item \textbf{Case A ($S$ was added in first for-loop):}
                        Let $x$ be the covariate that generated $S$.
                        Since $S$ was added in the first for loop, it must be the case that $x\in \outlier{}(\nfrac{\beta}{3}, \wh{e})$ (i.e., ${\wh{e}(x)(1-\wh{e}(x))} <\nfrac{\beta}{3}$) and, hence, by \cref{fact:outliers}, $x\in \outlier{}(\beta, e)$.
                        Therefore, \cref{asmp:sparsity} and the fact that $\diam{(S)}\geq \alpha$, implies that $\cD\inparen{
                                S\backslash \outlier{\inparen{\beta}}
                            }
                            \geq 
                            \Omega(1) \cdot  \cD\inparen{S }$.
                        Now, we further divide into two cases.
                        \begin{enumerate}
                            \item \textbf{Case A.I ($\cD(S\backslash \outlier{}(\nfrac{\beta}{3};\wh{e})\leq \sfrac{\rho}{k}$):}
                                Uniform convergence over the collection of $\ell_p$ balls with respect to distributions (1) $\cD$ and (2) $\cD_w$ implies that: for any $x\not\in \outlier{}(\nfrac{\beta}{3};\wh{e})$
                                \[
                                    w_{S}(x) = \frac{
                                        \cD(S\cap \outlier{}(\nfrac{\beta}{3};\wh{e})) \pm \frac{\rho\eps}{k} + \frac{\rho}{k}
                                    }{
                                        \cD(S\backslash \outlier{}(\nfrac{\beta}{3};\wh{e})) \pm \frac{\rho\eps}{k} + \frac{\rho}{k}
                                    }
                                    \geq \Omega(1)\,.
                                \]
                                Therefore, as $\outlier{}(\beta)\supseteq \outlier{}(\nfrac{\beta}{3};\wh{e})$ and $\cD(S\backslash\outlier{}(\beta))\geq \Omega(1)\cD(S)$, 
                                \[
                                    \cD_w(S\backslash \outlier{}(\beta))
                                    \geq \Omega(1)\cdot \cD(S\backslash \outlier{}(\beta))
                                    \geq \Omega(1)\cD(S)
                                    \geq \Omega(1) \cD_w(S)\,.
                                \]
                            \item \textbf{Case A.II ($\cD(S\backslash \outlier{}(\nfrac{\beta}{3};\wh{e})\leq \sfrac{\rho}{k}$):}
                                In this case, for any $x\not\in\outlier{}(\nfrac{\beta}{3};\wh{e})$
                                \begin{align*}
                                    \cD_w(S\backslash \outlier{}(\nfrac{\beta}{3};\wh{e})) 
                                    &=w_S(x)\cdot \cD(S\backslash \outlier{}(\nfrac{\beta}{3};\wh{e}))\\
                                    &= 
                                    \frac{
                                        \cD(S\cap \outlier{}(\nfrac{\beta}{3};\wh{e})) \pm \frac{\rho\eps}{k} + \frac{\rho}{k}
                                    }{
                                        \cD(S\backslash \outlier{}(\nfrac{\beta}{3};\wh{e})) \pm \frac{\rho\eps}{k} + \frac{\rho}{k}
                                    }
                                    \cdot \cD(S\backslash \outlier{}(\nfrac{\beta}{3};\wh{e}))\\
                                    &\geq \Omega(1)\cdot{\cD(S\cap \outlier{}(\nfrac{\beta}{3};\wh{e}))}\,.
                                \end{align*}
                                Since $\outlier{}(\nfrac{\beta}{9})\subseteq \outlier{}(\nfrac{\beta}{3}; \wh{e})$, it follows that 
                                \begin{align*}
                                    \cD_w(S\backslash \outlier{}(\nfrac{\beta}{9}))
                                    \geq \Omega(1)\cdot \cD(S\cap \outlier{}(\nfrac{\beta}{9}))
                                    \,.
                                \end{align*}
                        \end{enumerate}
                        In both cases, \cref{fact:coarsePropensity} implies that $e(S)\in \insquare{\Omega(\beta), 1-\Omega(\beta)}$ and, hence, by \cref{fact:propensityScores}~(1), $e(S)(1-e(S))\geq \Omega(\beta)$.
                    \item \textbf{Case B ($S$ was added in second for-loop):}
                        Since $S$ was added in the second for-loop, it is of the form $S=\inbrace{x}$ where $x\not\in \outlier{}(\nfrac{\beta}{3}, \wh{e})$ and, hence, by \cref{fact:outliers}, $x\not\in \outlier{}(\nfrac{\beta}{9}, e)$.
                        The claim follows as, since $S$ is a singleton, $e(S)(1-e(S))=e(x)(1-e(x))\geq \nfrac{\beta}{9}$.
                \end{itemize}
                    Finally, we show that all sets in $(\cS, N)$ are disjoint.
                    Since $N$ is composed of points that were not included in any set in $\cS$, it is disjoint by construction.
                    Further, since sets added in the second for-loop are singletons, they are disjoint from each other, and since they only contain points that were not covered in the first for-loop, they are also disjoint with sets added in the first for-loop.
                    It remains to show that no two sets added in the first for-loop are disjoint.

                    Toward this, consider two sets $S_x$ and $S_z$ added in the first for-loop, and generated from covariates $x$ and $z$ respectively.
                    Without loss of generality, suppose $S_x$ was added first and, hence $z\not\in S_x$.
                    Since $S_x$ contains all points $w$, such that, $\norm{w-x}_\infty\leq \alpha$, $z$ is more than $\alpha$ away from $x$, which implies that $x$ and $z$ do not belong to the same ball among $B_1,B_2,\dots,B_k$.
                    Then, as the centers of any pairs of balls $B_i$ and $B_j$ are more than $3\alpha$ apart and the diameter of $B_i$ and $B_j$ is $\alpha$, $x$ and $z$ must in fact be $2\alpha$ apart, which implies that $S_x$ and $S_z$ do not overlap as their radius is $\alpha$.
                \end{proof}

            \noindent \medskip\noindent\textit{\underline{Step 4 (Non-singleton sets in $\cS$ have mass at most $2\rho+k\eps$):}}
                We prove the following lemma.
                \begin{lemma}\label{lem:algorithm:step4}
                    Suppose \cref{asmp:sparsity,asmp:isolation} hold with constants $\alpha,\beta>0$ and $k\geq 1.$
                    It holds that $\sum_{S\in \cS\colon \diam(S)>0} \cD_w(S)\leq 5\rho$, where $\rho=\cD(\outlier{}(\beta,e))$.
                \end{lemma}
                \begin{proof}
                    Since for any set $T$, $\cD_w(T)=\cD_w(T\cap \outlier{}(\beta))+\cD_w(T\backslash \outlier{}(\beta))$, and $\cD(\outlier{}(\beta))\leq \rho$, it suffices to upper bound the sum where $\cD_w(S)$ replaced by $\cD_w(S\backslash\outlier{}(\beta))$ by $\rho+O(k\eps)$.
                    Moreover as $\outlier{}(\beta)\supseteq \outlier{}(\nfrac{\beta}{3};\wh{e})$, it suffices to upper bound the sum where $\cD_w(S\backslash\outlier{}(\beta))$ is replaced by $\cD_w(S\backslash\outlier{}(\nfrac{\beta}{3};\wh{e}))$
                    By using uniform convergence over the distribution of the collection of $\ell_p$ balls with respect to (1) the distribution $\cD$ supported on $\outlier{}(\nfrac{\beta}{3}; \wh{e})$ and (2) distribution $\cD$, it follows that 
                    \[
                        \cD_w\inparen{
                                        S\backslash\outlier{}\inparen{
                                        \nfrac{\beta}{3};\wh{e}
                                        }
                                    }
                                = \cD\inparen{
                                        S\backslash\outlier{}\inparen{
                                        \nfrac{\beta}{3};\wh{e}
                                        }
                                    }
                                    \cdot 
                                    \frac{
                                        \cD(S\cap \outlier{}(\nfrac{\beta}{3};\wh{e})) \pm \frac{\rho\eps}{k} + \frac{\rho}{k}
                                    }{
                                        \cD(S\backslash \outlier{}(\nfrac{\beta}{3};\wh{e})) \pm \frac{\rho\eps}{k} + \frac{\rho}{k}
                                    }\,.
                    \]
                    Next, we consider two cases.
                    \begin{enumerate}
                        \item \textbf{Case A ($\cD(S\backslash\outlier{}(\nfrac{\beta}{3};\wh{e}))\leq \sfrac{\rho}{k}$):}
                            Since $\eps\leq \frac{1}{2}$, it holds that
                            \[
                                \cD\inparen{
                                        S\backslash\outlier{}\inparen{
                                        \nfrac{\beta}{3};\wh{e}
                                        }
                                    }
                                    \cdot 
                                    \frac{
                                        \cD(S\cap \outlier{}(\nfrac{\beta}{3};\wh{e})) \pm \frac{\rho\eps}{k} + \frac{\rho}{k}
                                    }{
                                        \cD(S\backslash \outlier{}(\nfrac{\beta}{3};\wh{e})) \pm \frac{\rho\eps}{k} + \frac{\rho}{k}
                                    }
                                \leq 2\cD(S\cap \outlier{}(\nfrac{\beta}{3};\wh{e}))  + \frac{3\rho}{k}\,.
                            \]
                        \item \textbf{Case B ($\cD(S\backslash\outlier{}(\nfrac{\beta}{3};\wh{e})\geq \sfrac{\rho}{k}$):}
                            Since $\eps\leq \frac{1}{2}$, it holds that
                            \[
                                \cD\inparen{
                                        S\backslash\outlier{}\inparen{
                                        \nfrac{\beta}{3};\wh{e}
                                        }
                                    }
                                    \cdot 
                                    \frac{
                                        \cD(S\cap \outlier{}(\nfrac{\beta}{3};\wh{e})) \pm \frac{\rho\eps}{k} + \frac{\rho}{k}
                                    }{
                                        \cD(S\backslash \outlier{}(\nfrac{\beta}{3};\wh{e})) \pm \frac{\rho\eps}{k} + \frac{\rho}{k}
                                    }
                                \leq \cD(S\cap \outlier{}(\nfrac{\beta}{3};\wh{e}))  + \frac{3\rho}{2k}\,.
                            \]
                    \end{enumerate}
                    Now, since $\outlier{}(\nfrac{\beta}{3}; \wh{e}) \subseteq \outlier{}(\beta)$ and $\abs{\cS}\leq k$, it follows that 
                    \[
                        \sum_{S\in\cS\colon \diam{(S)}>0} \cD_w\inparen{
                                        S\backslash\outlier{}\inparen{
                                        \nfrac{\beta}{3};\wh{e}
                                        }
                                    }
                        \leq \sum_{S\in\cS\colon \diam{(S)}>0} \cD\inparen{S\cap \outlier{}(\beta)} + \frac{3\rho}{k}
                        \leq 2\rho+ 3\rho\,.
                    \] 
                    
                \end{proof}

            \noindent \medskip\noindent\textit{\underline{Step 5 (Asymptotic normality):}}
                In this step, we show that $\tau_{\cS, N, w}$ is asymptotically normal:
                it follows because $\tau_{\cS, N, w}$ is a CIPW estimator on the extended domain (\cref{sec:cipw}) and, for any set $T$, $e(T)$ is the same in the original and extended domains, and, hence by \cref{lem:algorithm:step3}, $e(S)(1-e(S))\geq \sfrac{\beta}{9}$ in the extended domain.
                \cref{thm:asymptoticNormality} (which is applicable as $\eps\leq \beta/10$) implies asymptotic normality.

\subsection{Plug-In Rates for Doubly Robust Estimators}
        \label{sec:alg:plugin}
            Recall that doubly robust estimators use two nuisance parameters, the propensity scores and estimates of the expected outcomes $\mu_0,\mu_1\colon\R^d\to [-1,1]$.
            Since they depend on estimates of $\mu_0,\mu_1\colon\R^d\to [-1,1]$, they do not fit the family of CIPW estimators (\cref{sec:cipw}). 
            However, for a partition $(\cS, N)$, one can consider the following generalization of a doubly robust estimator:
                given propensity scores $e\colon \R^d\to (0,1)$ and conditional means $\mu_0,\mu_1\colon \R^d\to [-1,1]$, define
            \[
                    \tau_{{\rm DR}, \cS, N}{(\dataset; \mu_0, \mu_1, e)}
                    \coloneqq \frac{1}{\abs{\inbrace{i\in [n]\colon x_i\not\in N}}} \sum_{\substack{S\in \cS\\i\colon x_i\in S}}  \inparen{
                        \frac{
                            \inparen{t_i-e(S)}\inparen{y-\mu_1(x_i)}
                        }{e(S)}
                            -
                        \frac{
                            \inparen{e(S)-t_i}\inparen{y-\mu_0(x_i)}
                        }{1-e(S)}
                    }
                \,.
            \]
            This is the doubly robust estimator defined on the domain $\cS$ (where, for each $S\in \cS$, all elements of $S$ are assumed to have the same covariate).
            
            We show that when $(\cS, N)$ is a good-local partition then it has a small robust MSE.
            \begin{proposition}[{Robust MSE of Coarse DR Estimator}]\label{prop:plugInRate}
                Suppose \Cref{asmp:lipschitzness} holds.
                For any $\eps\in [0, \sfrac{\beta}{2}]$ and an \goodlocal{\alpha,\beta,\gamma} $\inparen{\cS, N}$ for some $\alpha,\beta > 0$ and $\gamma\in[0,1/2]$, 
                \[
                    \RMSE_{\cD, \epsilon}{\inparen{\tau_{{\rm DR}, \cS, N}{(\mu_0, \mu_1)}}}
                        \leq O\inparen{\alpha L + \frac{\eps}{\beta} + \gamma  + \frac{1}{\sqrt{\beta n}}} %
                        \,.
                \]
            \end{proposition}
            Thus, combining this result with the efficient algorithm for finding a good-local partition in \cref{thm:main}, we get an efficient method for constructing a coarse doubly robust estimator with a small robust MSE.
            This estimator has a small robust MSE because of the same reason why CIPW estimators defined by a good-local partition have a small robust MSE: the MSE of a (standard) doubly robust estimator is $\propto \Ex\insquare{\frac{1}{e(x)(1-e(x))}}$ while the MSE of the above coarse doubly robust estimator is $\propto \Ex\insquare{\frac{1}{e(S)(1-e(S))}}$.  
            The latter is small as $e(S)(1-e(S))\geq \beta$ for all $S\in \cS$ in a good-local partition.
            Moreover, the bias introduced is also small because of Lipschitzness and the fact that each $S\in \cS$ has a small diameter. 
            The same approach can also be used to construct coarse variants of other estimators: given an estimator $E$ that depends on the propensity scores, evaluate $E$ on the coarse domain $\cS$ with the coarse propensity scores $\inbrace{e(S)\colon S\in \cS}$.
            This should significantly improve the robust MSE of $E$ whenever $\variance{(E)}\propto \Ex\insquare{\frac{1}{e(x)(1-e(x))}}$ or $\propto \max_x \frac{1}{e(x)(1-e(x))}$.

            \medskip 
            
            \begin{proof}\proofsketch{\cref{prop:plugInRate}}
                Consider the distribution $\cD'$ from which we sample as follows:
                    first draw a sample $(x,y,t)\sim \cD$ and, subsequently, let $(x,y-\mu_t(x), t)$ be the sample from $\cD'$.
                That is we \toa{center} the outcome at each covariate $x$.
                Observe that $\tau_{{\rm DR}, \cS, N}{(\dataset; \mu_0, \mu_1, e)}$ is exactly the IPW estimator with respect to $\cD'$ defined by the partition $(\cS, N)$.
                Since $(\cS, N)$ is a good-local partition, the result follows from the upper bound on the $\eps$-Robust RMSE of CIPW estimators arising from good local partitions.
            \end{proof}

\section{Proofs of Results Comparing CIPW to Baselines}
    
    \subsection{Formal Statement and Proof of \cref{infthm:2}}\label{sec:proofof:prop:comparisonToBaselines}
        In this section, we prove \cref{infthm:2}.
        Its formal statement is as follows.
        \begin{theorem} 
        \label{prop:comparisonToBaselines2}
            For any $\eta,\eps>0$ and $n\geq 1$, there is an unconfounded distribution $\cD$ satisfying \cref{asmp:lipschitzness,asmp:sparsity,asmp:isolation} with parameters $(\alpha, \beta, L, k)=(\eta,\sfrac{1}{9},3,3)$ such that given inaccurate propensity scores $\wh{e}$ with $\norm{\wh{e}-e}_\infty\leq \eps$ and $n=\Omega(d/\eps^2)$ samples %
            \begin{itemize}
                \item[$\triangleright$] The RMSE of IPW and doubly robust estimators (which are given correct conditional outcomes $\mu_0$ and $\mu_1$) is $\Omega\left(\sfrac{1}{{\eta}} \right)$; %
                \item[$\triangleright$] The RMSE of $\eps$-Trimmed IPW estimator (which removes all $\eps$-outliers) is $\Omega(1)$; 
                \item[$\triangleright$] The RMSE of the Estimator \cref{infthm:1} is $O(\eps+{\sfrac{1}{\sqrt{n}}})$.
            \end{itemize}
        \end{theorem}

        \begin{proof}\proofof{\cref{prop:comparisonToBaselines2}}
                    Fix a constant $\delta\ll n^{-1}\eta^2$.
                    We will construct an unconfounded distribution $\cD$ over one-dimensional covariates in the interval $[0,1]$.
                    Fix the following parameters 
                    \[
                        \cD_X(x)
                        = \begin{cases}
                            \nfrac{(1-\eps)}{6} & \text{if $x\in \inbrace{0,1-\eps}$},\\
                            \nfrac{(1-\eps)}{3} & \text{if $x\in \inbrace{\eps,1}$},\\
                            \eps & \text{otherwise}
                        \end{cases}
                        \qquad\text{and}\qquad
                        e(x)=\begin{cases}
                            \delta & \text{if } x=\inbrace{\eps,1-\eps},\\
                            \nfrac{1}{2} & \text{otherwise}
                        \end{cases}\,.
                    \]
                    Where $\eps>0$ is the constant in the theorem.
                    For each $x\in [0,1]$, let $\mu_0(x)=v_0(x)=0$.
                    Finally, define 
                    \[
                          \mu_1(x) 
                            = \begin{cases}
                                0 & \text{if } x\leq \nfrac{1}{3},\\
                                3\inparen{x-\nfrac{1}{3}} & \text{if } \nfrac{1}{3}\leq x\leq \nfrac{2}{3},\\
                                1 & \text{otherwise}
                            \end{cases}
                            \qquad\text{and}\qquad
                            v_1(x)=1
                            \,.
                    \] 
                        We can verify the assumptions as follows.
                        \begin{enumerate}
                            \item \textbf{(Lipschitzness)} 
                                $\mu_0$ and $\mu_1$ are 3-Lipschitz Functions
                            \item \textbf{(Isolation)} 
                                With $\beta=\nfrac{1}{9}$, there are two outliers: $\outlier{}(\beta)=\inbrace{\eps,1-\eps}$.
                                These two outliers can be covered by $k=2\leq 3$ balls of radius $\eps$:
                                    $B_1=[0, \eps]$ and $B_2=[1-\eps, 1]$.
                            \item \textbf{(Sparsity)} 
                                Finally, any ball $B$ of radius $r\geq \eps$ containing one of the outliers also contains either $x=0$ or $x=1$ which are not in $\outlier{}(\beta)$ and, hence, guarantees that $\cD(B\backslash\outlier{}(\beta))\geq \frac{1}{3}\cD(B)$.
                        \end{enumerate}

                    \paragraph{Upper Bound on Robust MSE of Estimators.}
                        Since all three assumptions required by \cref{thm:main} hold, substituting the values of $d, k, L, \alpha,$ and $\beta$ in \cref{thm:main}'s guarantee implies that for $n=\Omega(1/\eps^2)$, the estimator $\tau_{A}$ in \cref{infthm:1} satisfies: $\RMSE_{\cD, O(\eps)}{(\tau_{A})}
                            \leq O\inparen{\eps+\inparen{\sfrac{1}{\sqrt{n}}}}.$

                    \paragraph{Lower Bound on MSE of Estimators.}
                        Using standard lower bounds on the variance of IPW and DR estimators (e.g., see \citep{wager2020notes}), implies that 
                        \[
                            \variance_\cD(\ipw{}),
                            \variance_\cD(\doublyrobust{})
                            \geq 
                            \frac{1}{n}\int_{x\in [0,1]} 
                                \frac{v_1(x)+\mu_1(x)^2}{e(x)}
                                d\cD(x)
                            \geq \frac{(1-\eps)}{3n\delta}
                            \geq \frac{1}{\eta^2}
                            \,.
                        \]
                        Since $\tau_{{\rm trim}(\eps)}$ drops the points that are $\eps$-outliers, its expected value is 
                        \[
                            \Ex\insquare{
                                \tau_{{\rm trim}(\eps)}
                            }
                            = \frac{
                                    \frac{1-\eps}{6} \cdot 0
                                    +\frac{1-\eps}{3} \cdot 1
                                    + 
                                    \frac{\eps}{2}
                                }{
                                    1-\frac{1}{2}(1-\eps)
                                }
                            = \frac{2}{3} \pm O(\eps)\,
                            \,.
                        \]
                        Hence, since $\tau=\frac{1}{2}$, $\bias_\cD{(\tau_{{\rm trim}(\eps)})}\geq \frac{1}{6}\pm O(\eps)$.
                        The result now follows by using that, for any estimator $E$, $\RMSE_\cD(E)=\bias_\cD(E)+\sqrt{\variance_\cD(E)}$. %
                \end{proof}

    \subsection{Formal Statement and Proof of \cref{prop:IPW_has_bad_mse}}\label{sec:robustMSEIPW}  
            In this section, we show that small errors in propensity scores can increase the RMSE of IPW and doubly robust estimators (even those that are given accurate conditional means $\mu_0$ and $\mu_1$) by an arbitrary amount.
            Concretely, we prove the following result.
        \begin{proposition}[{Standard Estimators Are Not Robust}]\label{prop:IPW_has_bad_mse}
            For any $\eta\in (0,\sfrac{1}{4}]$ and $n\geq 1$, there is an unconfounded distribution $\cD=\cD_\eta$ such that
            for any $0\leq \eps<\eta$, the $\eps$-Robust RMSE's of \ipw{} and \doublyrobust{($\mu_0,\mu_1$)} (which is given the accurate conditional means $\mu_0$ and $\mu_1$) are
            \[
                \RMSE_{\cD, \eps}\inparen{\ipw{}}= 
                \Theta\inparen{\frac{\eps}{\eta-\eps} + \frac{1}{\sqrt{n\eta}}}
                 \quadand  
                \RMSE_{\cD, \eps}\inparen{\doublyrobust{}(\mu_0,\mu_1)}= 
                \Theta\inparen{\frac{\sqrt{\eta}}{\sqrt{n} (\eta-\eps)}}
                \,.
            \]
            Therefore for $\epsilon = 0$
            \[
                \RMSE_{\cD}\inparen{\ipw{}}, \RMSE_{\cD}\inparen{\doublyrobust{}(\mu_0,\mu_1)}
                =\frac{O(1)}{\sqrt{\eta n}}\,,
                \tag{upper bound on non-Robust RMSEs}
            \]
            and letting $\eps\to\eta$,
            \[
                \RMSE_{\cD, \eps}\inparen{\ipw{}}, \RMSE_{\cD, \eps}\inparen{\doublyrobust{}(\mu_0,\mu_1)}\to \infty\,.  
                \tag{Lower bound on Robust RMSEs}
            \]
        \end{proposition}
        Thus, inaccuracies in propensity scores can arbitrarily increase the RMSE of IPW and, even, doubly robust estimators that are given accurate conditional means $\mu_0$ and $\mu_1$.

        {As mentioned in \cref{sec:examplesIPW}, we focus on the doubly-robust estimator in \cref{sec:examplesIPW}.
        The above result demonstrates that this estimator is fragile to errors in propensity scores.
        Since most guarantees of other doubly robust estimators also crucially rely on the absence of $O(1)$-outliers, we expect other doubly robust estimators to be fragile to errors as well.}

        \label{sec:proofof:prop:IPW_has_bad_mse}
        \medskip        
                \begin{proof}\proofof{\cref{prop:IPW_has_bad_mse}}
                    Consider a distribution supported on two points $x_1$ and $x_2$ which are very close to each other.
                    Construct $\cD=\cD(\eta)$ such that both points have mass $\frac{1}{2}$.
                    Let $x_1$ have the following parameters
                    \[
                        \text{
                            $e(x_1)=\eta$\,,\quad
                            $\mu_1(x_1)=1$\,,\quad
                            $\mu_0(x_1)=0$\,,\quad 
                            $v_1(x_1)=1$\,,\quad and\quad
                            $v_0(x_1)=0$\,.}  
                    \]
                    Let $x_2$ have the same parameters except that its propensity is $\frac{1}{2}$: 
                    \[
                        \text{
                            $e(x_2)=\frac{1}{2}$\,,\quad
                            $\mu_1(x_2)=1$\,,\quad
                            $\mu_0(x_2)=0$\,,\quad 
                            $v_1(x_2)=1$\,,\quad and\quad
                            $v_0(x_2)=0$\,.}  
                    \]
                    Since $\mu_0(x)=v_0(x)=0$ for both $x\in \inbrace{x_1,x_2}$, standard bounds (e.g., see \citet{wager2020notes}) imply
                    \begin{align*}
                        \bias_{\cD}\inparen{\ipw{}(\wh{e})}
                            &=
                            \abs{
                                \frac{e(x_1)\mu_1(x_1)}{2\wh{e}(x_1)}
                                + \frac{e(x_2)\mu_1(x_2)}{2\wh{e}(x_2)}
                                - \frac{\mu_1(x_1)+\mu_1(x_2)}{2}
                            }
                        \,,\\ 
                        \variance_{\cD}\inparen{\ipw{}(\wh{e})}
                            &=
                            \frac{1}{n}
                            \inparen{
                                \frac{e(x_1)v_1(x_1)}{2\wh{e}(x_1)^2}
                                +\frac{e(x_2)v_1(x_2)}{2\wh{e}(x_2)^2}
                            }\,.
                    \end{align*}
                    For the doubly robust estimator, it holds that 
                    \[
                        \bias_{\cD}\inparen{\doublyrobust{}(\mu_0,\mu_1,\wh{e})}
                            =0
                         \quadand  
                        \variance_{\cD}\inparen{\doublyrobust{}(\mu_0,\mu_1,\wh{e})}
                            =
                            \frac{1}{n}
                            \inparen{
                                \frac{e(x_1)v_1(x_1)}{2\wh{e}(x_1)^2}
                                +\frac{e(x_2)v_1(x_2)}{2\wh{e}(x_2)^2}
                            }
                            \,.
                    \]
                    Consider the inaccurate propensity score $\wh{e}(x_1)=\eta-\eps$.
                    We can verify that $\wh{e}\in B(e,\eps)$.
                    For this choice, the above expressions give the following bound
                    \begin{align*}
                        \bias_{\cD}\inparen{\ipw{}(\wh{e})}
                            &= \frac{1}{2} \abs{\frac{\eta}{\eta-\eps}-1}
                            ~~\stackrel{(\eps<\eta)}{=}~~ {\frac{\eps}{2(\eta-\eps)}}\,,\\
                        \variance_{\cD}\inparen{\ipw{}(\wh{e})}
                            &= \frac{\eta}{2n(\eta-\eps)^2}+\frac{1}{4n((\nfrac{1}{2})-\eps)^2}\,,\\
                        \bias_{\cD}\inparen{\doublyrobust{}(\mu_0,\mu_1,\wh{e})}
                            &=0\,,\\
                        \variance_{\cD}\inparen{\doublyrobust{}(\mu_0,\mu_1,\wh{e})}
                            &=
                            \frac{\eta}{2n(\eta-\eps)^2}
                            +
                            \frac{1}{4n((\nfrac{1}{2})-\eps)^2}\,.
                    \end{align*}
                    The result follows as, for any estimator, $\RMSE{(E)}=\bias{(E)}+\sqrt{\variance{(E)}}$.
                \end{proof} 
                \begin{remark}[{CIPW Estimators have a small $\eps$-Robust RMSE in the above example}]\label{rem:prop:IPW_has_bad_mse}
                    Consider the distribution $\cD$ in the proof of \cref{prop:IPW_has_bad_mse}.
                    Define the partition $\cS=\inbrace{\inbrace{x_1,x_2}}$ and $N=\emptyset.$
                    Substituting $\cD$'s parameters in the RMSE expression in \cref{lem:exp_of_mse}, implies that for all $0\leq \eps<\eta$
                    \begin{align*}
                        \bias_\cD{(\tau_{\cS, N}(\wh{e}))}
                        &= \frac{\frac{\eta}{2}+\frac{1}{4}}{\frac{\eta}{2}+\frac{1}{4}\pm \eps}-1 
                        ~\stackrel{(\eps<\eta<\frac{1}{4})}{\leq}~~8\eps\,,\\
                        \variance_\cD{(\tau_{\cS, N}(\wh{e}))}
                        &= \frac{1}{n}
                            \inparen{
                                \frac{\frac{\eta}{2}(1+1)}{\frac{\eta}{2}+\frac{1}{4}\pm \eps}
                                +\frac{\frac{1}{4}(1+1)}{\frac{\eta}{2}+\frac{1}{4}\pm \eps}
                                - 1
                            }
                            ~\stackrel{(\eps<\eta<\frac{1}{4})}{\leq} 
                                \frac{8}{n}
                            \,.
                    \end{align*}
                    Therefore, for all $0\leq \eps<\eta$
                    \[
                        \MSE_{\cD,\eps}{(\tau_{\cS, N})}
                        = 8\eps + \frac{4}{\sqrt{n}}.
                    \]
                \end{remark}

\section*{Acknowledgements}
    We would like to thank the anonymous COLT reviewers for their comments and suggestions.
    AM thanks Colleen Chan, Chris Harshaw, and Shinpei Nakamura-Sakai for helpful discussions and references.

\newpage

\printbibliography

\appendix
\addtocontents{toc}{\protect\setcounter{tocdepth}{2}}

\newpage

\section{Table of Notations}
    \vspace{-1mm}
    
    For ease of reference, we present all the relevant notations below (in \cref{tab:notation}).
    \begin{table}[h!]
        \vspace{-2mm}
        \centering
        \subfigure[\normalsize \textit{Specific Constants}]{  
            \begin{tabular}{p{1.5cm}p{1.5cm}p{12.5cm}}
            \toprule 
            \textbf{Symbol} & \textbf{Range} & \textbf{Meaning}\\
            \midrule{}%
            $n$ & $\N$ &Number of samples provided (i.e., size of the dataset)\\
            $d$ & $\N$ & Dimension of the covariates\\
            $\eps$ & $[0,1]$ & Error in propensity scores\\
            $m$ & $\N$ & Size of the domain in \cref{proofOverivews:hardness}, where it is assumed to be finite\\ 
            \bottomrule{}\\[-5mm]
        \end{tabular}
        }
        
        \subfigure[\normalsize \textit{General Notation}]{  
            \begin{tabular}{p{2cm}p{13.5cm}}
            \toprule
            \textbf{Symbol} & \textbf{Meaning}\\
            \midrule{}%
            $x$ or $X$ & Covariate; domain $\R^d$\\
            $t$ or $T$ & Treatment indicator. Domain $\zo$\\
            $y(t)$ or $Y(t)$ & Outcome when treatment is $T=t$; domain $[-1,1]$\\
            $y$ or $Y$ & Observed outcome. $Y=TY(1)+(1-T)Y(0)$; domain $[-1,1]$\\
            \midrule{}
            $\cD$ & Data distribution. Supported over a subset of $\R^d\times \R\times \zo$ \\
            $\cD_X$ & Marginal of $\cD$ over covariates \\
            $\dataset$ & Censored dataset of size $n$ generated from $\cD$ \\
            \midrule{}
            $B_p(x,r)$ & $\ell_p$-ball of radius $r$ centered at $x$. Subscript is omitted when implied or irrelevant.\\
            $\outlier{}(\beta;\wh{e})$ & The set of all $\beta$-outliers with respect to $\wh{e}$, i.e., $\inbrace{x\in \R^d\colon \wh{e}(x)(1-\wh{e}(x)<\beta}$,\\
            &where the dependence on $\wh{e}$ is hidden when $\wh{e}=e$\\
            \bottomrule{}\\[-5mm]
        \end{tabular}
        }
        
        \subfigure[\normalsize \textit{Parameters of Data Distribution $\cD$}]{  
            \begin{tabular}{p{1.3cm}p{14.2cm}}
            \toprule 
            \textbf{Symbol} & \textbf{Meaning}\\
            \midrule{}%
            $\mu_T(x)$ & Expected value of $Y$ conditioned on $X=x$ and $T=t$\\
            $v_T(x)$ & Variance of $Y$ conditioned on $X=x$ and $T=t$\\
            $e(x)$ & Propensity score at $x$: $e(x)=\Pr[T=1\mid X=x]$\\
            $\cD(x)$ & Probability mass at $X{=}x$. When the domain is continuous, it is the density at $X{=}x$ \\
            $e(S)$ & Average propensity over the subset $S$\\
            $\cD(S)$ & Total mass assigned to subset $S$\\
            \bottomrule{}\\[-5mm]
        \end{tabular}
        }

        \subfigure[\normalsize \textit{Quantities Related to an Estimator $E$}]{  
            \begin{tabular}{p{5cm}p{10.5cm}}
            \toprule 
            \textbf{Symbol} & \textbf{Meaning}\\
            \midrule{}%
            $\bias{(E(\dataset;\nu))}$, $\variance{(E(\dataset;\nu))}$, $\RMSE{(E(\dataset;\nu))}$, $\MSE{(E(\dataset;\nu))}$ & Respectively the bias, variance, RMSE, and MSE of estimator $E$ when evaluated on censored dataset $\dataset$ with nuisance parameters $\nu$\\
            \midrule{}
            $\bias_\cD{(E(\nu))}$, $\variance_\cD{(E(\dataset;\nu))}$,  $\RMSE_\cD{(E(\nu))}$, $\MSE_\cD{(E(\nu))}$ & Respectively the expected values of $\bias{(E(\dataset;\nu))}$, $\variance{(E(\dataset;\nu))}$, $\RMSE{(E(\dataset;\nu))}$, and $\MSE{(E(\dataset;\nu))}$ over censored data $\dataset$ of size $n$ generated from $\cD$\\            
            \midrule{}
            $\RMSE_{\cD, \eps}{(E(\nu))},\; \MSE_{\cD, \eps}{(E(\nu))}$ & Defined as $\max_{\wh{e}\in B(e, \eps)} \RMSE_\cD{(E(\nu, \wh{e}))}$ and $\max_{\wh{e}\in B(e, \eps)}$ $\MSE_\cD{(E(\nu, \wh{e}))}$ respectively \\
            \bottomrule{}\\[-5mm]
        \end{tabular}
        }
        
        \caption{ 
            \normalsize Table of notations.
        }
        \vspace{-6mm}
        \label{tab:notation}
    \end{table}

\newpage
\section{Examples of Propensity-Score-Based Estimators}\label{sec:examplesIPW}
    In this section, we present examples of propensity-score-based estimators.
    Recall that an estimator $E$ is a function that takes a censored dataset $\dataset =\inbrace{\inparen{x_i, y_i, t_i}\colon 1\leq i\leq n}$ as input and outputs a scaler value $E(\dataset)$.
    When $E$ gets some additional or \textit{nuisance} parameters $\mu$ (e.g., propensity scores), we overload the notation to $E(\dataset;\mu).$
    
    \paragraph{Neyman Estimator.} 
        The simplest estimator, called the Neyman Estimator, goes back to \citet{neyman1934two}: %
        \[
        \neyman{}(\dataset) = 
    {\frac{\sum_{i\colon t_i=1}{y_i}}{\sum_{i\colon t_i=1}1} }- {\frac{\sum_{i\colon t_i=0}{y_i}}{\sum_{i\colon t_i=0}1}}\,,
    \tag{Neyman Estimator}
    \] 
    which computes the difference in the average outcomes on the treatment group and control groups.
    This estimator does not rely on any nuisance parameters. 
        
        \paragraph{IPW Estimator.} %
        The most popular estimator that uses a nuisance parameter is the Inverse Propensity-Score Weighted (IPW) Estimator \cite{horvitz1952generalization}. %
        Given propensity scores $e \colon \R^d\to (0,1)$, the IPW estimator is defined as 
        \[
            \ipw{}\inparen{\dataset; e}
            = 
            \frac{1}{n}
            \sum_i
                \inparen{
                    \frac{t_i y_i}{e(x_i)} 
                    - \frac{(1 - t_i) y_i}{1-e(x_i)}
            }
            \,,
            \tag{IPW Estimator}
        \]
        where $n=\abs{\dataset}$.

    \paragraph{{$\eta$-Trimmed IPW Estimators.}}  
        This is a variant of the IPW estimator which is parameterized by an additional number, say, $\eta\in (0,1)$ \citep{imbens2015causal}.
        For any dataset $\dataset$, let $\dataset(\eta)$ be the subset of $\dataset$ that does not contain any sample $(x_i,y_i,t_i)$ where $e(x_i)\not\in [\eta, 1-\eta]$.
        The $\eta$-Trimmed IPW estimator is the IPW estimator on $\dataset(\eta)$, i.e., 
        \[
            \tau_{\rm trim}(\dataset; e, \eta)
            = 
            \ipw{}(\dataset(\eta); e)\,.
            \tag{Trimmed IPW Estimator}
        \]

    \paragraph{{Balancing Score Estimator.}}
        The Balancing Score Estimator is a generalization of the IPW estimator \citep{imbens2015causal}.
        A balancing score $b\colon \R^d\to \R$ is any function such that the treatment ($T$) is independent of the covariates ($X$) conditioned on the value of the balancing score ($b(X)$).
        Fix any balancing score $b$.
        For any $z\in \R$, the $b$-propensity score at $z$, $e_b(z)$, is defined as $e_b(z)=\Pr_\cD[T=1\mid b(X)=z]$.
        The balancing score estimator is defined as follows:
            given a dataset $\dataset$, a balancing score $b\colon \R^d\to (0,1)$, and a $b$-propensity score $e_b$, 
        \[
            \tau_{\rm Balancing}(\dataset; e_b, b) = 
            \frac{1}{n}
            \sum_i
                \inparen{
                    \frac{t_i y_i}{e_b(b(x_i))} 
                    - \frac{(1 - t_i) y_i}{1-e_b(b(x_i))}
            }
            \,.
            \tag{Balancing Score Estimator}
        \]
        In other words, the $b$-balancing score estimator is the IPW estimator over the covariates $\inbrace{b(x)\colon x\in \R^d}$.

    \paragraph{{Blocking Estimators.}}
        A blocking estimator is defined by a partition of the interval $[0,1]$ into blocks $\hypo{B}=\inbrace{B_1,B_2,\dots,B_s}$.
        For each $1\leq i\leq s$, let $e(B_i)=\Pr_\cD[T=1\mid e(X)\in B_i]$.
        Given a set of blocks $\hypo{B}$ and values $e(\hypo{B})=\inbrace{e(B)\mid B\in \hypo{B}}$, the blocking estimator is 
        \[
            \tau_{\rm Blocking}(\dataset; \hypo{B}, e(\hypo{B})) = 
            \frac{1}{n}
            \sum_i
                \inparen{
                    \frac{t_i y_i}{e(B(x_i))} 
                    - \frac{(1 - t_i) y_i}{1-e(B(x_i))}
            }
            \,,
            \tag{Blocking Estimator}
        \]
        where $B(x_i)$ is the unique set $B\in \hypo{B}$ containing $x_i$.
        In other words, the blocking estimator is the IPW estimator over the \toa{covariates} $\hypo{B}$.

    \paragraph{Doubly Robust Estimators.}   
        The family of doubly robust estimators is a significant extension of the IPW estimators which, in addition to propensity scores $e$, also take (estimates of) the conditional means $\mu_0$ and $\mu_1$ respectively as input \citep{robins2005doublyRobust,chernozhukov2018doubleML,foster2023orthognalSL}.
        {All doubly robust estimators enjoy the property that they are unbiased estimates of $\tau$ whenever either of the two nuisance parameters (the propensity scores or the mean outcomes) are correct \citep{robins2005doublyRobust}. 
        More recent doubly robust estimators come with stronger guarantees: under the assumptions that there are no $O(1)$-outliers (i.e., all propensity scores lie in the interval $[\Omega(1), 1-\Omega(1)]$),  their RMSE scales with, roughly, the product of the RMSE of the given propensity score and the given estimates of conditional-outcome-means \citep{chernozhukov2018doubleML,foster2023orthognalSL}.}

        {In this work we focus on the simplest doubly-robust estimator and demonstrate that it is fragile to errors in propensity scores.
        Since most guarantees doubly robust estimators crucially rely on the absence of $O(1)$-outliers, we expect them to be fragile to errors as well.}
        Concretely, we focus on the following doubly robust estimator:
            given conditional means $\mu_0, \mu_1\colon \R^d\to [-1,1]$, and propensity scores $e\colon \R^d\to (0,1)$
            the doubly robust estimator is 
        \[
            \doublyrobust{(\dataset; \mu_0, \mu_1, e)}
            = \frac{1}{n} \sum_i \inparen{
                \frac{
                    \inparen{t_i-e(x_i)}\inparen{y-\mu_1(x_i)}
                }{e(x_i)}
                    -
                \frac{
                    \inparen{e(x_i)-t_i}\inparen{y-\mu_0(x_i)}
                }{1-e(x_i)}
            }\,.
            \tag{Doubly Robust Estimator}
        \]

\section{Further Discussion on the Need for Data-Dependence}\label{sec:data-dependent}
    In this section, we further explore the need to consider data-dependent estimators.
    Recall the following result from \cref{sec:intro}, whose proof appears later in this section. %

    \needDataDependence*

    \noindent It shows that not only is it impossible to weakly beat the RMSE of the IPW estimator, but any \CIPW{} estimator different from IPW has at least a \textit{constant} RMSE for some unconfounded $\cD$ irrespective of the number of samples $n$.
    In particular, it shows that unless we use a data-dependent CIPW estimator, there is no hope of achieving a smaller RMSE than the IPW estimator.
    Ideally, for each unconfounded distribution $\cD$ satisfying some mild properties, we want to find a corresponding CIPW estimator that has a smaller Robust RMSE than the IPW estimator.
    A natural preliminary question is:
        is there any unconfounded distribution $\cD$ for which \textit{some} CIPW estimator has a smaller Robust RMSE than the IPW estimator?
    Our next result answers this in the affirmative. %
    \begin{restatable}[{CIPW Estimators can have much smaller RMSE than IPW}]{lemma}{sanitycheck}\label{lem:sanitycheck}
        For any $\eta,\epsilon >0$, there is a distribution $\cD$ satisfying unconfoundedness (Equation~\eqref{eq:unconfounded}) and a partition $(\cS,N)$ such that 
        \[
            \RMSE_{\cD,\epsilon}(\tau_{\cS, N})
            \leq 
            O(\eta)\cdot \RMSE_{\cD, \epsilon }(\ipw{})\,.
        \]
    \end{restatable}
    In fact, in \cref{lem:sanitycheck}, even the non-robust RMSE of the CIPW estimator is significantly better than that of the IPW estimator. In particular, one can show that, by fixing $\epsilon=0$ in the above result,  $\RMSE_{\cD}(\tau_{\cS, N})=O\inparen{\sfrac{1}{\sqrt{n}}}$ and $\RMSE_{\cD}(\ipw{})=\eta^{-1} \times O\inparen{\sfrac{1}{\sqrt{n}}}$.
    
    More fundamentally, \cref{lem:sanitycheck} suggests that it may be possible to find data-dependent CIPW estimators that have a significantly smaller confidence interval than the IPW estimator; a possibility that we confirm for a large class of unconfounded distributions in \cref{infthm:1}.

    \subsection*{Proof of \cref{lem:sanitycheck}}
    \label{sec:proofof:lem:sanitycheck}
            Next, we prove \cref{lem:sanitycheck}.

            \medskip

            \begin{proof}\proofof{\cref{lem:sanitycheck}}
            Construct a distribution $\cD$ supported on two covariates $\inbrace{x_1,x_2}$.
            Let 
            \[
                \cD(x_1)=\cD(x_2)=\frac{1}{2}\,.
            \]
            Next, for each $x_i$, set $\mu_0(x_i)=v_0(x_i)=0$.
            Further, set 
            \[
                \mu_1(x_1)=\mu_1(x_2) =v_1(x_1)=v_1(x_2) = \frac{1}{2}\,.
            \]
            Finally, set the following propensity scores.
            \[
                \text{$e(x_1)=\eps$\quad and \quad $e(x_2)=\frac{1}{2}$}\,.
            \]
            For each $1\leq i\leq 3$, let
            \[
                \Ex\insquare{Y^1\mid X=x_i}=\mu_i.
            \]
            Using the expression of RMSE (e.g., in \cref{lem:exp_of_mse}), it follows that 
            \[
                \RMSE{(\ipw{})}\geq \var{(\ipw{})}^2 
                = 
                    \sum_{i\in [2]}\frac{\cD(x_i)}{n}\frac{v_1(x_i)+\mu_1(x_i)^2}{e(x_i)(1-e(x_i))}
                = \Omega\inparen{\frac{1}{\eps{n}}}\,.
            \]
            Select the partition $\cS=\inbrace{\inbrace{1,2}}$ and $N=\emptyset.$
            The expression of RMSE in \cref{lem:exp_of_mse} implies that
            \[
                \RMSE{(\tau_{\cS, N})}^2
                = \var{(\ipw{})} 
                = \frac{1}{n(1-\cD(N))}\sum_{S\in \cS}
                    \frac{\sum_{x\in S}\cD(x)\inparen{v_1(x)+\mu_1(x)^2}}{e(S)(1-e(S))}
                \leq \frac{O(1)}{n}
                \,,
            \]
            where the first equality used the fact that $\tau_{\cS, N}$ is unbiased in this example as $\mu_1=\mu_2$. %
        \end{proof}

    \subsection*{Proof of \cref{lem:need_data_dependence}} %
            In the remainder of this section, we prove the following result.
            \needDataDependence*

            \begin{proof}\proofof{\cref{lem:need_data_dependence}}
                Since $(\cS, N)$ is nontrivial, one of the following two cases holds.
                \begin{enumerate}
                    \item \textbf{Case A ($\exists~ S\in \cS$ with $\abs{S}\geq 2$):}
                        Without loss of generality, let $S\supseteq \inbrace{x_1, x_2,\dots}$.
                        Define a data distribution $\cD$ supported on $\inbrace{x_1, x_2}$ with the following parameters 
                        \[
                            \cD(x_1)=\cD(x_2)=\frac{1}{2}\,,\quad 
                            e(x_1)=\frac{7}{8}\,,\quad 
                            e(x_2)=\frac{1}{8}\,,\quad 
                            \mu_1(x_1)=1\,, \quadand  
                            \mu_1(x_2)=0\,.
                        \]
                        Further, set $\mu_0(x_i)=v_0(x_i)$ for all $i$.
                        In this example:
                        \[
                            \tau=\sum_{i\in [2]}\cD(x_i)\mu_1(x_i) = \frac{1}{2}
                             \quadand  
                            \Ex\insquare{\tau_{\cS, N}}
                                =
                                    \frac{
                                        \sum_{i\in [2]} e(x_i)\cD(x_i)\mu_1(x_i)
                                    }{
                                        \sum_{i\in [2]} e(x_i)\cD(x_i)
                                    }
                                =\frac{7}{8}
                                \,.
                        \]
                        Therefore, it follows that $\bias_\cD{}(\tau_{\cS, N})\geq \frac{3}{8}$.
                        Hence
                        \[
                            \RMSE_\cD{(\tau_{\cS, N})} > \frac{3}{8}\geq \frac{1}{3}\,.
                        \]
                        Moreover, as all propensity scores are bounded away from $0$ and $1$, standard bounds on the RMSE of the IPW estimator imply that 
                        $\RMSE_\cD{(\ipw{})}=\frac{O(1)}{\sqrt{n}}$.

                    \item \textbf{Case B ($N\neq \emptyset$):}
                        If $N$ contains all covariates, then, in the example from the previous case, $\tau_{\cS, N}=0$ and, hence, $\bias_\cD{(\tau_{\cS, N})}=\frac{1}{2}$.
                        This implies the result as $\RMSE_\cD{(\tau_{\cS, N})}>\frac{1}{3}$.
                        
                        Hence, we assume that $N$ does not contain all covariates.
                        Without loss of generality, let $x_1\in N$ and $x_2\not\in N$.
                        Consider the example in the previous case.
                        It follows that $\Ex\insquare{\tau_{\cS, N}}=1$ and, since $\tau=\frac{1}{2}$, $\bias_\cD{(\tau_{\cS, N})}=\frac{1}{2}$.
                        This implies the result as $\RMSE_\cD{(\tau_{\cS, N})}>\frac{1}{3}$.
                \end{enumerate}

            \end{proof}

\section{Formal Statement and Proof of Lower Bounds on RMSE}\label{sec:infoLB}
    In this section, we prove the following result.

    \begin{theorem}[{Lower Bound on RMSE}]\label{thm:infoLB:alphaL} 
        For any $n\geq 1$ and $\alpha,\beta,L>0$, there is an unconfounded distribution $\cD$ that satisfies \cref{asmp:lipschitzness,asmp:sparsity,asmp:isolation} with parameters $\alpha,\beta,L$, and $k=1$, such that, there is no algorithm that, given 
                    a dataset $\dataset$ of size $n$ generated from $\cD$, and 
                    (accurate) propensity scores $e\colon \R^d\to (0,1)$,
                outputs $\tau'$ such that $\Pr\insquare{\abs{\tau-\tau'}\leq \rho\cdot\min\inbrace{1,\alpha L}}\geq \frac{1}{2}$. 
    \end{theorem}
    Observe that without \cref{asmp:lipschitzness} we can set $L=\infty$ and, hence, we get the following corollary.
    \begin{corollary}[{Lower Bound on RMSE}]\label{coro:infoLB:rho}
        For any $n\geq 1$ and $\alpha,\beta,L>0$, there is an unconfounded distribution $\cD$, such that, there is no algorithm that, given 
                    a dataset $\dataset$ of size $n$ generated from $\cD$, and 
                    (accurate) propensity scores $e\colon \R^d\to (0,1)$,
                outputs $\tau'$ such that $\Pr\insquare{\abs{\tau-\tau'}\leq \rho}\geq \frac{1}{2}$. 
    \end{corollary}
    \begin{proof}\proofof{\cref{thm:infoLB:alphaL}}
        Fix any $n\geq 1$ and $\alpha,\beta,L>0$.
        Let $0<\eta\ll n^{-1}$.
        We will construct an unconfounded distribution $\cD$ supported on two points $x_1$ and $x_2$ such that $\norm{x_1-x_2}_p=\alpha$.
        Fix the following parameters of $\cD$
        \[
            e(x_1)=\eta\,,\quad 
            e(x_2)=\frac{1}{2}\,,\quad 
            \mu_1(x_2)=0\,,\quad
            \mu_0(x_1)=0\,,\quadand
            \mu_0(x_2)=0\,.
        \]
        Further let 
        \[
            \cD(x_1)=\rho\quadand\cD(x_2)=1-\rho\,.
        \]
        Since $\mu_1$ is $L$-Lipschitz, we have the following bounds 
        \[
            \mu_1(x_1)\in \insquare{
                    -\alpha L, \alpha L
            } \cap [-1,1]\,.
            \yesnum\label{eq:boundOnMU1}
        \]
        Therefore, 
        \[
            \tau \in [-\rho\cdot\min\inbrace{1,\alpha L}, \rho\cdot\min\inbrace{1,\alpha L}]\,.
        \]
        Moreover, depending on the choice of $\mu_1(x_1)$, $\tau$ can take any value in this range and, hence, without gaining further information about $\mu_1(x_1)$, it is information-theoretically impossible to output any estimate $\tau'$ closer than $\rho\cdot\min\inbrace{1,\alpha L}$ to $\tau$.

        Since $\eta\ll n^{-1}$, with probability at least $\frac{1}{2}$, $\dataset$ does not contain any sample $(x,y,t)$ with covariate $x=x_2$ and treatment $t=1$ and, hence, one does not gain any additional information about $\mu_1(x_1)$ after observing these samples and, hence, by the previous argument it remains information theoretically impossible to estimate $\tau$ within a distance of $\rho\cdot \min\inbrace{1,\alpha L}$.
    \end{proof}

\end{document}